\documentclass[11pt]{article}
\def\draft{1}
\usepackage{macros}

\title{Marton's conjecture in polynomial time}
\author{ 
Srinivasan Arunachalam \thanks{\href{mailto:Srinivasan.Arunachalam@ibm.com}{srinivasan.arunachalam@ibm.com}}
\\ \small IBM Research
\and
Arkopal Dutt \thanks{\href{mailto:arkopal@ibm.com}{arkopal@ibm.com} } 
\\ \small IBM Research  
\and 
Sabee Grewal \thanks{\href{mailto:sabee@ibm.com}{sabee@ibm.com} }
\\ \small IBM Research
\\ \small Columbia University
\and
Aparna Gupte \thanks{\href{mailto:agupte@mit.edu}{agupte@mit.edu.} This work was done while the author was an intern at IBM Research.}
\\ \small MIT  }

\date{}

\begin{document}

\maketitle

\begin{abstract}

Gowers, Green, Manners, and Tao (Annals '25) recently resolved Marton's polynomial Freiman--Ruzsa conjecture. 
We give an algorithmic counterpart to their result: 
given uniform sampling and membership-oracle access to a set $A \subseteq \F_2^n$ with doubling constant at most $K$, our algorithm outputs a subspace of size at most $|A|$ whose $K^{O(1)}$ translates cover $A$.
The algorithm runs in $\poly(n,K)$ time. 
As applications, we obtain polynomial-time algorithms for a variety of learning problems, including quadratic Goldreich--Levin, improper agnostic tomography of stabilizer states, and tomography of quantum states with bounded stabilizer extent. 
\end{abstract}

\newpage 
\setcounter{tocdepth}{2}
{\small \tableofcontents}

\newpage

\section{Introduction}\label{sec:intro}
In a recent breakthrough, Gowers, Green, Manners, and Tao~\cite{ggmt2025conjecture,ggmt2024torsion} proved Marton's conjecture, also known as the polynomial Freiman--Ruzsa (PFR) conjecture, first over $\F_2^n$ and subsequently over all abelian groups of bounded torsion.
The theorem characterizes, up to polynomial losses, the algebraic structure of sets $A$ with small \emph{doubling constant}.
For a nonempty set $A \subseteq \F_2^n$, define its \emph{doubling constant} as
\[
    K(A):=\frac{|A+A|}{|A|}=\frac{|\{x+y:x,y\in A\}|}{\{x:x\in A\}}.
\]
If $A$ is a subspace, then $K(A) = 1$. 
More generally, PFR\footnote{Ruzsa's bounded-torsion analog of Freiman's theorem and subsequence improvements~\cite{ruzsa1999analog,green2009freiman} already showed that a set with doubling constant $K$ is contained in a single subspace whose size is exponential in $K$ times $|A|$.
Marton~\cite{ruzsa1999analog} conjectured a substantially sharper description: rather than embedding $A$ into one exponentially larger subspace, one should be able to cover $A$ by \emph{polynomially} many translates of a subspace no larger than $A$ itself. This became known as the \emph{polynomial Freiman--Ruzsa conjecture} and was resolved, nearly three decades later, by Gowers, Green, Manners, and Tao~\cite{ggmt2025conjecture}.} asserts that a set with small doubling constant must resemble a subspace. Over characteristic $2$, which is the version relevant to this work, the theorem is as~follows:

\begin{theorem}[Polynomial Freiman--Ruzsa (PFR) theorem~\cite{ggmt2025conjecture}]
\label{thm:combinatorial_pfr}
There is a polynomial $P \colon \RR_+ \to \RR_+$ such that, for every $n\geq 1$, $K\geq 1$, and nonempty set $A\subseteq \F_2^n$ satisfying $|A+A|\leq K|A|$, there is a subspace $V\leq \F_2^n$ with $|V|\leq |A|$ such that $A$ is covered by at most $P (K)$ translates of $V$.
\end{theorem}

The theorem formalizes the principle that small additive growth forces algebraic structure. If $K(A)$ is small, then $A$ cannot behave like an arbitrary subset of $\F_2^n$; rather, it must be covered by only $K^{O(1)}$ translates of a subspace $V$ whose size is at most $|A|$.  We refer to such a subspace as a \emph{PFR subspace} for $A$. Thus, sets with small doubling exhibit subspace-like structure, with the extremal example being a subspace itself. 

PFR and related additive-combinatorial structure theorems have played an important role throughout theoretical computer science, with applications to property testing~\cite{samorodnitsky2007low}, pseudorandomness~\cite{zewi2011affine}, communication complexity~\cite{ben2014additive}, coding theory~\cite{bhowmick2013new,aggarwal2014non}, worst-case-to-average-case reductions~\cite{asadi2022worst,asadi2024quantum}, sparsification~\cite{bedert2025strong}, and quantum computing~\cite{ad2024tolerant,arunachalam2026tomography,bao2025tolerant,mehraban2024improved,hinsche2025clifford,SCforstabilizers}. In these applications, PFR is used as a structural theorem: it guarantees the existence of a subspace with the desired additive structure.

From an algorithmic perspective, the PFR theorem leaves open a natural question. 
The theorem guarantees that a small-doubling set is governed by a low-complexity algebraic object---a single subspace and polynomially many
of its translates---but the statement is purely \emph{existential}.  In the settings we consider, the set $A\subseteq\F_2^n$ may be exponentially large and therefore cannot be given explicitly. 
We instead assume two natural forms of oracle access: uniform samples from $A$ and a membership oracle that, on input $x \in \F_2^n$, determines whether $x \in A$.
Such implicit access models are standard through theoretical computer science when the underlying object is too large to enumerate: sampling reveals typical elements of $A$, while membership queries allow the algorithm to probe points of its choice. Thus, the model gives succinct access to $A$
 without assuming an explicit representation of the set.
This leads to the algorithmic analogue of the PFR problem:
\begin{center}
    \emph{Given uniform sample access as well as membership oracle access to a set $A \subseteq \FF_2^n$ \\
    with doubling constant at most $K$, can we find a PFR subspace for $A$ in time $\poly(n, K)$?}
\end{center}

\subsection{Main result}
Our main result answers this question in the affirmative. 

\begin{theorem}[Algorithmic PFR]
\label{thm:main}
Let $0<\delta<1$, and let $A\subseteq \F_2^n$ be nonempty and satisfy $|A+A|\leq K|A|$. Given uniform sampling and membership-oracle access to $A$, there is a randomized algorithm that, with probability at least $1-\delta$, outputs a basis for a subspace $V\leq \F_2^n$ such that $|V| \le |A|$ and $A$ can be covered by at most $K^{O(1)}$ translates of $V$.
The algorithm uses $\poly(n,K,\log(1/\delta))$ time, samples, and membership queries.
\end{theorem}

\paragraph{Prior work.}
The algorithmic study of PFR was initiated recently in~\cite{arunachalam2026classical,castro2026algorithmic}, motivated in part by its role in quantum learning~\cite{ad2024tolerant,SCforstabilizers}. These works gave classical and quantum algorithms for recovering PFR structure with running time $\poly(n,2^K)$. 
The underlying ideas arose through a somewhat indirect route: from a quantum algorithm for agnostically learning stabilizer states~\cite{chen2025stabilizer}, through its subsequent dequantization~\cite{briet2026near}, and ultimately to classical algorithms for recovering additive structure~\cite{arunachalam2026classical}. 
In additive combinatorics, often $K$ is viewed as being ``small" (say constant or logarithmic) in terms of the ambient dimension, in which case their algorithm is efficient. The exponential dependence on $K$, however, is a significant limitation when the doubling constant grows with the problem size; which happens to be the case for several applications of interest, including those considered here, $K$ may itself be polynomially related to the underlying parameters. This work improves the dependence on $K$ from exponential to polynomial.

\paragraph{Our approach.} 
Conceptually, our approach differs substantially from the previous algorithmic PFR results.
Rather than reaching PFR through quantum learning-theoretic machinery, we return directly to the GGMT proof and show how to turn its entropic argument into an efficient algorithm.  Several new ideas are necessary for algorithmizing the existential proof. The GGMT proof follows an iterative argument that starts out with the uniform distribution over the set $A$, and repeatedly transforms the distribution in order to eventually be close to the uniform distribution on a PFR subspace for $A$. This is done by making choices and comparing quantities appear difficult to compute efficiently, and the intermediate distributions may have exponentially large support and therefore cannot be held explicitly.
Moreover, even if these distributions can be represented implicitly, one must ensure that the cost of accessing them does not grow prohibitively over the course of the recursion.

The main technical contribution of this work is to overcome these obstacles while preserving polynomial dependence on both $n$
and $K$. 
\cref{sec:tecoverview} gives a self-contained overview of the GGMT argument and explains how each ingredient is implemented algorithmically.

\subsection{Applications}

Beyond giving an efficient algorithmic form of PFR, our result can be used as a general-purpose primitive for recovering additive structure inside larger
algorithms. We illustrate this through two groups of applications. First, on the classical side, we obtain a polynomial-time quadratic Goldreich--Levin algorithm and, consequently, related learning problems. 
Second, we prove a robust variant of algorithmic PFR and use it to obtain polynomial-time
algorithms for several problems in quantum learning,
including agnostic tomography of stabilizer states, estimating stabilizer fidelity, and tomography of states with bounded stabilizer extent.

\subsubsection{Quadratic Goldreich-Levin}
A fundamental algorithmic primitive in Fourier analysis is the \emph{Goldreich-Levin} algorithm~\cite{goldreich1989hard}. 
Given query access to a bounded function $f:\F_2^n\to[-1,1]$ and a threshold $\tau>0$, the algorithm efficiently finds all linear characters $x\mapsto (-1)^{\langle \alpha,x\rangle}$
having correlation at least $\tau$ with $f$. Equivalently, it finds the large Fourier coefficients of $f$ in time polynomial in $n$ and $1/\tau$. 
Goldreich--Levin and its variants have become basic tools in coding theory~\cite{goldreich1989hard,TrevisanCoding}, computational learning theory~\cite{KM93}, pseudorandomness and cryptography~\cite{goldreich1989hard,HILL99}, and complexity theory.

There is a useful higher-order perspective on this result. 
The $U^2$ Gowers norm measures linear structure, and its inverse theorem states that
\[
    \|f\|_{U^2}\ \text{large}
    \qquad\Longrightarrow\qquad
    \left|\E_x f(x)(-1)^{\langle \alpha,x\rangle}\right|
    \ \text{large}
\]
for some $\alpha\in\F_2^n$. 
From this viewpoint, Goldreich--Levin is an
\emph{algorithmic inverse theorem for the $U^2$ norm}: 
it not only guarantees the existence of a correlated linear phase, but efficiently finds one.

The next level of this hierarchy is the $U^3$ norm, whose associated structured objects are quadratic phases. 
In a seminal work, Samorodnitsky~\cite{samorodnitsky2007low} proved, assuming the PFR conjecture,
an inverse theorem of the form
\[
    \|f\|_{U^3}\ \text{large}
    \qquad\Longrightarrow\qquad
    \left|\E_x f(x)(-1)^{q(x)}\right|
    \ \text{large}
\]
for some quadratic polynomial $q:\F_2^n\to\F_2$. 
The resolution of PFR therefore gives an unconditional polynomial inverse theorem for the $U^3$ norm.

This raises the corresponding algorithmic question: can one efficiently find the correlated quadratic phase promised by the inverse theorem? 
Tulsiani and Wolf~\cite{tulsiani2014quadratic} initiated the study of this \emph{quadratic Goldreich--Levin problem} and gave an algorithm using $O(n^4)$ queries, but with running time $2^n$. 
More recently, Bri\"et and Castro~\cite{briet2026near}
improved the query complexity to $O(n^2)$ and obtained
$n^{O(\log n)}$ running time. 
Using our algorithmic PFR theorem, we obtain polynomial dependence on both parameters.\footnote{The term \emph{quadratic Goldreich--Levin} is used in the literature for several closely related guarantees. Tulsiani and Wolf~\cite{tulsiani2014quadratic} study a multiplicative approximation to the maximum quadratic correlation, whereas Bri\"et and Castro~\cite{briet2026near} obtain an additive approximation. Our result is in the former regime.}

\begin{theorem}[Quadratic Goldreich--Levin]
\label{thm:quadratic-gl}
Let $f:\F_2^n\to[-1,1]$, let
$\varepsilon,\delta>0$. Let
\[
\tau\coloneqq   \max_{\substack{q:\F_2^n\to\F_2\\ q\ \mathrm{quadratic}}}
    \left|
        \E_{x\in\F_2^n} [f(x)(-1)^{q(x)}]
    \right|
\]
There is an algorithm that, given query
access to $f$, runs in time
\[
    \poly\!\left(n,1/\tau,\log1/\delta\right)
\]
and, with probability at least $1-\delta$, outputs a quadratic polynomial
$q:\F_2^n\to\F_2$ satisfying
\[
    \left|
        \E_{x\in\F_2^n} [f(x)(-1)^{q(x)}]
    \right|
    \geq
   \tau^C,
\]
for some universal constant $C>1$. 
\end{theorem}

As in the work of Tulsiani and Wolf, a quadratic Goldreich--Levin algorithm also yields an efficient quadratic decomposition theorem: a bounded function can be decomposed into a polynomial number of quadratic phases together with error terms that are small in $U^3$ and $L_1$.

\begin{corollary}
Let $\varepsilon\geq 0$ and $f:\F_2^n\rightarrow [-1,1]$. Given query access to $f$, there is a $\poly(n,1/\varepsilon)$-time algorithm that 
outputs with probability \(\geq 2/3\) a decomposition
\[
    f
    =
    c_1(-1)^{p_1(\cdot)}
    + \cdots
    + c_r(-1)^{p_r(\cdot)}
    + g + h,
\]
where the \(c_i\) are constants, \(p_i\)s are quadratic polynomials, 
\(r\leq \poly(1/\epsilon)\), $  \|g\|_{U^3}\leq \epsilon$ and $    \|h\|_1\leq \epsilon.$
\end{corollary}

Furthermore, \cref{thm:main} also has other learning-theoretic consequences. 
Since quadratic polynomials over $\F_2$ are precisely the codewords of the second-order Reed--Muller code $\mathsf{RM}(n,2)$, it gives a polynomial-time improper agnostic learner for $\mathsf{RM}(n,2)$.\footnote{An open question is whether one can obtain a \emph{proper} agnostic learner for $\mathsf{RM}(n,2)$, as well as for stabilizer states.} 
In addition, standard reductions from agnostic learning~\cite{kearns1992toward,feldman2009distribution} yield polynomial-time PAC learners for low-weight linear thresholds of quadratic phases.

\subsubsection{Quantum applications}

The study of learning near-stabilizer structure has developed through a sequence of works, including stabilizer learning and estimation~\cite{montanaro-bell-sampling,gross2021schur,grewal2023improved}, tomography of states with high stabilizer dimension~\cite{grewal2023efficient,Leone2024learningtdoped}, and agnostic tomography of stabilizer states~\cite{Grewal2026agnostictomography,chen2025stabilizer}.

More recently, Arunachalam and Dutt~\cite{ad2024tolerant,SCforstabilizers,arunachalam2026tomography} uncovered a surprising connection between these quantum learning problems and additive combinatorics, showing that efficient algorithmic PFR would imply efficient algorithms for several of them. 
Their reductions, however, require a slightly more robust form of algorithmic PFR than the one stated in \cref{thm:main}: instead of receiving perfectly uniform samples from the underlying small-doubling set, the algorithm must tolerate samples drawn from a distribution that is only approximately uniform. We show that our algorithm extends to precisely this setting, and this allows us to make the resulting quantum algorithms fully polynomial-time.

\paragraph{Approximate algorithmic PFR theorem.} 

Recall that \cref{thm:main} assumes uniform sampling and membership-oracle access to a small-doubling set $A$. For the quantum applications, the sampling distribution arising from the reduction need not be exactly uniform. We therefore consider a distribution $\mu_A$ supported on $A$ that is $R$-uniform, meaning that
\[
\frac{1}{R|A|}\leq \mu_A(x)\leq \frac{R}{|A|}
\]
for every $x \in A$.
We prove that algorithmic PFR remains efficient under this weaker sampling model: given membership-oracle access to $A$ and samples from an $R$-uniform distribution on $A$, one can recover a PFR subspace in $\poly\left(n,K,R,\log\frac1\delta\right)$ time.

The main idea is to \emph{uniformize} the samples before invoking our PFR algorithm. We define a lazy random walk on $A$ whose stationary distribution is uniform on $A$. The upper bound in the $R$-uniformity condition gives a warm start for this walk, while the lower bound, together with the small-doubling hypothesis, yields a polynomial lower bound on its conductance. A warm-start mixing bound of Lov\'asz and Simonovits then implies that the walk reaches a distribution arbitrarily close to uniform in $\poly(K,R,1/\varepsilon)$ steps. Replacing each uniform sample required by \cref{thm:main} with an independently uniformized sample gives the desired robust version of algorithmic PFR.

\paragraph{Consequences for quantum learning.} 
We now combine approximate algorithmic PFR with the reductions developed in~\cite{SCforstabilizers,arunachalam2026tomography}. For a pure state $\ket{\psi}$, define its \emph{stabilizer fidelity} by
\[
\stabfidelity{\ket{\psi}}=\max_{\ket{\phi}\in \textsf{Stab}}\{|\langle \phi|\psi\rangle|^2\},
\]
where $\textsf{Stab}$ is the class of $n$-qubit stabilizer states.
We obtain the following consequences.

\begin{enumerate}
    \item \textbf{Self-correction.}
    Given copies of a pure state $\ket{\psi}$ satisfying   $\stabfidelity{\ket{\psi}}\geq \tau$,
        there is a $\poly(n,1/\tau)$-time algorithm that outputs a stabilizer state $\ket{\phi}$ satisfying
    \[
        |\langle\phi|\psi\rangle|^2\geq \tau^C
    \]
    for some universal constant $C > 1$.
    \item \textbf{Improper agnostic tomography of stabilizer states.}
    Combining self-correction with the boosting framework of~\cite{arunachalam2026tomography}
    gives a polynomial-time improper agnostic tomography algorithm for stabilizer states. 
    Given copies of an arbitrary pure state $\ket{\psi}$, the algorithm outputs an efficiently representable state $|\widehat \psi\rangle$ satisfying 
    \[
        |\langle\widehat{\psi}|\psi\rangle|^2
        \geq   \stabfidelity{\ket{\psi}}-\epsilon
    \]
    in time polynomial in $n$ and $1/\varepsilon$.

    \item \textbf{Tomography of bounded-stabilizer-extent states.}
    Combining self-correction with the extent-boosting framework of~\cite{arunachalam2026tomography} gives a tomography algorithm for pure states of bounded stabilizer extent. If \(\ket{\psi}\) has stabilizer extent at most \(\xi\), then it can be learned to trace-distance error \(\varepsilon\) using $\poly(n,\xi,1/\varepsilon)$ time and copies.
   \item \textbf{Improper agnostic tomography of high stabilizer-dimension states.} 
   Prior work~\cite{SCforstabilizers} showed that an efficient algorithmic PFR theorem implies self-correction for states of high stabilizer dimension.\footnote{An $n$-qubit stabilizer state has stabilizer dimension $k$ if it is stabilized by an Abelian group of $2^k$ Paulis.}. 
   Concretely, if $\ket\psi$ has fidelity at least $\tau$ with a state of stabilizer dimension $n-t$, then in time $\poly(n, 2^t, 1/\tau)$ we can find a stabilizer state $\ket{\phi}$ satisfying 
   \[
   |\langle \phi | \psi \rangle|^2 \geq \poly(2^{-t}\tau).
   \]
   Combined with the boosting framework of~\cite{arunachalam2026tomography}, this gives a $\poly(n,2^t,1/\varepsilon)$-time improper agnostic tomography algorithm for states with stabilizer dimension $n-t$. This extends the polynomial-time algorithm of~\cite{grewal2023efficient} to the agnostic setting, improving on the previously known quasipolynomial-time algorithm even for $t=O(\log n)$~\cite{chen2025stabilizer}.
   
   \item \textbf{Learning Clifford structure.} 
   The reduction of~\cite{dutt2026learning} turns stabilizer self-correction into an efficient learner for Clifford structure. 
    In particular, given query access to a unitary $U$ of Clifford fidelity at least $\tau$, we can efficiently find a unitary $V$ satisfying
    $|\langle\!\langle V|U \rangle\!\rangle|^2 \geq \poly(\tau)$.
    Together with boosting~\cite{arunachalam2026tomography,dutt2026learning}, this yields tomography algorithms for unitaries and Hamiltonians of bounded Clifford extent $\xi$, running in $\poly(n,\xi,1/\varepsilon)$ time.\footnote{We write $\choiket{U}\coloneqq(U\otimes I)\ket{\mathrm{EPR}_n}$, where $\ket{\mathrm{EPR}_n}$ consists of $n$ EPR pairs.}
   
\end{enumerate}

In each case, the previous reductions were conditional on an efficient robust form of algorithmic PFR, or consequently inherited a quasipolynomial dependence on the relevant accuracy parameter. Our approximate algorithmic PFR theorem supplies the missing polynomial-time primitive and yields polynomial dependence on these parameters.

\paragraph{Acknowledgments.} AG thanks Seyoon Ragavan for useful discussions, and the UK AISI alignment project for supporting access to ChatGPT Pro. AG is supported in part by a Simons investigator award and NSF grant CNS-2534400. SG is supported by the Herman Goldstine Memorial Postdoctoral Fellowship. SA and AD thank Jop Briet,  Davi Castro Silva, and Tom Gur for several discussions surrounding this topic in the past.

\paragraph{AI Disclosure.} 
The algorithm and proof strategy underlying the main theorem were suggested by ChatGPT 5.6 Sol. 
We found it crucial to prompt the model specifically to study the GGMT proof and search for a way to algorithmize its argument; more open-ended attempts to solve the problem without this direction were not successful. 
Subsequent interactions with AI were used to refine and clarify the proof. 
The authors independently verified, developed, and simplified the argument presented here and take full responsibility for its correctness, exposition, and attribution.

\section{Technical overview}
\label{sec:tecoverview}
To motivate the algorithm underlying \cref{thm:main}, we first revisit the proof of the PFR theorem due to Gowers, Green, Manners and Tao~\cite{ggmt2024torsion,ggmt2025conjecture} (\Cref{sec:tech-overview-ggmt}). \Cref{sec:tech-overview-algo} then discusses the main ideas necessary to algorithmize the existential proof.

\subsection{The GGMT existential argument}
\label{sec:tech-overview-ggmt}
\paragraph{An entropic formulation of the PFR theorem.} The starting point of the proof by \cite{ggmt2025conjecture} is an entropic formulation of the PFR theorem, that deals with \emph{random variables} instead of sets.
Throughout, all random variables take values in
$\F_2^n$. Given random variables $X$ and $Y$, 
let $X'$ and $Y'$ be independent copies of $X$ and $Y$, respectively. For a random variable $X$, its entropy is $\entropy[X] \coloneqq \sum_{x\in \FF_2^n}\Pr[X=x]\log \frac{1}{\Pr[X=x]}$.
The \emph{entropic Ruzsa distance}\footnote{Traditional Ruzsa calculus studies the \emph{combinatorial} Ruzsa distance $\textsf{d}_{\rm comb} [A;B] := \log |A+B| - \frac12 \log |A| - \frac12 \log |B|$ between \emph{sets} $A, B$. The entropic Ruzsa distance was first introduced by \cite{tao2010sumset}. It turns out that sets and the combinatorial Ruzsa distance are too rigid for the iterative argument in \cite{ggmt2025conjecture} to go through. A key proof idea of \cite{ggmt2025conjecture} is to instead work with \emph{distributions} over sets and the entropic Ruzsa distance, to capture more flexible and fine-grained information.} between $X$ and $Y$ is defined as 
$$
    \ruzsadist[X;Y]
    \coloneqq 
    \entropy [X'+Y']
    -\frac12 \entropy [X']
    -\frac12 \entropy[Y'].
$$
The entropic Ruzsa distance may be viewed as a measure of entropy growth under convolution. It is not a true metric---in particular, the self-distance $\ruzsadist[X;X]$ measures the entropy growth produced by adding two independent samples from the same distribution, and need not be $0$. 
It vanishes precisely when $X$ is uniformly distributed on an affine subspace of $G$.
Thus, small entropic self-distance indicates that $X$ has approximate additive structure.

To connect the distributional notion with the usual combinatorial notion of small doubling, let $A\subseteq \FF_2^n$ be nonempty and let $U_A$ denote the uniform distribution on $A$.  
If $|A+A|\leq K|A|$, then, for an independent
copy $U_A'$ of $U_A$, 
$$
    \ruzsadist[U_A;U_A]
    =
    \entropy[U_A+U_A']-\entropy[U_A]
    \leq
    \log |A+A|-\log |A|
    \leq
    \log K,
$$
using the trivial bound of $\entropy[D] \leq \log \supp(D)$. 
Hence, a set with doubling constant at most $K$ gives rise to a distribution with entropic self-distance at most $\log K$. 
The entropic formulation of PFR established by \cite{ggmt2025conjecture} states that for every $\FF_2^n$-valued random variable $X$, there is a subspace $H\leq \FF_2^n$ such that
$$
\ruzsadist[X;U_H]
=
O\left(\ruzsadist[X;X]\right),
$$
where $U_H$ denotes the uniform distribution on $H$. In this language, the PFR theorem is an \emph{inverse theorem} for entropic Ruzsa distance: a distribution with small self-distance must be close, in entropic Ruzsa distance, to the uniform distribution on a subspace.  See \Cref{thm:entropic-pfr} for the formal statement.

\paragraph{An entropy-decrement argument.} 
The proof of Gowers, Green, Manners, and Tao~\cite{ggmt2025conjecture} is an iterative entropy-decrement argument. Starting from 
$$ X^{(0)}=Y^{(0)}=U_A, \qquad \ruzsadist[X^{(0)};Y^{(0)}]\leq\log K,$$
the aim is to repeatedly replace the current pair by one with smaller entropic Ruzsa distance, eventually reaching the structured regime. 
At the same time, one must control how far the distributions move along the way, otherwise the structured distribution found at the end could have little to do with the original distribution $U_A$.
To balance these two objectives, given a current pair $(X,Y)$, one measures the quality of a candidate update $(X', Y')$ using the potential
$$
    \Phi_{X,Y}(X',Y')
    \coloneqq
    \ruzsadist[X';Y']
    +
    \eta\left(
        \ruzsadist[X;X']
        +
        \ruzsadist[Y;Y']
    \right),
$$
where $\eta>0$ is a sufficiently small absolute constant. 
The first term measures how much additive complexity remains in the new pair, 
while 
the latter two terms penalize moving too far from the current pair. 
Thus, a good update must both reduce the entropic Ruzsa distance and remain quantitatively connected to the distributions from which it arose.

The central decrement statement is as follows: whenever $\ruzsadist[X;Y]$ is above a fixed absolute constant, GGMT showed that there are five explicit update families to the random variables $(X,Y)$, one of which produces a pair $(X',Y')$ satisfying
\[
 \mathbb E\bigl[
    \Phi_{X,Y}(X',Y')\bigr]
    \leq
    (1-\Omega(1))\cdot \ruzsadist[X;Y].
\]
Some of these updates involve randomness, and the expectation is over this randomness. All the five update families are built from two basic operations that we define shortly: \emph{summing} and \emph{fibering}.
In a bit more detail, the five updates consist of self-sums, cross-sums, self-fibers, mixed fibers, and an ``endgame'' transformation consisting of first fibering and then summing.

Summing replaces the current variables $(X,Y)$ by self-sums $(X+X', Y+Y')$ or cross-sums $(X+Y, X'+Y')$. The structured distributions we seek are \emph{stable under addition}: if $U_H$ is uniform on a subspace $H$, then $U_H + U_H \sim U_H$.
For certain kinds of distributions resembling $U_H$, summing can achieve a large decrement in the entropic Ruzsa distance. For example, uniform distributions over dense random subsets $A \subseteq H$ that are contained within a subspace $H$, must remain within $H$, but expand to take up more room inside $H$, become closer to $U_H$ and therefore decreasing the entropic Ruzsa distance.

However, there are other random variables for which summing \emph{increases} the entropic Ruzsa distance. For example, consider (a uniform distribution over) the union of a small number of unevenly spaced cosets of a subspace $H \le \FF_2^n$. Summing is a bad idea here---it approximately doubles the entropic Rusza distance. Luckily, for such distributions, our second operation, \emph{fibering}, effectively filters out a small number of cosets, making the distance drop to $O(1)$. Fibering can be thought of as ``quotienting'' operation that isolates algebraic structure in a way that complementary to summing: it restricts attention to preimages giving rise to a particular sum. Concretely, given a pair $(X,Y)$, fibering outputs the conditional distribution $X| (X+Y) = z$, for a randomly sampled $z \sim X+Y$.
We refer to $z$ as the \emph{fiber label}. 
The randomness in the decrement statement above comes from sampling these fiber labels according to their natural distribution.

Summing and fibering therefore provide two complementary ways to make progress. Summing studies global behavior of the current distributions under addition, whereas fibering searches locally for a structured component hidden inside an additive level set. Miraculously, \cite{ggmt2025conjecture} prove that these operations are sufficient.
Four of five candidate updates arise directly from these operations: self-sums, cross-sums, self-fibers, and mixed fibers. 
Roughly speaking, either one of the relevant sums already exhibits lower additive complexity or one of its fibers reveals a more structured conditional distribution. The fifth, or \emph{endgame}, update handles the case in which none of the four simpler transformations gives sufficient progress. Its analysis exploits additional structure forced by this failure; importantly for us, the resulting update can still be implemented using only the same primitives, namely two fiber operations followed by a sum.

So the high-level idea of GGMT is, starting from
$$
    X_0=Y_0=U_A,
    \qquad
    \ruzsadist[X_0;Y_0]\leq \log K,
$$
each favorable transformation decreases the entropic Ruzsa distance by a constant factor.  
After \(O(\log\log K)\) favorable steps, the process reaches a pair at constant entropic Ruzsa distance, while the movement penalties ensure that the total distance traveled from \(U_A\) remains \(O(\log K)\). For our purposes, the essential takeaway is that each step comes from five transformations built from sums and fibers, one of which makes substantial progress. 

The next subsection explains how to implement this recursion efficiently without computing the relevant entropies or identifying which transformation is favorable.

\subsection{Our algorithmic approach}
\label{sec:tech-overview-algo}
The exposition above suggests a natural algorithmic strategy.
Recall that the algorithm is given uniform sample access to an unknown set $A \subseteq \F_2^n$ with doubling constant at most $K$, together with a membership oracle for $A$. 
Starting from $X_0=Y_0=U_A$, we would like to follow the GGMT recursion, repeatedly applying transformations that decrease the potential until the entropic Ruzsa distance becomes constant. 
Since each favorable transformation reduces the distance by a constant factor, only $O(\log\log K)$ successful transformations are needed.  Turning this strategy into an efficient algorithm presents three main difficulties.
\begin{enumerate}
    \item  The GGMT proof identifies a favorable transformation through entropy inequalities, but estimating the relevant entropies of the intermediate distributions appears to be computationally infeasible.
    \item The intermediate random variables are defined recursively through sums and fibers, and we need an efficient way to sample from and manipulate these distributions without ever writing down their probability mass function, which could have exponential support.
    \item Once the recursion reaches constant entropic Ruzsa distance, we need to recover an explicit subspace from oracle access to the terminal distribution. 
\end{enumerate}
We address each of these difficulties in turn.

\paragraph{Random traversal of the GGMT tree.} 
We begin with the first difficulty. 
For a current pair $(X,Y)$, the decrement theorem of GGMT guarantees that one of the five GGMT update families has small average potential. 
Determining which family achieves the decrement, by estimating $\Phi_{X,Y}(X',Y')$, would require estimating entropies of distributions that may have exponentially large support. This appears computationally infeasible in general.\footnote{At the root of the recursion, the self-distance of $U_A$ can in fact be estimated efficiently in $\poly(n,K)$ time using uniform samples and membership queries. It is unclear how to extend this approach to the intermediate, generally nonuniform distributions, for which no analogous membership oracle is available.}

Our algorithm therefore \emph{never computes the entropy, entropic Ruzsa distance, or the potential.}
Instead, at every step it chooses one of the five candidate update families uniformly at random, and samples any required fiber labels from their natural distributions.
The decrement theorem of \cite{ggmt2025conjecture,ggmt2024torsion} shows that a random step has a constant positive probability of producing a constant-factor
decrement.
Consequently, a sequence of $O(\log \log K)$ favorable random steps reaches the constant-distance regime with probability $\exp\bigl(-O(\log\log K)\bigr) = 1/{\operatorname{polylog} K}$.

\paragraph{Implicit access to intermediate distributions.} 
The second difficulty is maintaining algorithmic access to the random variables encountered along this traversal. 
Sampling access alone is not sufficient, since it does not allow us to continue the traversal.
For example, suppose we wish to implement the fiber $X_z \coloneqq X \mid X+Y=z$.
Writing $p_X$ and $p_Y$ for the probability mass functions of $X$ and $Y$, we have 
\[
    \Pr[X_z=x]
    =
    \Pr[X=x \mid X+Y=z]
    =
    \frac{p_X(x)p_Y(z-x)}{p_{X+Y}(z)}.
\]
A natural way to sample from $X_z$ is by rejection sampling: repeatedly sample $x \sim X$ and accept with probability proportional to $p_Y(z-x)$.
A sampler for $Y$, however, gives no way to implement such an acceptance probability at a specified point. 

We therefore augment sampling access with precisely the additional primitive needed for this rejection sampler as above: for every intermediate random variable $X$, we maintain not only an exact sampler $\Sample_X$, but also a procedure $\Coin_X$ which, on input $x \in \F_2^n$, samples and outputs a bit satisfying
\[
\Pr[\Coin_X(x) = 1] = \frac{p_X(x)}{M_X}, 
\]
where $M_X \ge \norm{p_X}_\infty$ is an implicit normalization we call the \emph{envelope}. 
Importantly, the envelope need not beknown to the algorithm.
This fact turns out to be crucial in ensuring efficiency.
Thus, $\Coin_X$ allows us to accept $x$ with probability proportional to its probability mass $p_X(x)$, without ever explicitly evaluating $p_X(x)$.
The pair $(\Sample_X, \Coin_X)$ provides our needed algorithmic access to $X$.

Importantly, this access is available at the root of the recursion directly from the input. For $X^{(0)}, Y^{(0)} = U_A$, the sampler $\Sample_{U_A}$ is given by the hypothesis, and the membership oracle serves as $\Coin_{U_A}$, since we may take $M_{U_A} = {1} / {|A|}$.
Although $|A|$ is unknown to the algorithm, the corresponding coin satisfies 
\[
\Pr[\Coin_{U_A}(x)=1]
=
\frac{p_{U_A}(x)}{M_{U_A}}.
\]
Hence the membership oracle for $A$ implements $\Coin_{U_A}$ exactly.

This form of access can be propagated exactly through both of the basic GGMT operations, summing and fibering.
For a sum $Z = X+Y$, an exact sample is obtained by drawing $x \sim X$ and $y\sim Y$ independently and returning $x+y$.
On input $z$, we may implement $\Coin_Z(z)$ by sampling $x \sim X$ and returning the outcome of $\Coin_Y(z-x)$, since 
\[
\Pr[\Coin_Z(z)=1]
= \sum_x p_X(x)\frac{p_Y(z-x)}{M_Y}
= \frac{p_Z(z)}{M_Y}.
\]
Thus, this is a valid coin for $Z$ with implicit normalization $M_Z=M_Y$.

For a fiber $X_z = X \mid X+Y =z$, we repeatedly sample $x \sim X$ until $\Coin_Y(z-x)$ accepts, and return the first accepted sample. 
The probability that a particular $x$ is returned is proportional to $p_X(x) p_Y(z-x)$, so the returned sample is distributed exactly according to $X_z$.
A coin for $X_z$ can be implemented by independently running $\Coin_X(x)$ and $\Coin_Y(z-x)$ and accepting if both succeed. Its acceptance probability is 
\[
\frac{p_X(x)p_Y(z-x)}{M_XM_Y}
= \frac{p_{X_z}(x)} {M_XM_Y/p_{X+Y}(z)},
\]
so it is again of the required form.
Consequently, starting only from the sample and membership oracle for $A$, we can recursively construct exact algorithmic access to every intermediate random variable appearing in the GGMT recursion.

\paragraph{Bucketing to maintain efficient access.}
We have explained how to maintain the algorithmic access needed for every intermediate random variable appearing in the GGMT recursion. 
We have to ensure that this access can be implemented efficiently.  The main difficulty comes from fibering. 
If $Z=X+Y$, then for a fiber label $z$, each sample $x \sim X$ is accepted by the rejection sampler with probability ${p_Z(z)}/{M_Y}$.
Thus, the expected number of samples needed to obtain one sample from $X_z$ is ${M_Y}/{p_Z(z)}$. To track how these rejection-sampling costs evolve through the recursion, we define the \emph{slack} of our algorithmic access to $X$ by 
\[
\Delta[X]
\coloneqq
\entropy[X]+\log M_X
=
\mathbb E_{x\sim X}
\left[
\log\frac{1}{r_X(x)}
\right], 
\]
where we set $r_X(x) \coloneqq \Pr[\Coin_X(x) = 1]$ for convenience.\footnote{Note that the slack $\Delta[X]$ depends on the particular envelope $M_X$ we use in the algorithmic representation of $X$, and, contrary to what the notation suggests, is not determined by $X$ alone.}
Thus, $\Delta[X]$ measures the typical logarithmic inverse acceptance probability of $\Coin_X$ on a sample $x \sim X$. 
Small slack means that $\Coin_X$ has reasonably large acceptance probability on a typical sample from $X$.

The efficiency of our representations can be quantified through the slack and the entropic Ruzsa distance. We start out with no slack,
\begin{align*}
    \Delta[U_A] = \entropy[U_A] + \log \frac{1}{|A|} = 0.
\end{align*}
Although the root distribution has zero slack, sums and fibers can increase it.
For sums, a direct calculation gives 
\[
\Delta[X+Y]
\leq
\Delta[Y]+2\der{X}{Y},
\]
and for a fiber $X_Z = X \mid X+Y=Z$, 
\[
\E_Z[\Delta[X_Z]] = \Delta[X]+\Delta[Y].
\]
This is problematic, because for a fixed fiber label $z$, the rejection sampler requires $\tfrac{M_Y}{p_Z(z)}$ trials in expectation, and 
\[
\E_{z\sim Z} \left[ \log\frac{M_Y}{p_Z(z)} \right] \leq \Delta[Y]+2\der{X}{Y}.
\]
Thus, controlling the slack also controls the typical cost of sampling fibers.

To address this, we insert a \emph{bucketing} step after each GGMT transformation.\footnote{Our bucketing step is analogous to the probability bucketing used in distribution testing~\cite{batu2001testing,chakraborty2010newresultsquantumproperty}, where the domain is partitioned into ranges of comparable probability mass so that the conditional distribution on each bucket is approximately flat.} 
Informally, bucketing conditions the current random variable on an approximate level set of the probability mass function, so that retained points have comparable values of $r_X(x)$.
Crucially, the bucketing can itself be implemented using only the sampler and coin access described above. 
Given a random variable $Z$, we sample $z \sim Z$ and repeatedly run $\Coin_Z(z)$. 
If the first success occurs on trial $T$, we record $J \coloneqq \lfloor \log_2 T \rfloor$ and condition on the value of $J$. 
Since $T$ is concentrated around $1/r_Z(z)$, fixing $J=j$ effectively restricts $Z$ to a probability scale on which $r_Z(z)$ is roughly $2^{-j}$.
The resulting conditional law $Z_J$ again admits exact sample-and-coin access, and the key guarantee we show is $\E_J[\Delta[Z_J]] = O(1)$.
Thus, bucketing regularly restores constant slack and prevents the rejection-sampling costs from accumulating throughout the recursion. 

There is one additional concern: conditioning the intermediate distributions could interfere with the entropy decrement that drives the GGMT argument. 
Fortunately, conditioning a random variable $Z$ by another random variable $J$ can increase the entropic Ruzsa distance only by their mutual information, and bucketing reveals only 
\[
\mathbf{I}[Z;J]=
O\bigl(\log\Delta[Z]\bigr)
\]
bits of information.\footnote{Strictly speaking, bounds such as $O(\log\Delta[Z])$ should be read as $O(\log(\Delta[Z]+2))$, so that the expression remains well behaved when $\Delta[Z]$ is small. We suppress such harmless additive constants inside logarithms here and throughout this paper.}
At entropic Ruzsa distance $d$, the relevant slack is $O(d)$, so bucketing incurs only an $O(\log d)$ loss, whereas a good GGMT step yields a decrement of order $\Omega(d)$. 
Above a sufficiently large constant threshold, the decrement therefore dominates the cost of bucketing.

\paragraph{The total cost of a GGMT trajectory.}
Above a fixed constant threshold $d_{\rm term}$, a random GGMT update followed by bucketing has constant probability of both achieving the desired potential decrement and producing a new pair with constant slack and controlled access cost. 
Once the distance falls below $d_{\rm term}$, the same analysis gives a constant-probability update that keeps the distance, movement penalties, and slack bounded by absolute constants. 
We therefore run the recursion for a fixed $T=O(\log\log K)$ iterations without attempting to detect when the terminal regime is reached. 
With probability $1/\polylog K$, the resulting trajectory is good, and its terminal variable satisfies
\[
\der{X_T}{X_T}=O(1),\qquad
\der{U_A}{X_T}=O(\log(2K)),\qquad
\Delta[X_T]=O(1).
\]

It remains to bound the cost of accessing $X_T$.
Let $d_t \coloneqq \ruzsadist[X_t, Y_t]$.
On a good trajectory, the rejection-sampling bounds above imply that the sampler-and-coin cost grows from level $t$ to level $t+1$ by at most a factor $2^{O(d_t)}$.
Hence, the cumulative overhead is
\[
\prod_{t<T}2^{O(d_t)}
=2^{O(\sum_{t<T}d_t)}.
\]
Before reaching the terminal regime, the distances decrease by a constant factor from $d_0 \le \log K$; afterward, they remain $O(1)$. 
Since $T= O(\log\log K)$, we have $\sum_{t<T}d_t = O(\log K)$ an therefore the total overhead is 
\[
2^{O(\log K)}=K^{O(1)}.
\]
Thus, along a good trajectory, both constructing the terminal access procedures and calling them have polynomial expected cost (\Cref{claim:cost_good_trajectories}).

\paragraph{Extracting a candidate subspace.} 
It remains to extract a PFR subspace from the terminal random variable $X_T$. 
A naive approach would be to sample from $X_T$ and take the span, but even a few rare outliers can make this span much too large. 
Instead, we look at the self-sum $X_T + X_T'$. 
A value $z$ that occurs with large probability in this self-sum has many likely representations $z = x + x'$, and therefore captures an additive direction that is genuinely present in the high-mass part of $X_T$. 
We use $\Coin_{X_T + X_T'}$ to identify these \emph{heavy self-sums} and take the span of $O(n)$ independent heavy samples.

The constant self-distance and slack of $X_T$ ensure that heavy self-sums carry constant mass, while the entropic PFR theorem implies that they are confined to a constant-index enlargement of an underlying subspace.
Consequently, with constant probability their span $V$ satisfies $\der{X_T}{U_V}=O(1)$, and the triangle inequality gives $\der{U_A}{U_V}=O(\log K )$. 
This subspace extraction phase uses only $O(n\log n)$ calls to the terminal access algorithms for $X_T$ and polynomial additional computation.

\paragraph{Verifying a PFR subspace.}
A remaining issue is that the candidate subspace $V$ can be larger than $A$. 
The distance bound implies $\abs{V}\le K^{O(1)}\abs{A}$.
Thus, by shrinking $V$ arbitrarily by $O(\log K)$ dimensions, we obtain a subspace $V' \le V$ such that either $V' = \{0\}$ or $\abs{V'} \le \abs{A}/2$. 
Importantly, this does not require knowing $\abs{A}$.
Moreover, for a uniform sample $a \sim U_A$, the coset $a + V'$ contains a $K^{-O(1)}$ fraction of $A$ with constant probability.

To recognize successful candidates, we estimate
\[
\alpha\coloneqq\frac{|A\cap(a+V')|}{|A|},\qquad
\beta\coloneqq\frac{|A\cap(a+V')|}{|V'|}.
\]
We estimate $\alpha$ by sampling from $A$ and testing membership in $a + V'$, and estimate $\beta$ by sampling uniformly from $a + V'$ and testing membership in $A$.
The first quantity measures the fraction of $A$ captured by the coset, while the identity $\beta/\alpha=|A|/|V'|$ lets us certify that $|V'| \le |A|$ without knowing $|A|$. 
With appropriate margins, every accepted candidate has $|V'|\le |A|$ and captures a $K^{-O(1)}$ fraction of $A$, 
except with the allotted verification error probability.

Finally, these guarantees imply the desired covering property. Let $B\coloneqq A\cap(a+V')$,
so that \(|B|\geq K^{-O(1)}|A|\), and take a maximal family of pairwise disjoint translates \(B+t_1,\ldots,B+t_m\) with \(t_i\in A\). Since each translate lies in \(A+A\),
\[
    m|B|\leq |A+A|\leq K|A|,
\]
and hence \(m=K^{O(1)}\). By maximality, for every \(t\in A\), the translate \(B+t\) intersects some \(B+t_i\). Thus there exist \(b,b'\in B\) such that \(b+t=b'+t_i\). Since \(b,b'\) lie in the same coset \(a+V'\), we have \(b'-b\in V'\), and therefore \(t-t_i\in V'\). Hence
\[
A\subseteq\bigcup_{i=1}^m(t_i+V'),
\]
so \(A\) is covered by \(K^{O(1)}\) translates of \(V'\).

\paragraph{Amplifying success.}
A single trial consists of (1) traversing the GGMT tree to obtain a terminal random variable $X_T$, (2) extracting a candidate subspace $V$, and (3) verifying the resulting candidate. 
If the trial makes bad choices along the way, the rejection sampling can become too expensive, so we impose a sufficiently large polynomial cutoff time for traversal and extraction, after which the trial is aborted. 
A single trial succeeds with probability $\Omega(1/\polylog K )$. 
Repeating $\polylog(K)\log(2/\delta)$ independent trials, with verification error set appropriately, gives overall success probability at least $1-\delta$ and total running time $\poly(n,K,\log(1/\delta))$, proving \Cref{thm:main}.

\section{Preliminaries}\label{sec:prelims}
\paragraph{Notation and Terminology.}
Throughout this paper, we will work in the ambient space of $\FF_2^n$. For $n \in \mathbb{N}$, we define $[n]:= \{1,2,\ldots,n\}$ and $[n]_0 := \{0\} \cup [n]$. All the random variables in this paper will take values in $\FF_2^n$. For a non-empty finite set $A \subseteq G$, we let $U_A$ denote the uniform distribution on $A$. All logarithms (denoted by $\log$) are base $2$ and we will specifically use $\ln$ to denote the natural logarithm. For a random variable $X$, we implicitly denote $p_X$ to the distribution of the random variable $X$. 
Throughout this paper, we use the shorthand $\mathrm{polylog}(K)$ for $\log^{O(1)}(2K)$ to avoid weird behaviour when $K\in [1,2]$, $O(\log K)$ as a shorthand for $O(\log K+2)$ for similar reasons. 

\subsection{Additive combinatorics}
For a set $A \subseteq \FF_2^n$, we denote the sumset $A+A := \{a_1 + a_2 : a_1, a_2 \in A\}$, and more generally, we will use the sumset notation $\ell A:= A + \ldots + A$ to denote the sum of $\ell$ copies of $A$. 

\paragraph{Doubling constant.}
We call the ratio $|A+A|/|A|$ the \emph{doubling constant} of $A$. The PFR theorem \Cref{thm:combinatorial_pfr} tightly characterizes the doubling constant in terms of the algebraic structure it imposes, and is seen as a measure of how close $A$ is to being a subgroup/subspace.

\paragraph{Additive energy.}
Another notion that captures approximate subgroup structure of a set is its \emph{additive energy} $E(A)$, which has several equivalent expressions
\begin{align*}
    E(A) &:= \{ (a, b, c, d) : a, b, c, d \in A, a+b = c+d\}\\
        &= \sum_{v \in \FF_2^n} \bigl( \# \{(a, b) : a, b \in A, a+b = v\} \bigr)^2\\
        &= \sum_{v \in \FF_2^n} |A \cap (A+v)|^2.
\end{align*}
It can then be shown that a set $A$ with doubling constant $K$ satisfies $E_4(A) \geq |A|^3/K$.

\paragraph{Ruzsa distance.} For finite nonempty subsets \(A,B\) of $\F_2^n$, the (combinatorial) Ruzsa distance measures how much larger the sum/difference set \(A+B\) is than the geometric mean of the sizes of \(A\) and \(B\), so small Ruzsa distance means that \(A\) and \(B\) share strong additive structure. Formally, the Ruzsa distance is defined as
\[
    \textsf{d}_{\rm comb}(A,B)
    =
    \log\frac{|A+B|}{\sqrt{|A||B|}}.
\]
Although this is not a true distance, it satisfies non-negativity, symmetry and the triangle inequality.

\subsection{Information theory}
\paragraph{Entropy.} 
The entropy $\entropy[X]$ of a $G$-valued random variable is defined as
\begin{equation}\label{def:entropy}
    \entropy[X] := \sum_{x \in G} p_X(x)
    \log\frac{1}{p_X(x)},
\end{equation}
where $\log$ is base $2$, $p_X(x)$ is the probability distribution
function of $X$.

The \emph{conditional} entropy $\entropy[X|Y]$ of $X$ relative to $Y$ is given by
\begin{equation}\label{def:conditional_entropy}
\entropy[X|Y] := \sum_{y} p_Y(y) \entropy[X|Y=y],
\end{equation}
where $y$ ranges over the support of $p_Y$. We quickly note that the chain rules of 
\begin{equation}\label{eq:entropy_chain_rules}
\entropy[X,Y] = \entropy[X|Y] + \entropy[Y], \quad \entropy[X,Y|Z] = \entropy[X|Y,Z] + \entropy[Y|Z].
\end{equation}

Additionally, we have the following facts regarding entropy.
\begin{fact}\label{fact:ub_entropy_nonneg_int_RV}
Suppose $Z$ is a nonnegative integer-valued random variable with mean $\mu$. The geometric distribution then maximizes entropy of $Z$ and
$$
\entropy[Z] \leq (\mu + 1) \log (\mu + 1) - \mu \log \mu.
$$
\end{fact}

\begin{lemma}
\label{lem:collision-at-most-shannon}
For every discrete random variable $Z$,
$\entropy[Z]\geq -\log\left(\sum_z p_Z(z)^2\right)$. Equivalently,
\[
\sum_z p_Z(z)^2\geq 2^{-\entropy[Z]}.
\]
Moreover, for every $\varepsilon>0$, there is a set $T\subseteq\supp(Z)$ such that $\Pr[Z\in T]\geq1-\varepsilon$ and $|T|\leq2^{\entropy[Z]/\varepsilon}$.
\end{lemma}
\begin{proof}
Since $-\log$ is convex, Jensen's inequality applied to the random
variable $p_Z(Z)$ gives
\[
\entropy[Z]
=\E\left[\log\frac1{p_Z(Z)}\right]
\geq -\log\E[p_Z(Z)]
=-\log\left(\sum_z p_Z(z)^2\right).
\]
Exponentiating proves the equivalent formulation.  Finally, apply Markov's inequality to $\log(1/p_Z(Z))$ and take
$T=\{z:\log(1/p_Z(z))\leq\entropy[Z]/\varepsilon\}$.
\end{proof}

\paragraph{Mutual Information.}
The mutual information $\MI[X:Y]$ is defined as
\begin{equation}\label{def:mutual_info}
\MI[X:Y] := \entropy[X] + \entropy[Y] - \entropy[X,Y].
\end{equation}

\subsection{Entropic Ruzsa distance}
The \emph{entropic Ruzsa distance} between two independent $G$-valued random variables $X,Y$, denoted by $\ruzsadist[X;Y]$, is defined as
\begin{equation}\label{def:ruzsadist}
\ruzsadist[X;Y] := \entropy[X' + Y'] - \frac{1}{2} \entropy[X'] - \frac{1}{2} \entropy[Y'],
\end{equation}
where $X',Y'$ are independent copies of $X,Y$
respectively. Despite it being called a distance, the Ruzsa distance is not a true notion of distance as $\ruzsadist[X;X]=0$ only when $X$ is a coset of a subgroup~\cite{ggmt2025conjecture}. It however satisfies the triangle inequality, as stated below.
\begin{lemma}[Entropic Ruzsa triangle inequality]
\label{lem:ruzsa-triangle}
For $G$-valued random variables $X,Y,Z$,
\begin{equation}\label{eq:ruzsa_triangle}
    \ruzsadist[X;Z]
    \leq \ruzsadist[X;Y]+\ruzsadist[Y;Z]. 
\end{equation}
\end{lemma}

\begin{lemma}[Basic entropy and Ruzsa-distance facts]
\label{lem:basic-entropic-facts}
Let $X,Y$ be random variables on $\FF_2^n$.
\begin{enumerate}
    \item The entropy-difference bound is
    \[
        |\entropy[X]-\entropy[Y]|
        \leq 2\ruzsadist[X;Y].
    \]
    \item If $\varphi$ is a homomorphism, then entropic Ruzsa distance
    contracts under $\varphi$:
    \[
        \ruzsadist[\varphi(X);\varphi(Y)]
        \leq \ruzsadist[X;Y].
    \]
    \item If $V\leq W\leq\FF_2^n$ are subspaces, then
    \[
        \ruzsadist[U_V;U_W]
        =\frac12\log\frac{|W|}{|V|}.
    \]
\end{enumerate}
\end{lemma}
\begin{proof}
Since conditioning cannot increase entropy,
$\entropy[X+Y]\geq\max\{\entropy[X],\entropy[Y]\}$, which proves the
first assertion by expanding the definition of Ruzsa distance.  The
identity
\[
 2\ruzsadist[X;Y]
 =\MI[X:X+Y]+\MI[Y:X+Y]
\]
and data processing prove the second.  For the third, observe that
$U_V+U_W$ is uniform on $W$.
\end{proof}

\begin{theorem}[Entropic PFR theorem \cite{ggmt2025conjecture}]
\label{thm:entropic-pfr}
For every random variable $X$ on $\FF_2^n$, there is a subspace
$H\leq\FF_2^n$ such that
\[
\ruzsadist[X;U_H]=O(\ruzsadist[X;X]).
\]
\end{theorem}

\paragraph{Conditional entropic Ruzsa distance.}
We will also use the conditional variant of the entropic Ruzsa distance. If $X,Y,Z,W$ are $G$-valued random variables, we define
\begin{equation}\label{def:conditional_der}
\der{X\mid Z}{Y\mid W} := \E_{z\sim Z,w\sim W} \der{(X\mid Z=z)}{(Y\mid W=w)},
\end{equation}
If both $X,Y$ are conditioned on the same variable, we write $\der{X}{Y\mid Z} = \E_{z\sim Z}
\der{X}{(Y\mid Z=z)}$. We also have the following equivalent expression considering $X',Y',Z',W'$ as independent copies of $X,Y,Z,W$ respectively~\cite{ggmt2025conjecture}
\begin{equation}\label{def:equiv_conditional_der}
\der{X\mid Z}{Y\mid W} := \entropy[X' + Y' \mid Z', W'] - \frac{1}{2} \entropy[X' \mid Z'] - \frac{1}{2} \entropy[Y' \mid W']
\end{equation}

\begin{lemma}
\label{lem:ruzsa-conditioning}
Let $(X,J)$ be jointly distributed, with $X$ taking values in
$\FF_2^n$, and let $Y$ be an $\FF_2^n$-valued random variable
independent of $(X,J)$.  Then
\begin{align*}
 \E_{j\sim J}\ruzsadist[\law(X\mid J=j);Y]
 \leq\ruzsadist[X;Y]+\frac12\MI[X:J].
\end{align*}
In particular, if $J$ takes at most $m$ values, then
\[
 \E_{j\sim J}\ruzsadist[\law(X\mid J=j);Y]
 \leq\ruzsadist[X;Y]+\frac12\log m.
\]
\end{lemma}
\begin{proof}
Using an independent copy of $Y$ in every conditional Ruzsa distance,
the left-hand side equals
\begin{align*}
 \entropy[X+Y\mid J]
 -\frac12\entropy[X\mid J]-\frac12\entropy[Y]
 &=\entropy[X+Y]-\MI[X+Y:J]\\
 &\quad-\frac12\bigl(\entropy[X]-\MI[X:J]\bigr)
       -\frac12\entropy[Y]\\
 &=\ruzsadist[X;Y]-\MI[X+Y:J]+\frac12\MI[X:J].
\end{align*}
Dropping the nonnegative mutual-information term proves the
inequality.  The final assertion follows from
$\MI[X:J]\leq\entropy[J]\leq\log m$.
\end{proof}

\subsection{Boolean Fourier analysis and discrete derivatives}
\label{sec:boolean-fourier}
For $x,y\in\FF_2^n$, let
$\langle x,y\rangle:=\sum_{i=1}^n x_i y_i\in\FF_2$.
All expectations below are uniform over the indicated vector space. 

\paragraph{Fourier analysis.}
For
$g:\FF_2^n\to\RR$, define its normalized Fourier transform by
\[
 \widehat g(\xi):=\E_y g(y)(-1)^{\langle \xi,y\rangle},
 \qquad \xi\in\FF_2^n.
\]
We call $g$ Boolean when $g:\FF_2^n\to\{-1,1\}$.
Fourier inversion and Parseval's identity give
\[
 g(y)=\sum_{\xi\in\FF_2^n}\widehat g(\xi)
       (-1)^{\langle\xi,y\rangle},
 \qquad
 \E_y g(y)h(y)=\sum_{\xi\in\FF_2^n}
                 \widehat g(\xi)\widehat h(\xi).
\]
In particular,
$\sum_\xi\widehat g(\xi)^2=\E_y g(y)^2$, which equals $1$ when
$g$ is Boolean.  We will also use the autocorrelation identity
\[
 \E_y g(y)g(y+x)
 =\sum_{\xi\in\FF_2^n}\widehat g(\xi)^2
   (-1)^{\langle\xi,x\rangle}.
\]
Consequently, if $H:\FF_2^n\to\RR$ has nonnegative Fourier
coefficients, then Fourier inversion implies
\[
 H(w)\leq H(0)\qquad(w\in\FF_2^n).
\]

\begin{theorem}[Goldreich--Levin~\cite{goldreich1989hard}]
\label{thm:goldreich-levin}
Let $g:\FF_2^n\to\{-1,1\}$ and let $0<\gamma,\delta\leq1$. There is a
randomized algorithm which, given oracle access to $g$, outputs a list of
$O(\gamma^{-2})$ frequencies such that, with probability at least
$1-\delta$, the list contains every $\xi\in\FF_2^n$ satisfying
$|\widehat g(\xi)|\geq\gamma$. The algorithm uses
$\poly(n,1/\gamma,\log(1/\delta))$ time and queries to $g$.
\end{theorem}

\paragraph{Discrete derivatives.}
For $f:\FF_2^n\to\RR$ and $x\in\FF_2^n$, the multiplicative discrete
derivative of $f$ in direction $x$ is
\[
 f_x(y):=f(y)f(y+x).
\]

\paragraph{The Gowers $U^3$ norm.} The $U^3$-norm is defined by
\[
 \|f\|_{U^3}^8
 :=\E_{y,x_1,x_2,x_3}
 \prod_{\omega\in\{0,1\}^3}
 f(y+\omega_1x_1+\omega_2x_2+\omega_3x_3).
\]
Repeated use of Parseval's identity gives
\[
 \|f\|_{U^3}^8
 =\E_x\sum_{\xi\in\FF_2^n}\widehat{f_x}(\xi)^4.
\]

\paragraph{Helper lemmas.}
The following observation is adapted from
Tulsiani--Wolf~\cite[Claim~4.15]{tulsiani2014quadratic}.
\begin{claim}[Removing an affine offset]
\label{clm:removing-affine-offset}
For every $f:\FF_2^n\to\RR$, linear map
$T:\FF_2^n\to\FF_2^n$, and $b\in\FF_2^n$,
\[
 \E_x\left|\widehat{f_x}(Tx+b)\right|^2
 \leq
 \E_x\left|\widehat{f_x}(Tx)\right|^2.
\]
\end{claim}

\begin{proof}
Define
\[
 F_T(z):=\E_x\left|\widehat{f_x}(Tx+z)\right|^2.
\]
For $t\in\FF_2^n$, the normalized Fourier coefficient of $F_T$ at $t$
is
\begin{align*}
 \widehat{F_T}(t)
 &=\E_zF_T(z)(-1)^{\langle t,z\rangle} \\
 &=\E_{z,x,y,y'}f(y)f(y+x)f(y')f(y'+x)
   (-1)^{\langle Tx+z,y+y'\rangle+\langle t,z\rangle} \\
 &=2^{-n}\E_{x,y}f(y)f(y+x)f(y+t)f(y+t+x)
   (-1)^{\langle Tx,t\rangle} \\
 &=2^{-n}\E_{x,y}f_t(y)f_t(y+x)
   (-1)^{\langle T^{\mathsf T}t,x\rangle} \\
 &=2^{-n}\left|\widehat{f_t}(T^{\mathsf T}t)\right|^2
 \geq0.
\end{align*}
Here the third equality follows because averaging over $z$ forces
$y+y'=t$, and the last equality follows by expanding the square of
$\widehat{f_t}(T^{\mathsf T}t)$. Thus every Fourier coefficient of
$F_T$ is nonnegative. Fourier inversion therefore gives
\[
 F_T(b)
 =\sum_{t\in\FF_2^n}\widehat{F_T}(t)(-1)^{\langle t,b\rangle}
 \leq\sum_{t\in\FF_2^n}\widehat{F_T}(t)
 =F_T(0),
\]
which is the claimed inequality.
\end{proof}

\begin{claim}
  \label{clm:qgl-autocorrelation-identity}
  For a function $g: \FF_2^n \to \{-1, 1\}$, we have that
  \begin{align*}
    \E_x\left|\E_y g(y)g(y+x)\right|^2 = \sum_{\xi \in \FF_2^n} \widehat g(\xi)^4.
  \end{align*}
\end{claim}
\begin{proof}
  First, by Fourier inversion, we have that
  \begin{align*}
    \E_y g(y)g(y+x) &= \E_y \sum_{\xi, \xi' \in \FF_2^n} \widehat g(\xi)\widehat g(\xi')(-1)^{\langle y, \xi + \xi' \rangle + \langle x, \xi' \rangle}\\
    &= \sum_{\xi, \xi' \in \FF_2^n} \widehat g(\xi)\widehat g(\xi')(-1)^{\langle x, \xi' \rangle}\E_y (-1)^{\langle y, \xi + \xi' \rangle} \\
    &= \sum_{\xi\in\FF_2^n} \widehat g(\xi)^2 (-1)^{\langle x, \xi\rangle},
  \end{align*}
  where the final equality follows by the orthogonality of characters.
  Then, we have that
  \begin{align*}
    \E_x\left|\E_y g(y)g(y+x)\right|^2 &= \E_x\left(\sum_{\xi\in\FF_2^n} \widehat g(\xi)^2 (-1)^{\langle x, \xi\rangle}\right)^2 &\\
    &= \E_x\sum_{\xi, \xi' \in \FF_2^n} \widehat g(\xi)^2\widehat g(\xi')^2 (-1)^{\langle x, \xi + \xi' \rangle} \\
    &= \sum_{\xi \in \FF_2^n} \widehat g(\xi)^4,
  \end{align*}
  where again the final equality is by the orthogonality of characters.
\end{proof}

\section{Traversing the GGMT tree}
\label{sec:ggmt-recursion}

Throughout this section, we will need access to random variables for the algorithm and we formally define this in the definition below. 
\begin{definition}[Algorithmic access to a random variable]\label{def:algo_access}
We say that we have \emph{algorithmic access} to a random variable $X$ if we have the following pair of algorithms %
\begin{itemize}
    \item $\Sample_X$ outputs an independent sample $x \sim X$ and runs in time expected time $\poly(n, K)$, and
    \item $\Coin_X$, on input $x$, runs in expected time $\poly(n, K)$ and outputs an independent sample from the Bernoulli distribution with parameter 
    \begin{align*}
        r_X(x):= \frac{p_X(x)}{M_X},
    \end{align*}
    where $M_X\geq\|p_X\|_\infty$ is referred to as the \emph{envelope}.\footnote{Note that the envelope is determined by the particular implementation of the density coin, rather than being an intrinsic property of the law of $X$.} The numerical value of $M_X$ need not be known to the algorithm, even approximately.
\end{itemize}
\end{definition}

The input oracles for $A$ give us algorithmic access to the uniform distribution $U_A$ on the set $A$.
The given uniform sampler for $A$ implements
$\Sample_{U_A}$, while the membership oracle implements the density~coin
\[
    \Coin_{U_A}(x)=\mathbf{1}_A(x).
\]
Indeed, taking $M_{U_A}:=1/|A|$ gives
\[
    \Pr[\Coin_{U_A}(x)=1]
    =\mathbf{1}_A(x)
    =\frac{p_{U_A}(x)}{M_{U_A}}.
\]
Neither root algorithm aborts, and the value of $M_{U_A}$ need not be
known to the algorithm. Throughout this section, by algorithmic access to $X$, we mean constructing algorithms for $\Sample_{X}$ and $\Coin_{X}$. We define the \textbf{\textit{slack}} of a density coin $\Coin_X$, denoted by $\Delta[X]$, as
\begin{equation}\label{def:slack_coin}
\Delta[X] := \entropy[X] + \log M_X.
\end{equation} 
For the root algorithmic access above, $\Delta[U_A]=0$. Intuitively, the slack measures how small the density coin's acceptance probability is on a typical sample \(x\sim X\).  As we will see later in this section, a small slack means that the coin can typically be used efficiently, while large slack indicates that rejection-based procedures may require many trials. The main result of this section is the following.
\begin{theorem}
\label{thm:ggmt-recursion}
Let $K \in \mathbb{N}$. Suppose $A \subseteq\FF_2^n$ is non-empty set that satisfies $|A+A|\le K|A|$. Given membership and uniform-sampling access to $A$, there is a randomized algorithm that, with probability at least $\Omega(1/\polylog(K))$, outputs algorithmic access to a random variable $X_T$, where $T=O(\log\log K)$, such that, for some subspace $H \leq \FF_2^n$,
$$
\ruzsadist[X_T;U_H]=O(1),\quad
\ruzsadist[U_A;X_T]=O(\log K),\quad \text{ and }
\Delta[X_T]=O(1).
$$
The algorithm has overall sample and query complexity $\poly(K)$ and runs in time $O(n \cdot \poly(K) )$. The cost of one call to the algorithmic access procedures of $X_T$ has sample and query complexity $\poly(K)$ and time complexity $O(n \cdot \poly(K))$ as well.
\end{theorem}
We prove this theorem through the section. In Section~\ref{sec:ggmt-five}, we explain the five operations involved in the GGMT tree. In Section~\ref{sec:algoaccesstosumfiber}, we give algorithmic implementations for these operations. In Section~\ref{sec:bucketing}, we introduce bucketting to regularize the algorithmic access of the variables as we go through the GGMT tree and ensure that the slack of the variables is bounded. Finally in Section~\ref{sec:samplingtrajectory}, we show how one can probabilistically sample a path in the GGMT tree to find the random variable $X_T$ satisfying Theorem~\ref{thm:ggmt-recursion}.

\subsection{The five GGMT operations}\label{sec:ggmt-five}
Let $X$ and $Y$ be random variables over $\FF_2^n$.  Let
$X_0,X_1$ be independent copies of $X$ and let $Y_0,Y_1$ be independent
copies of $Y$, with all four variables mutually independent.  The five
GGMT operations applied to $(X,Y)$ are defined as follows: each operation should be thought of as taking a pair of distributions/random variables $(X,Y)$ and producing a new pair of random variables (equivalently, their distribution laws).
\begin{enumerate}
    \item \emph{Self-sum:}
    \[
        \mathsf{SelfSum}(X,Y)
        :=(X_0+X_1,\,Y_0+Y_1).
    \]

    \item \emph{Cross-sum:}
    \[
        \mathsf{CrossSum}(X,Y)
        :=(X_0+Y_0,\,X_1+Y_1).
    \]

    \item \emph{Self-fiber:}
    First sample
    \[
        x\sim \law(X_0+X_1),
        \qquad
        y\sim \law(Y_0+Y_1)
    \]
    independently.  We then restrict to the corresponding fibers of the
    addition map and define
    \[
        \mathsf{SelfFiber}(X,Y)
        :=
        \left(
            \law(X_0\mid X_0+X_1=x),
            \law(Y_0\mid Y_0+Y_1=y)
        \right).
    \]
Intuitively, the first argument is the distribution over a random variable defined as follows: among all pairs of $X_0,X_1$ whose sum equals $x$, how is the first argument distributed.
    \item \emph{Cross-fiber:}
    First sample
    \[
        z_0\sim \law(X_0+Y_0),
        \qquad
        z_1\sim \law(X_1+Y_1)
    \]
    independently.  We then define
    \[
        \mathsf{CrossFiber}(X,Y)
        :=
        \left(
            \law(X_0\mid X_0+Y_0=z_0),
            \law(Y_1\mid Y_1+X_1=z_1)
        \right).
    \]
    Here each fiber is obtained by conditioning on the sum of one copy of
    $X$ and one copy of $Y$.

    \item \emph{Endgame:}
    Define
    \[
        Z:=X_1+Y_0,
        \qquad
        W:=X_0+X_1+Y_0+Y_1.
    \]
    Sample $(z,w)$ according to the joint distribution of $(Z,W)$, and let
    \[
        U_{z,w}
        :=
        \law(X_1+Y_1\mid Z=z,\,W=w).
    \]
    We then set
    \[
        \mathsf{Endgame}(X,Y):=(U_{z,w},U_{z,w}).
    \]
\end{enumerate}

We then have the following promise from~\cite{ggmt2024torsion} which we have written for the special case of $\FF_2^n$.
\begin{theorem}\label{thm:ggmt-decrement}
There is a universal constant $\eta>0$ such that
at least one of these five raw candidate families, with $(X',Y')$ drawn
according to its natural law, satisfies
\[
\E\!\left[
    \der{X'}{Y'}
    +\eta\cdot \der{X}{X'}+\eta \cdot \der{Y}{Y'}
\right]
\leq (1-\eta)\der{X}{Y}.
\]
For the self-sum and cross-sum families the output laws are
deterministic, while for the other three families the expectation is over
their fiber or endgame labels.  We use this constant $\eta$
for the remainder of the section.
\end{theorem}

\subsection{Algorithmic access to candidate families}
\label{sec:algoaccesstosumfiber}
We now describe how the algorithmic access  to $(X_{t+1},Y_{t+1})$ is constructed given algorithmic access to $(X_t,Y_t)$ and corresponding to each of the possible candidate decisions as described in Section~\ref{sec:ggmt-five}.

\subsubsection{Sums}
We first give a procedure that given access to random variables $X,Y$ allows to perform the sum operation (the first two components of the GGMT operations).
\begin{lemma}\label{lem:closure-sum}
Suppose we have algorithmic access to $X,Y$. Then, there is a protocol
that gives algorithmic access to $X+Y$  with envelope $M_Y$ and slack satisfying
$$
\Delta[X+Y] \leq \Delta[Y] + 2 \ruzsadist[X; Y].
$$
\end{lemma}
\begin{proof}
We construct algorithmic access to $X+Y$ as follows:
\begin{enumerate}[$(i)$]
\item $\Sample_{X+Y}$ runs $x \gets \Sample_X, y \gets \Sample_Y$ independently and then outputs $x+y$.
\item $\Coin_{X+Y}(z)$ samples $x \gets \Sample_X$ and outputs $\Coin_Y(z-x)$.
\end{enumerate}
The sampler is correct by definition. We now note that
\begin{align*}
p_{X+Y}(z) = \sum_x p_X(x)p_Y(z-x) \leq \sum_x p_X(x)M_Y = M_Y, \quad \forall z \in \FF_2^n,
\end{align*}
which implies that $M_Y$ is an upper bound on $p_{X+Y}(z), \forall z$ and is a valid envelope for $X+Y$.
We also note that $\Coin_{X+Y}$ gives us the desired output by evaluating
\begin{align*}
\Pr[\Coin_{X+Y}(z) = 1]
= \frac{1}{M_Y}\sum_{x\in\FF_2^n}p_X(x)p_Y(z-x)
= \frac{p_{X+Y}(z)}{M_Y}.
\end{align*}
where the equality follows from 
\[
    p_{X+Y}(z)
    = \Pr[X+Y=z]
    = \sum_x \Pr[X=x,\,Y=z-x]
    =  \sum_x p_X(x)p_Y(z-x),
\]
which used the independence of $X,Y$. Finally, we compute the slack:
\begin{equation}\label{eq:temp-slack-sum}
\Delta[X+Y] = \entropy[X+Y] + \log M_{X+Y} = \entropy[X+Y] + \log M_Y = \Delta[Y] + \entropy[X+Y] - \entropy[Y],    
\end{equation}
where we used the definition of $\Delta[Y]$ in the last equality. Since $X$ and $Y$ are independent, 
\begin{align}\label{eqn:sum-slack-bound}
\entropy[X+Y] \geq \entropy[X+Y \mid Y] = \entropy[X].
\end{align}
From the definition of the entropic Ruzsa distance, we have
\begin{equation}\label{eq:interim_dist_R}
2 \ruzsadist[X;Y] = 2 \entropy[X+Y] - \entropy[X] - \entropy[Y] \geq \entropy[X+Y] - \entropy[Y],
\end{equation}
where we used Eq.~\eqref{eqn:sum-slack-bound} in the last inequality. Substituting the above into Eq.~\eqref{eq:temp-slack-sum} gives us
$$
\Delta[X+Y] \leq \Delta[Y] + 2 \ruzsadist[X;Y],
$$
which proves the lemma statement.
\end{proof}

\subsubsection{Fibers} 
We next give a procedure that given access to random variables $X,Y$ allows to perform the fibering operation (the third and fourth components of the GGMT operations).
\begin{lemma}\label{lem:closure-fiber}
Suppose we have algorithmic access to independent $X,Y$. Let $Z := X + Y$. Then, for every fixed $z \in \supp(Z)$, there is a protocol that constructs algorithmic access to $X_z:=\law(X\mid Z=z)$, with envelope $M_{X_z}:=M_X M_Y/p_Z(z)$. A sample from $\Sample_{X_z}$ is accepted with probability $p_Z(z)/M_Y$, so the expected number of times we run the sampler is $M_Y/p_Z(z)$.

Further, if the label is drawn as $z \sim X + Y$, then
\begin{align}
    \E_{z \sim Z}\Big[\Delta(X_z)\Big] &= \Delta[X] + \Delta[Y], \label{eq:fiber_slack} \\
    \E_{z \sim Z}\left[\log\left(\frac{M_Y}{p_Z(z)}\right)\right] &= \entropy[X+Y] + \log M_Y \leq \Delta[Y] + 2 \ruzsadist[X; Y]. \label{eq:fiber_cost}
\end{align}
\end{lemma}
\begin{proof}
For a fixed $z \in \supp(Z)$, we note that Baye's rule gives us
\begin{equation}\label{eq:prob_z}
p_{X_z}(x) = \Pr[X = x | X + Y = z] = \frac{p_X(x)p_Y(z-x)}{p_Z(z)}.
\end{equation}
We construct $\Sample_{X_z}$ using the following rejection sampler:
\begin{enumerate}[$(i)$]
    \item Sample $x \gets \Sample_X$.
    \item Run $w \gets \Coin_Y(z-x)$.
    \item If $w=1$, then return $x$ \textbf{else} repeat from $(i)$.
\end{enumerate}
For $\Sample_{X_z}$, the probability of outputting $x$ (i.e., sampling and accepting $x$) on a single trial is
$p_X(x)p_Y(z-x)/M_Y$. The total acceptance probability is then
$$
\sum_{x} p_X(x) \Pr[\Coin_Y(z-x)=1] = \sum_{x} p_X(x) \frac{p_Y(z-x)}{M_Y} = \frac{p_Z(z)}{M_Y}.
$$
Since we repeat trials in $\Sample_{X_z}$ until $\Coin_Y$ accepts, we have that conditioned on acceptance
$$
\Pr[\Sample_{X_z} \text{ returns } x] = \frac{p_X(x)p_Y(z-x)/M_Y}{p_Z(z)/M_Y} =\frac{p_X(x)p_Y(z-x)}{p_Z(z)},
$$
which matches the density of $X_z$ from Eq.~\eqref{eq:prob_z} as desired. This verifies the correctness of the above protocol for $\Sample_{X_z}$.

We construct the density coin $\Coin_{X_z}$ on input $x$ as follows:
\begin{enumerate}[$(i)$]
    \item Run $w_1 \gets \Coin_X(x)$ and $w_2 \gets \Coin_Y(z-x)$.
    \item Output $w_1 w_2$.
\end{enumerate}
The coin $\Coin_{X_z}(x)$ thus only accepts exactly when both density coins in $(i)$ succeed. In particular, $\Coin_{X_z}$ accepts with probability
$$
\Pr[\Coin_{X_z}(x) = 1] = \frac{p_X(x)}{M_X}\frac{p_Y(z-x)}{M_Y}
 =\frac{p_{X_z}(x)}{M_XM_Y/p_Z(z)},
$$
where we used the expression of $p_{X_z}(x)$ from Eq.~\eqref{eq:prob_z} in the seond equality. From the above expression, we then also have that the envelope is $M_{X_z} = M_X M_Y/p_Z(z)$.

We can show that Eq.~\eqref{eq:fiber_slack} holds from direct evaluation:
\begin{align*}
\E_{z \sim Z} \Delta[X_z] &= \E_{z\sim Z} \left[\entropy[X_z] + \log M_X + \log M_Y - \log p_Z(z)\right] \\
&= \entropy[X \mid Z] + \log M_X + \log M_Y + \entropy[Z] \\
&= \entropy[X, Z] + \log M_X + \log M_Y \\
&= \entropy[X] + \entropy[Y] + \log M_X + \log M_Y,
\end{align*}
where we used the definition of $\entropy[X|Z],\entropy[Z]$ in the second line, definition of $\entropy[X,Z]$ in the third line and then noted that $Z = X + Y$ along with the fact that $X,Y$ are independent random variables in the last line (i.e., $\entropy[X,Z] = \entropy[X,Y] = \entropy[X] + \entropy[Y]$). Finally, Eq.~\eqref{eq:fiber_cost} follows 
$$
\E_{z \sim Z}[\log(M_Y/p_Z(z))] = \entropy[Z] + \log M_Y \leq 2 \ruzsadist[X;Y] + \entropy[Y] + \log M_Y = 2\ruzsadist[X;Y] + \Delta[Y],
$$
where we used Eq.~\eqref{eq:interim_dist_R} in the second inequality after noting that $Z=X+Y$ by definition and the definition of $\Delta[Y]$ in the last inequality. This completes the proof.
\end{proof}

\subsubsection{Endgame}
Finally it remains to describe the algorithmic component of the final GGMT operation, i.e., the so-called endgame operation. We describe how to do that  below. 
\begin{lemma}\label{lem:access_endgame}
Suppose we have algorithmic access to $X,Y$.  Let $X_0,X_1$ be
independent copies of $X$, let $Y_0,Y_1$ be independent copies of $Y$,
and suppose all four variables are mutually independent.  Define
\[
Z:=X_1+Y_0,\qquad S:=X_0+Y_1,\qquad W:=Z+S.
\]
For a fixed $(z,w)\in\supp(Z,W)$, let $s:=z+w$ and
\[
U_{z,w}:=\law(X_1+Y_1\mid Z=z,W=w).
\]
Then there is a protocol that constructs exact algorithmic access to
$U_{z,w}$ with envelope $M_{U_{z,w}}=M_XM_Y/p_{X+Y}(s)$. The construction uses two fiber rejection samplers: one whose trial acceptance probability is $p_{X+Y}(z)/M_Y$, and one whose trial acceptance
probability is $p_{X+Y}(s)/M_X$. Moreover, when $(z,w)$ is drawn from
the natural law of $(Z,W)$,
\begin{align}
\E_{(z,w)\sim(Z,W)}\Delta[U_{z,w}]
&\leq \Delta[X]+\Delta[Y]+4\ruzsadist[X;Y],
\label{eq:endgame_slack}\\
\E_{(z,w)\sim(Z,W)}\left[
\log\frac{M_Y}{p_{X+Y}(z)}
+\log\frac{M_X}{p_{X+Y}(w+z)}
\right]
&=\Delta[X]+\Delta[Y]+2\ruzsadist[X;Y].
\label{eq:endgame_cost}
\end{align}
\end{lemma}

\begin{proof}
Consider a fixed $(z,w) \in \supp(Z,W)$ and set $s:=w+z$. We then define
$$
F_z:=\law(X_1\mid X_1+Y_0=z), \qquad G_s:=\law(Y_1\mid X_0+Y_1=s).
$$
Since the pairs $(X_1,Y_0)$ and $(X_0,Y_1)$ are independent, $Z$ and $S$ are independent. Moreover, the condition $Z=z$ and $W=w$ (or $Z+S=z+s$) is then equivalent to $Z=z$ and $S=s$.  Consequently, $U_{z,w} := (X_1 + Y_1 | Z=z, W=w)$ can be equivalently written as
$$
U_{z,w}=F_z+G_s,
$$
where the two summands on the right are independent. To implement algorithmic access to $U_{z,w}$, we first need algorithmic access to $F_z$ and $G_s$. We implement algorithmic access to $F_z$ using Lemma~\ref{lem:closure-fiber} and denote the corresponding algorithms as $\Sample_{F_z}$ and $\Coin_{F_z}$. We note that a single trial from $\Sample_{F_z}$ is accepted with probability $p_{X+Y}(z)/M_Y$ and the envelope is $M_{F_z} = M_X M_Y/p_{X+Y}(z)$. Similarly, we implement algorithmic access to $G_s$ again using Lemma~\ref{lem:closure-fiber} and denote the corresponding algorithms as $\Sample_{G_s}$ and $\Coin_{G_s}$. A single trial from $\Sample_{G_s}$ is accepted with probability $p_{X+Y}(s)/M_X$ and the envelope is $M_{G_s} = M_X M_Y/p_{X+Y}(s)$. Algorithmic access to $U_{z,w}$ is then implemented using Lemma~\ref{lem:closure-sum} with $F_z$ and $G_s$. The corresponding envelope is simply
$$
M_{U_{z,w}} = M_{G_s} = M_X M_Y/p_{X+Y}(s).
$$
It remains to prove the averaged bounds on the slack of $U_{z.w}$ and the acceptance probabilities involved as part of the construction of algorithmic access to $U_{z,w}$. Upon direct evaluation:
\begin{align}
\E_{(z,w) \sim (Z,W)} \Delta [U_{z,w}] = \E_{(z,s) \sim (Z,S)} \Delta[F_z + G_s] 
&\leq \E_{(z,s) \sim (Z,S)} \Delta[G_s] + 2 \E_{(z,s) \sim (Z,S)} \der{F_z}{G_s} \nonumber \\
&\leq \E_{s \sim S} \Delta[G_s] + 2 \E_{(z,s) \sim (Z,S)} \der{F_z}{G_s} \nonumber \\
&\leq \Delta[X] + \Delta[Y] + 2 \E_{(z,s) \sim (Z,S)} \der{F_z}{G_s}, \label{eq:interim1_slack_Uzw}
\end{align}
where we have used the earlier observation $Z=z,W=w$ is equivalent to $Z=z,S=s$ in the first equality in the first line, and used Lemma~\ref{lem:closure-sum} in the second equality in the first line to comment $\Delta[F_z + G_s] \leq \Delta[G_s] + 2 \der{F_z}{G_s}$. In the second line, we used the fact that $Z$ and $S$ are independent and both have the law of $X+Y$. In the final line, we used Lemma~\ref{lem:closure-fiber} from which we have
$$
\E_{s \sim S} \Big[ \Delta[G_s] \Big] = \Delta[X]+\Delta[Y]
$$
We now bound the last term on the right in Eq.~\eqref{eq:interim1_slack_Uzw}. Let us define $h:=\entropy[X]+\entropy[Y]-\entropy[X+Y]$ and note that 
\begin{equation}\label{eq:int2_cond_H_X1_Z}
\entropy[X_1\mid Z] = \entropy[X_1,Z] - \entropy[Z] = \entropy[X_1,Y_0] - \entropy[Z] = \entropy[X] + \entropy[Y] - \entropy[X+Y] = h,    
\end{equation}
where we have used Eq.~\eqref{eq:entropy_chain_rules} in the first equality, noted that the $Y_0 = Z + X_1$ and hence they determine each other in the second equality, and used the fact that $X_1,Y_0$ are independent along with the definition of $Z$ in the final equality. Similarly, we can show that
\begin{equation}\label{eq:int2_cond_H_Y1_S}
\entropy[Y_1\mid S] = \entropy[Y_1,S] - \entropy[S] = \entropy[Y_1,X_0] - \entropy[S] = \entropy[X] + \entropy[Y] - \entropy[X+Y] = h.
\end{equation}
Using the definition of conditional entropic Ruzsa distance (Eq.~\eqref{def:conditional_der}) and its equivalent expression (Eq.~\eqref{def:equiv_conditional_der}), we have that
\begin{align}\label{eq:bound_der_Fz_Gs}
\E_{z,s}\ruzsadist[F_z;G_s] = \entropy[X_1 + Y_1 \mid Z, S] - \frac{1}{2} \entropy[X_1 \mid Z] - \frac{1}{2} \entropy[Y_1 \mid S] \leq \entropy[X + Y] - h = 2\ruzsadist[X;Y]
\end{align}
where we have used the fact that conditioning cannot increase entropy along with Eqs~\eqref{eq:int2_cond_H_X1_Z}--\eqref{eq:int2_cond_H_Y1_S} in the first equality, and the definition of $\der{X}{Y}$ in the last equality. Substituting Eq.~\eqref{eq:bound_der_Fz_Gs} into Eq.~\eqref{eq:interim1_slack_Uzw} gives us the desired result of Eq.~\eqref{eq:endgame_slack}.

Finally, to obtain Eq.~\eqref{eq:endgame_cost}, we observe upon direct evaluation that
$$
\E_{(z,s) \sim (Z,S) }\left[ \log\frac{M_Y}{p_{X+Y}(z)} + \log\frac{M_X}{p_{X+Y}(s)} \right]
= \E_{z \sim Z} \log\frac{M_Y}{p_{X+Y}(z)} + \E_{s \sim S} \log\frac{M_X}{p_{X+Y}(s)}
\leq \Delta[X]+\Delta[Y]+2\ruzsadist[X;Y],
$$
where the first equality follows from $Z,S$ being independent and the last inequality follows from application of Eq.~\eqref{eq:fiber_cost} of Lemma~\ref{lem:closure-fiber}. This completes the proof.
\end{proof}

\subsection{Bucketing to regularize access}
\label{sec:bucketing}
So far, we showed that algorithmic access to all the intermediate random variables used in GGMT operations is possible given access to random variables in the iterative process.   However, we also need to 
control the quantitative \emph{quality} of this algorithmic representation as the recursion
proceeds. To this end, recall that the \emph{slack} of our algorithmic access to a random variable $X$ is 
\[
\Delta[X] := \entropy[X] + \log M_X = \E_{x\sim X} \left[ \log \frac{1}{r_X(x)} \right],
\qquad
r_X(x) := \Pr[\mathsf{Coin}_X(x)=1] = \frac{p_X(x)}{M_X}.
\]
Intuitively, slack measures the typical inverse acceptance probability of the density coin, i.e., measures how small the density-coin acceptance probability
typically is on a sample $x\sim X$. So $\Delta[X]=O(1)$ means that, for a typical sample from $X$, the coin
acceptance probability $r_X(x)$ is not exponentially small. This is exactly
the quantity relevant to the rejection-sampling procedures used later in the
recursion.   Since fiber sampling is implemented by rejection sampling using that coin, large slack means the next fiber can require exponentially many trials cost of the rejection samplers
used in subsequent fiber operations.

At the root of the recursion, $X=U_A$ and $\Delta[U_A]=0$, so the
density-coin representation starts with zero slack. 
The main issue in the GGMT operations is fibering: sampling a fiber is implemented by rejection sampling, whose cost is governed by the acceptance probability of the current density coin. While this cost is controlled for a single typical fiber, successive fiber operations can cause the slack of our representation to accumulate, making later rejection samplers increasingly expensive.
Equivalently, the density coins of the
intermediate distributions may have smaller and smaller acceptance 
probabilities on typical samples, making subsequent rejection samplers
increasingly expensive.  Since a GGMT trajectory contains several such
operations, controlling the cost of each operation in isolation is not
enough.

To prevent this accumulation, after every GGMT transformation we apply a
\emph{bucketing} operation.  Informally, bucketing conditions the current
random variable on a coarse dyadic estimate of the acceptance probability
$r_X(x)$.  Within a selected bucket, the values of $r_X(x)$ are all on
approximately the same scale, and this regularizes the density coin and
restores constant slack on average.  The key point is that this computational regularization is cheap from the information-theoretic point 
of view: the bucket label reveals only $O(\log \Delta[X])$ bits of
information and furthermore the new slack is only a constant.  Thus bucketing can restore efficient algorithmic access keeping slack constant and 
without losing more than a logarithmic amount in the entropic Ruzsa-distance
potential (which we show in the next section).

\begin{lemma}
\label{lem:bucket}
Given algorithmic access to $X$, there is a random variable $J$ taking nonnegative integer values such that $\MI[X:J] = O(\log \Delta[X])$. Moreover, there is an algorithm that constructs algorithmic access to $X_j := (X \mid J = j)$ with envelope and expected slack satisfying
\begin{align*}
\E_{j \sim J} [\Delta[X_{j}] ] = O(1), \text{ and } M_{X_j} = M_X/(2^j \cdot \Pr[J = j]),
\end{align*}
respectively. Further, $\Sample_{X_j}$ uses an expected number of parent calls $O(2^{j+1}/\Pr[J = j])$,\footnote{Here, ``parent calls'' refers to calls to the algorithmic access procedures
for the original random variable $X$ before bucketing, namely
$\mathsf{Sample}_X$ and $\mathsf{Coin}_X$.} and one call to $\Coin_{X_j}$ uses at most $2^{j+1}$ parent coins.
\end{lemma}

\paragraph{Bucketing algorithm.} Before giving the proof of the above lemma, let us introduce some relevant notation and the underlying algorithm. We recall that $\Coin_X$ accepts $x$ with probability $r_X(x) = p_X(x)/M_X$ and the slack is $\Delta[X] = \E_{x \sim x} \log(1/r_X(x))$. We would like to regularize $X$ to have constant slack by conditioning $X$ on a coarse estimate of $\log(1/r_X(x))$. As $r_X(x)$ for a fixed $x \in \supp(X)$ is not known from the algorithmic access to $X$, we need to obtain an estimate of this quantity. This can be done by sampling $x \sim X$, then repeatedly running $\Coin_X(x)$ and letting $N_X$ be the first successful trial. Conditioned on $X=x$, we note that $N_x$ is geometric with mean $1/r_X(x)$ i.e., $N_X \sim \textsf{Geom}(r_X(x))$ and $\E[N_X | X=x] = 1/r_X(x)$. 

Let us define the random variable $J$ which takes value $J=j$ when $2^j \leq N_X < 2^{j+1}$ or in other words $J := \floor{\log N_X}$ is a dyadic estimate of $\log(1/r_X(x))$. We then define
\begin{equation}\label{def:bucket_quantities}
    \quad X_j := (X\mid J=j) = (X \mid 2^j \le N_X < 2^{j + 1}), \quad \text{and} \quad \pi_j := \Pr[J=j].
\end{equation}
The goal is to then implement algorithmic access to $X_j$, which is the desired regularized version of $X$. The above protocol for obtaining $N_X$ and implementing access to $X_j$ is described in Algorithm~\ref{algo:bucket}, which we will use for Lemma~\ref{lem:bucket}.
\begin{myalgorithm}
\begin{algorithm}[H]
    \caption{$\mathsf{Bucket}(X)$}
    \label{algo:bucket}
    \setlength{\baselineskip}{1.5em} %
    \DontPrintSemicolon %
    \KwInput{Algorithmic access to $X$ (Definition~\ref{def:algo_access})}
    \KwOutput{Label $j$, algorithmic access to $X_j$}
    \vspace{2mm}
    Sample $x \gets \Sample_X$. \\ 
    Run $\Coin_X(x)$ until the first success and denote $N_x$ as that time. \\ 
    Set $j:=\floor{\log N_x}$. \label{algo_line:sample_j} \\
    Define \Fn{$\Sample_{X_j}$ \label{algo_line:def_sample_Xj}}{
    Sample $x\gets\Sample_X$ and draw independent $\Coin_X(x)$ outcomes \label{algo_line:bucket_sample_x} \;
    If the first $2^j-1$ coins fail and at least one of the next $2^j$ coins succeeds, output $x$. \;
    Otherwise repeat from step \ref{algo_line:bucket_sample_x}. \;
    } \vspace{2mm}
    Define \Fn{$\Coin_{X_j}$ \label{algo_line:def_coin_Xj}}{
        Draw $2^{j+1}$ independent $\Coin_X(x)$ outcomes. \\
        Accept if the first $2^j$ trials contain exactly one success and the next $2^j$ trials contain at least one success. \label{algo_line:parent_calls_coin_Xj} \\
    } \vspace{2mm}
    \textbf{Return} ($j$, $\Sample_{X_j}$, $\Coin_{X_j}$).
\end{algorithm}
\end{myalgorithm}

\paragraph{Guarantees of bucketing.}
We are now ready to provide a proof of Lemma~\ref{lem:bucket} by first showing the correctness of the implementation of algorithmic access to $X_j$ in Algorithm~\ref{algo:bucket} and then proving two information-theoretic bounds.  
\begin{claim}\label{claim:bucket_algo_access}
Consider the context of Lemma~\ref{lem:bucket}. Algorithm~\ref{algo:bucket} implements algorithmic access to $X_j$, with envelope $M_{X_j} = M_X/(2^j \Pr[J = j])$, for $j$ sampled as in line~\ref{algo_line:sample_j}. $\Sample_{X_j}$ uses an expected number of $O(2^{j+1}/\Pr[J = j])$ calls to the algorithmic access of $X$ and $\Coin_{X_j}$ uses at most $2^{j+1}$ calls to $\Coin_X$.
\end{claim}
\begin{proof}
Let us denote $r := r_X(x)$ and $m=2^j$. Conditioned on $X=x$, the probability of the event $J=j$ (or $2^j\leq N_x<2^{j+1}$) is then
\begin{equation}\label{eq:conditional_prob_J}
\Pr[J=j | X=x] :=(1-r)^{m-1}\bigl(1-(1-r)^m\bigr).        
\end{equation}
Using the notation of $\pi_j:=\Pr[J=j]$,  Bayes' rule then gives us that
\begin{equation}\label{eq:prob_Xj}
p_{X_j}(x) =\frac{p_X(x) \Pr[J = j | X=x]}{\pi_j}.    
\end{equation}
We implement $\Sample_{X_j}$ as given in step~\ref{algo_line:def_sample_Xj} of Algorithm~\ref{algo:bucket}. We note that one trial of the rejection-sampling loop accepts with precisely the following probability as check if the first $m-1$ parent coins fail after sampling $x \sim X$ and require at least one success among the next $m$ parent coins
$$
\sum_x p_X(x) \Pr[J=j | X=x] = \pi_j
$$
Conditioned on acceptance, it then returns $x$ with probability $p_{X_j}$ as given in Eq.~\eqref{eq:prob_Xj}. The sampler is thus exact and the expected number of trials is $1/\pi_j$.  Each iteration uses fewer than $2^{j+1}$ parent coins.  This proves the claimed $O(2^{j+1}/\pi_j)$ expected parent-call bound.    

We implement $\Coin_{X_j}$ using step~\ref{algo_line:def_coin_Xj} of Algorithm~\ref{algo:bucket}. The goal is to ensure that $\Pr[\Coin_{X_j}(x)=1]$ is proportional to $p_{X_j}(x)$ in Eq.~\eqref{eq:prob_Xj} with an appropriately defined envelope. We note that the probability of acceptance associated with line~\ref{algo_line:parent_calls_coin_Xj} is
$$
\Pr[\Coin_{X_j}(x)=1] = mr(1-r)^{m-1}\bigl(1-(1-r)^m\bigr) = m r \Pr[J=j | X=x] =\frac{m\pi_j}{M_X}p_{X_j}(x),
$$
where we have used Eq.~\eqref{eq:conditional_prob_J} in the second equality, and used the expression of $p_{X_j}(x)$ (Eq.~\eqref{eq:prob_Xj}) along with the definition of $r = r_X(x) = p_X(x)/M_X$ in the final equality. Noting that the above is a valid probability law and is true $\forall x$, we then have that $X_j$ has an envelope $M_{X_j}={M_X}/({2^j\pi_j})$. This verifies the correctness of $\Coin_{X_j}$. The cost of one call to $\Coin_{X_j}$ is the $2^{j+1}$ calls to $\Coin_{X}$ to carry out step~\ref{algo_line:def_coin_Xj}. This completes the~proof.  
\end{proof}

We now show that conditioning $X$ on $J$ leads to suppression of the slack and regularizes the access.
\begin{claim}\label{claim:bucket_regularizing}
Consider the context of Lemma~\ref{lem:bucket}. Algorithm~\ref{algo:bucket} implements algorithmic access to $X_j$ with slack satisfying  $\E_{j \sim J}[\Delta[X_j]] < 8$
and $\MI[X:J] \leq \log( 3 \Delta[X] + 3)$.
\end{claim}
\begin{proof}
We note that the expected slack of $X_j$ is given by
\begin{align}
\E_{j \sim J}[\Delta[X_j]] 
&= \E_{j \sim J}\left[\entropy[X_j] + \log M_X - j - \log \Pr[J=j] \right] \nonumber \\
&= \entropy[X \mid J] + \log M_X - \E[J] + \entropy[J] \nonumber \\
&= \entropy[J \mid X] + \log M_X - \E[J] + \entropy[X] \nonumber \\
&= \entropy[J \mid X] + \log M_X - \E_X \Big[ \E[J \mid X] \Big] + \E_X[\log(1/p_X)] \nonumber\\
&= \entropy[J \mid X] - \E_X \Big[ \E[J \mid X] \Big] + \E_X[\log(1/r_X)], \label{eq:expression_slack_Xj}
\end{align}
where we have used the definition of slack (Eq.~\eqref{def:slack_coin}) and Claim~\ref{claim:bucket_algo_access} in the first line, definition of conditional entropy in the second line, the entropy chain rule of Eq.~\eqref{eq:entropy_chain_rules} to say $\entropy[X \mid J] + \entropy[J] = \entropy[X] + \entropy[J \mid X]$ in the third line, and definition of $r_X$ in the last line. We now bound the three terms in Eq.~\eqref{eq:expression_slack_Xj} separately.

Let us consider a fixed $x\in\supp(X)$. We define  $r:=r_X(x)$ and $k:=\lfloor \log(1/r) \rfloor$. Since $N_x$ is geometric with parameter $r$, we have that
\begin{equation}\label{eq:prob_Jgeqj}
\Pr[J\geq j\mid X=x] = \Pr[N_x \geq 2^{j}] = (1-r)^{2^j-1}.    
\end{equation}
For $j\leq k$, we also have
\begin{equation}\label{eq:prob_Jleqj}
\Pr[J<j\mid X=x] \leq \sum_{i=1}^{2^j - 1} \Pr[N_X = i \mid X = x] = \sum_{i=1}^{2^j - 1} (1-r)^{i-1} r \leq r 2^j,
\end{equation}
where we applied the union bound in the first inequality. We now bound $\E[J \mid X = x]$. Noting that for every non-negative integer-valued random variable $Z$, we have $\E Z = \sum_{j \geq 1} \Pr[Z \geq j]$, we have that
\begin{align}
\E[J \mid X = x] = \sum_{j=1}^\infty \Pr[J \geq j \mid X = x] 
& \geq \sum_{j=1}^k \Pr[J \geq j \mid X = x] \nonumber \\
& \geq \sum_{j=1}^k (1 - \Pr[J < j \mid X = x]) \nonumber \\
& \geq \sum_{j=1}^k (1 - r 2^j) \nonumber \\
& = k - r (2^{k+1} - 2) > k - 2 \geq \log(1/r) - 3, \label{eq:lb_E_J_X}
\end{align}
where we have used Eq.~\eqref{eq:prob_Jleqj} in the third line and definition of $k = \floor{\log(1/r)}$ in the last inequality. We can upper bound $\E[J \mid X = x]$ as
\begin{equation}\label{eq:ub_E_J_X}
    \E[J \mid X = x] = \E[\log N_X \mid X = x] \leq \log \E[N_X \mid X=x] = \log(1/r),
\end{equation}
where we have used Jensen's inequality in the second inequality and the fact that $N_x$ has mean $\log(1/r)$ in the last equality.

We now bound the conditional entropy $\entropy[J \mid X]$. Since $k = \floor{\log(1/r)}$, we equivalently have $2^{-(k+1)} < r \leq 2^{-k}$. For an integer $s \in [1,k]$, we have that
\begin{equation}\label{eq:prob_J_tail1}
\Pr[J \leq k - s \mid X = x] = \Pr[J < k - s + 1 \mid X = x] \leq r 2^{k-s+1} \leq 2^{1-s},
\end{equation}
where we applied Eq.~\eqref{eq:prob_Jleqj} in the second inequality and used $r \leq 2^{-k}$ in the last inequality. Note for $s > k$, we trivially have that the upper bound is $0$ since $J \geq 0$. To obtain an upper tail bound, consider any integer $s \geq 0$ and note that
\begin{equation}\label{eq:prob_J_tail2}
\Pr[J \geq k + s \mid X = x] = (1-r)^{2^{k+s} - 1} \leq \exp(-r(2^{k+s} - 1)) \leq \exp(1 - 2^{s-1}),
\end{equation}
where we used $1 - u \leq \exp(-u), \forall u \geq 0$. By a union bound and using Eqs.~\eqref{eq:prob_J_tail1}--\eqref{eq:prob_J_tail2}, we then have
\begin{equation}\label{eq:prob_J_union_tail_bound}
\Pr[|J-k| \geq s \mid X = x] \leq 2^{1-s} + \exp(1 - 2^{s-1}). 
\end{equation}
Using the fact that for every non-negative integer-valued random variable $Z$, we have $\E Z = \sum_{j \geq 1} \Pr[Z \geq j]$, we can then evaluate
\begin{align}\label{eq:entropy_Jminusk}
\E[ |J-k| \mid X=x] = \sum_{s=1}^\infty \Pr[|J-k| \geq s \mid X = x] \leq \sum_{s=1}^\infty 2^{1-s} + \sum_{s=1}^\infty \exp(1 - 2^{s-1}) < 2 + 3/2 < 4,
\end{align}
where we used Eq.~\eqref{eq:prob_J_union_tail_bound} in the second inequality and the fact that $2^{s-1} \geq 2(s-1), \forall s \geq 2$ in the penultimate inequality. We now obtain
\begin{align}
\entropy[J \mid X = x] = \entropy[ J - k \mid X = x] & \leq 1 + \entropy[|J-k| \mid X = x] \nonumber \\
& < 1 + (1 + \E[|J-k| \mid X = x) \log (1 + \E[|J-k| \mid X = x) < 5,    \label{eq:cond_entropy_ub}
\end{align}
where the entropy remains unchanged in the first equality since $k$ is deterministic, and the second inequality follows from noting that $J-k$ is determined by $|J-k|$ and one sign bit. The third equality follows from Fact~\ref{fact:ub_entropy_nonneg_int_RV} and the final inequality then follows from Eq.~\eqref{eq:entropy_Jminusk}. We can now bound the expected slack from Eq.~\eqref{eq:expression_slack_Xj}
\begin{equation}
\E_{j \sim J}[\Delta[X_j]] = \entropy[J \mid X] - \E_X \Big[ \E[J \mid X] \Big] + \E_X[\log(1/r_X)] \leq 5 + 3 - \E_X[\log(1/r_X)] + \E_X[\log(1/r_X)] = 8,
\end{equation}
where we used Eq.~\eqref{eq:cond_entropy_ub} to bound the first term and Eq.~\eqref{eq:lb_E_J_X} to bound the second term in the second inequality. This gives us the desired result regarding the slack.

To bound $\MI[X:J]$, we note that
\begin{equation}\label{eq:ub_EJ}
\E[J] = \E_X\Big[\E[J \mid X] \Big] \leq \E_{x\sim X}\log\frac1{r_X(x)} =\entropy[X]+\log M_X =\Delta[X],    
\end{equation}
where the second inequality follows from Eq.~\eqref{eq:ub_E_J_X} and then evaluate
\begin{equation}\label{eq:ub_MI_X_J}
\MI[X:J]
\leq \entropy[J] \leq (\E[J] + 1)\log(\E[J] + 1)- \E[J] \log \E[J] \leq \log(e \E[J] + e) = \log(3\Delta[X] + 3),    
\end{equation}
where we used Fact~\ref{fact:ub_entropy_nonneg_int_RV}. This completes the proof.
\end{proof}
Combining Claim~\ref{claim:bucket_algo_access} and Claim~\ref{claim:bucket_regularizing} gives us the proof of the desired Lemma~\ref{lem:bucket}.

\subsection{Sampling a trajectory from the GGMT tree}
\label{sec:samplingtrajectory}
We have so far shown how to implement algorithmic access for each of the possible candidate families (Section~\ref{sec:ggmt-five}) at every iteration of the GGMT tree exactly. Moreover, when there was an approximate post-selection involved (such as in the fibering steps), we showed that bucketing prevents the slack (defined in Eq.~\eqref{def:slack_coin}) of intermediate random variables from accumulating (Lemma~\ref{lem:bucket}) and hence can still be implemented efficiently. What remains is to navigate the GGMT tree itself. In the original combinatorial proof, one uses entropic inequalities to identify a favorable update and continues until the entropic Ruzsa distance becomes small. 

We now explain how to navigate the GGMT recursion (Theorem~\ref{thm:ggmt-decrement}) without ever computing an entropy or identifying which candidate gives the desired decrement (which is what the combinatorial argument of GGMT did). At this point the observation we have is simple, the depth of the GGMT tree is $O(\log \log K)$ and the tree at each layer has fanout size $5$~\footnote{This glosses over the fact that some of the candidate operations involve fiber sampling which have $2^n$ support. However, we will show a large fraction of these labels are favorable.}, so one can potentially just \emph{enumerate} all possible leaves of this tree with just a $\polylog(K)$ blow up. Our idea then is to simply choose among the five candidate families at random in each iteration. When the current entropic Ruzsa distance $d_t := \der{X_t}{Y_t}$ is above some specified terminal threshold $d_{\rm term}$, we will show that the GGMT decrement theorem (Theorem~\ref{thm:ggmt-decrement}) together with bucketing (Lemma~\ref{lem:bucket}) implies that an iteration has constant probability of simultaneously producing a constant-factor decrease in $d_t$, restoring constant slack, and incurring cost at most $2^{O(d)}=\polylog(K)$.

Concretely, this leads to the following algorithm (Algorithm~\ref{algo:ggmt-recursion}) that we will use to prove Theorem~\ref{thm:ggmt-recursion}.
\begin{myalgorithm}
\begin{algorithm}[H]
    \caption{Sampling trajectories from the GGMT tree}
    \label{algo:ggmt-recursion}
    \setlength{\baselineskip}{1.8em} %
    \DontPrintSemicolon %
    \KwInput{Sample access to $U_A$, query access to $1_{A}$, doubling constant $K$}
    \KwOutput{Algorithmic access to $X_T$ (Definition~\ref{def:algo_access}) satisfying Theorem~\ref{thm:ggmt-recursion} w.p. $\Omega(1/\polylog K)$.}
    \vspace{2mm}
    Initialize $X_0 = Y_0 = U_A$. \\
    Set maximum number of iterations $T = O(\log \log K)$ (Claim~\ref{claim:many-good-trajectories}). \\ 
    Set complexity budgets $B_S = \poly(K)$, $B_Q = \poly(K)$, and $B_{TC} = O(n \cdot \poly(K))$ (Claim~\ref{claim:ggmt_tree_with_limited_budget}). \\
    Initialize budget counters $N_S = N_Q = N_{TC} = 0$. \\
    \For{$t=1$ \KwTo $T$}
    {
        Obtain $X'_{t}, Y'_{t}$ from $X_{t-1},Y_{t-1}$ by choosing one of the five GGMT candidates (Section~\ref{sec:ggmt-five}) uniformly at random. \label{algo_line:choose_candidate} \\
        Construct algorithmic access to $X'_{t},Y'_{t}$ using the corresponding lemma (Lemmas~\ref{lem:closure-sum}--\ref{lem:access_endgame}) \label{algo_line:algo_access} \\
        Obtain algorithmic access to $X_{t},Y_{t}$ after bucketing on $X'_{t},Y'_{t}$ using Lemma~\ref{lem:bucket}. \label{algo_line:bucket} \\
        Update $N_S$, $N_Q$, $N_{TC}$ by adding the sample complexity to $U_A$, query complexity to $1_A$, and time complexity, respectively used as part of above three steps. \\
        \If{$N_S > B_S$ or $N_Q > B_Q$ or $N_{TC} > B_{TC}$}{Abort and return $\perp$}
    }
    \textbf{Return} Algorithmic access to $X_T$.
\end{algorithm}
\end{myalgorithm}

We prove Theorem~\ref{thm:ggmt-recursion} using the following three main arguments, whose intuition we summarize below. As part of these arguments and claims to come, we will require the following definition. For the fixed absolute constant $\eta>0$ from \Cref{thm:ggmt-decrement}, define the potential $\Phi$ as
$$
\Phi_{X,Y}(X',Y') :=\der{X'}{Y'} + \eta\bigl(\der{X}{X'}+\der{Y}{Y'}\bigr).
$$
\begin{enumerate}
\item We first comment on the effects of bucketing as part of each iteration in Algorithm~\ref{algo:ggmt-recursion}. Particularly, we show in Claim~\ref{claim:potential-after-bucketing} that bucketing in any iteration $t$ after randomly choosing the GGMT family does not increase the potential too much i.e., $\Phi_{X_{t-1},Y_{t-1}}(X_t,Y_t)$ is close to $\Phi_{X_{t-1},Y_{t-1}}(X_t',Y_t')$ which would have been obtained without bucketing. We also show in Claim~\ref{claim:cost-algo-access-after-bucketing} that the cost of algorithmic access to $X_t,Y_t$ remains bounded due to bucketing, in expectation.
\item We next show that in iteration $t$, whenever the current entropic Ruzsa distance $d_t := \der{X_{t-1},Y_{t-1}}$ is still above the terminal threshold $d_{\rm term}$, a random GGMT step followed by bucketing has constant probability of making a constant-factor decrement while restoring constant slack and keeping the local sampling costs under control. We call such iterations \emph{good}. Moreover, we show that Algorithm~\ref{algo:ggmt-recursion} is \emph{stable}. If the current entropic Ruzsa distance is less than $d_{\rm term}$ then with constant probability, the new entropic Ruzsa distance remains less than $d_{\rm term}$ and thus remain in the terminal regime. This is done in Claim~\ref{claim:good-iteration}.
\item We then finally iterate the one-step guarantee in the step above. Since every good iteration decreases the distance geometrically from its initial value at most $\log K$, only $O(\log\log K)$ good iterations are needed to reach constant entropic Ruzsa distance. Particularly, we show that an all-good trajectory reaches the terminal regime, remains there, and has sample/query complexity $\poly(K)$ and time complexity $O(n \cdot \poly(K))$. Such trajectories occur with probability at least $\Omega(1/\polylog K)$. This is shown in Claims~\ref{claim:good-iteration}--\ref{claim:ggmt_tree_with_limited_budget}.
\end{enumerate}
We prove these three claims in more detail~below.

\subsubsection{Effect of bucketing on the potential}
We first show that performing bucketing does not increase the potential for any of the five GGMT families, by too much.
\begin{claim}\label{claim:potential-after-bucketing}
Suppose we have algorithmic access to random variables $(X,Y)$. Define $d:=\der{X}{Y}$. Fix $f$ to be one of the five families of Theorem~\ref{thm:ggmt-decrement} and let $(X',Y')$ be obtained by applying $f$ to $(X,Y)$. Let $(\widehat{X},\widehat{Y})$ be the variables obtained after applying bucketing of Lemma~\ref{lem:bucket} to $(X',Y')$. If $\Delta[X],\Delta[Y] \leq L_0$ then 
$$
\E_{J_X,J_Y}[\Phi_{X,Y}(\widehat{X},\widehat{Y})] \leq\E_Z[\Phi_{X,Y}(X',Y')] + 2\log(6L_0 + 12d + 3),
$$
where $J_X,J_Y$ are the bucketing labels corresponding to $X',Y'$ respectively and $Z$ is the random variable(s) corresponding to any fiber label of $f$.
\end{claim}
\begin{proof}
Note that for $f$ being $\textsf{SelfFiber}$, $\textsf{CrossFiber}$, $\textsf{Endgame}$ involves sampling fiber labels as part of implementing $X',Y'$. We denote these natural labels as $z$ and the corresponding random variable(s) as $Z$. We now condition on $Z=z$ so that $X',Y'$ are fixed laws, and let $J_X,J_Y$ be their independent bucket labels.

From Ruzsa's triangle inequality~(Eq.~\eqref{eq:ruzsa_triangle}), we have that $\ruzsadist[X;X],\ruzsadist[Y;Y]\leq 2d$. We now observe that from Lemmas~\ref{lem:closure-sum}--\ref{lem:access_endgame} that $(X',Y')$ obtained after application of any family $f$ from Theorem~\ref{thm:ggmt-decrement} satisfies
\begin{equation}\label{eq:bound_delta_raw_outputs}
\E_{z \sim Z}[\Delta[X']] \leq 2L_0 + 4d, \quad \text{and} \quad \E_{z \sim Z}[\Delta[Y'] ]\leq 2L_0 + 4d.
\end{equation}
For \textsf{SelfSum} and \textsf{CrossSum}, the bounds are $\Delta[X'],\Delta[Y'] \leq L_0 + 4d$  and $\Delta[X'],\Delta[Y'] \leq L_0 + 2d$, respectively~(Lemma~\ref{lem:closure-sum}). For $\textsf{SelfFiber}$ and $\textsf{CrossFiber}$, the bound in both cases is
$\Delta[X'],\Delta[Y'] = 2L_0$ in expectation (Lemma~\ref{lem:closure-fiber}). For $\textsf{Endgame}$, the bounds are $\Delta[X'],\Delta[Y'] \leq 2L_0 + 4d$ in expectation (Lemma~\ref{lem:access_endgame}). This verifies Eq.~\eqref{eq:bound_delta_raw_outputs} for every case of $f$.

Let us define 
$$
I_X:=I[X':J_X \mid z],\qquad I_Y:=I[Y':J_Y \mid z].
$$
Applying Lemma~\ref{lem:ruzsa-conditioning} successively to the two bucket labels then gives
\begin{align}
\E_{J_X,J_Y}[\ruzsadist[\widehat{X};\widehat{Y}]\mid z] & \leq\ruzsadist[X';Y']+\frac12(I_X+I_Y), \nonumber \\
\E_{J_X}[\ruzsadist[X;\widehat{X}]\mid z] &\leq\ruzsadist[X;X']+\frac12 I_X, \nonumber \\
\E_{J_Y}[\ruzsadist[Y;\widehat{Y}]\mid z] &\leq\ruzsadist[Y;Y']+\frac12 I_Y.
\label{eq:int1_ub_dent}
\end{align}
Starting from the definition of $\Phi$, we then obtain
\begin{align}
\E_{J_X,J_Y}[\Phi_{X,Y}(\widehat{X},\widehat{Y})\mid z] 
&= \E_{J_X,J_Y}[\der{\widehat{X}}{\widehat{Y}} \mid z]  + \eta \Big( \E_{J_X}[\der{X}{\widehat{X}} \mid z]  + \E_{J_Y}[\der{Y}{\widehat{Y}} \mid z] \Big) \nonumber \\
&\leq \Phi_{X,Y}(X',Y') + \frac{1 + \eta}{2}(I_X + I_Y), \label{eq:int1_ub_potential}
\end{align}
where we used Eq.~\eqref{eq:int1_ub_dent} and the definition of $\Phi$ in the second line. We now bound the mutual information $I_X,I_Y$. Starting from Claim~\ref{claim:bucket_regularizing} which is true for any $z$, we obtain
\begin{align}
\E_{Z}[I_X + I_Y] &\leq \E_Z[\log(3 \Delta[X'] + 3)] + \E_Z[\log(3 \Delta[Y'] + 3)] \nonumber \\
&\leq \log(3 \E_Z[\Delta[X']] + 3) + \log(3 \E_Z[\Delta[Y']] + 3) \\
&\leq 2\log(6L_0 + 12d + 3), \label{eq:int1_ub_MI}
\end{align}
where we have used Jensen's inequality in the second line after noting that the concavity of $\log(\cdot)$ and used Eq.~\eqref{eq:bound_delta_raw_outputs} in the final line. Noting that $(1+\eta)/2 \leq 1$ and substituting Eq.~\eqref{eq:int1_ub_MI} into Eq.~\eqref{eq:int1_ub_potential} gives us
$$
\E_{J_X,J_Y}[\Phi_{X,Y}(\widehat{X},\widehat{Y})] \leq \E_Z[\Phi_{X,Y}(X',Y')] + 2\log(6L_0 + 12d + 3),
$$
which is the desired result. This completes the proof.
\end{proof}

\paragraph{Cost of algorithmic access after bucketing.}
We first make the notion of cost more precise. In each iteration $t$ of Algorithm~\ref{algo:ggmt-recursion}, we take a pair of independent random variables $X_{t-1},Y_{t-1}$, choose a family $F$ from Theorem~\ref{thm:ggmt-decrement}, set $X_t',Y_t'$ to be the pair obtained after applying $F$ and then obtain $X_t,Y_t$ by applying bucketing (Lemma~\ref{lem:bucket} and Algorithm~\ref{algo:bucket}). We do not store the laws of $X_t,Y_t$ explicitly but rather implement algorithmic access by specifying their corresponding samplers and density coins. These are implemented from the algorithmic access to $X_{t}',Y_{t}'$ as was shown in Lemma~\ref{lem:bucket} which in turn is implemented from access to $X_{t-1},Y_{t-1}$ as was shown in Lemmas~\ref{lem:closure-sum}--\ref{lem:access_endgame}. As access is implemented from prior level access, the cost of the procedures for $X_t,Y_t$ will then depend on the number of calls to the access procedures of $X_{t-1},Y_{t-1}$ which we need to control.
  
Going from $X_{t-1},Y_{t-1}$ to $X_t',Y_t'$, we mainly need for the rejection samplers used for sampling labels of fibers as part of applying $F$. If the fiber proposal acceptance probability is $a$, the sampler makes $a^{-1}$ proposals in expectation. Since these overheads multiply when access procedures are composed across iterations, it is convenient to record them logarithmically as $\log(a^{-1})$. We therefore define $\calC_{\rm fib}$ to be the sum of the logarithms of the reciprocal acceptance probabilities (see Lemmas~\ref{lem:closure-fiber}, and \ref{lem:access_endgame} for formal expressions) of the fiber rejection samplers used by the selected family $F$ as follows:
\begin{equation}\label{eq:cost_fiber}
\calC_{\rm fib}:=
\begin{cases}
0,
&\text{for \(\mathsf{SelfSum}\) and \(\mathsf{CrossSum}\),}\\[2mm]
\displaystyle
\log\frac{M_X}{p_{X+X}(z_X)}+
\log\frac{M_Y}{p_{Y+Y}(z_Y)},
&\text{for \(\mathsf{SelfFiber}\) with labels \((z_X,z_Y)\),}\\[3mm]
\displaystyle
\log\frac{M_Y}{p_{X+Y}(z_0)}+
\log\frac{M_X}{p_{X+Y}(z_1)},
&\text{for \(\mathsf{CrossFiber}\) with labels \((z_0,z_1)\),}\\[3mm]
\displaystyle
\log\frac{M_Y}{p_{X+Y}(z)}+
\log\frac{M_X}{p_{X+Y}(w+z)},
&\text{for \(\mathsf{Endgame}\) with label \((w,z)\).}
\end{cases}    
\end{equation}
For $\textsf{SelfSum}$ and $\textsf{CrossSum}$, no fiber sampling is needed, and hence $\calC_{\rm fib}=0$.

We now account for the overhead due to bucketing as we go from $X_t',Y_t'$ to $X_t,Y_t$. From Claim~\ref{claim:bucket_algo_access}, we have that the bucketed sampler uses $O\!\left(\frac{2^{j+1}}{\pi_{X,j}}\right)$ calls to its parent access procedures (i.e., $X_t',Y_t'$) conditional on the bucketing label $J_X=j$, while its density coin uses at most $2^{j+1}$ parent-coin calls. Since $\pi_{X,j} \leq 1$ and $\pi_{Y,j} \leq 1$, the former expression controls both access types. This motivates defining the overall logarithmic cost overhead of going from $X_{t-1},Y_{t-1}$ to $X_t,Y_t$ as
\begin{equation}\label{def:cost-access-after-bucketing}
\calC_{\rm access} : = \calC_{\rm fib} + \log \frac{2^{J_X + 1}}{\pi_{J_X}} + \log \frac{2^{J_Y + 1}}{\pi_{J_Y}}
\end{equation}
Up to some constant, $2^{\calC_{\rm access}}$ then bounds the calls made by access procedures of $X_t,Y_t$ to $X_{t-1},Y_{t-1}$ in iteration $t$. We now show that $\calC_{\rm access}$ remains bounded in expectation over the course of any iteration.
\begin{claim}\label{claim:cost-algo-access-after-bucketing}
Consider the context of Claim~\ref{claim:potential-after-bucketing}. We then have
$$
\E [\calC_{\rm access}] \leq 24(d + L_0 + 1).
$$
\end{claim}
\begin{proof}
Note that for $f$ being $\textsf{SelfFiber}$, $\textsf{CrossFiber}$, $\textsf{Endgame}$ involves sampling fiber labels as part of implementing $X',Y'$. We denote these natural labels as $z$ and the corresponding random variable(s) as $Z$. We now condition on $Z=z$ so that $X',Y'$ are fixed laws, and let $J_X,J_Y$ be their independent bucket labels. We denote $\pi_{J_X}(j) := \Pr[J_X = j]$ and $\pi_{J_Y}(j) := [J_Y = j]$.

To bound the expression in the claim statement, we will bound the terms separately. From Ruzsa's triangle inequality~(Eq.~\eqref{eq:ruzsa_triangle}), we have that $\ruzsadist[X;X],\ruzsadist[Y;Y]\leq 2d$. We now observe that from Lemmas~\ref{lem:closure-sum}--\ref{lem:access_endgame} that 
\begin{equation}\label{eq:bound_cost_fiber}
\E_{z \sim Z}[\calC_{\rm fib}] \leq 2L_0 + 4d, 
\end{equation}
For \textsf{SelfSum} and \textsf{CrossSum}, the expression is equal to $0$ since no fiber sampling is involved~(Lemma~\ref{lem:closure-sum}). For $\textsf{SelfFiber}$ and $\textsf{CrossFiber}$, the bounds are $L_0 + 4d$ and $L_0 + 2d$, respectively~(Lemma~\ref{lem:closure-fiber}). For $\textsf{Endgame}$, the bound is $2L_0 + 2d$~(Lemma~\ref{lem:access_endgame}). This verifies Eq.~\eqref{eq:bound_delta_raw_outputs} for every case of $f$.

We now note that
\begin{equation}\label{eq:int1_ub_second_term_cost}
\E_{J_X}[J_X + \log(1/\pi_{X,J_X})] \leq \Delta[X'] + \entropy[J_X] \leq \Delta[X'] + \log(3 \Delta[X'] + 3),
\end{equation}
where we have used $\E_{J_X}[J_X] \leq \Delta[X']$ (Eq.\eqref{eq:ub_EJ} from proof of Claim~\ref{claim:bucket_regularizing}) and $\E_{J_X}[\log(1/\pi_{X,J_X})] = \entropy[J_X]$ since $\pi_{X,j} = \Pr[J_X = j]$ in the second inequality followed by Eq.~\eqref{eq:ub_MI_X_J} (from proof of Claim~\ref{claim:bucket_regularizing}). Similarly, we can show that
$$
\E_{J_Y}[J_Y + \log(1/\pi_{Y,J_Y})] \leq \Delta[Y'] + \entropy[J_Y] \leq \Delta[Y'] + \log(3 \Delta[Y'] + 3).
$$
Combining Eq.~\eqref{eq:bound_cost_fiber} and Eq.~\eqref{eq:int1_ub_second_term_cost} gives us
\begin{align}
& \E_{Z,J_X,J_Y} \left[\calC_{\rm fib} + \log \frac{2^{J_X}}{\pi_{X,J_X}} + \log \frac{J_Y}{\pi_{Y,J_Y}}  \right]  \nonumber \\
\leq \, & (2L_0 + 4d) + \E_{z \sim Z}\bigl[\Delta[X'] + \Delta[Y']\bigr] + \E_{z \sim Z} \bigl[ \log(3\Delta[X'] + 3) + \log(3\Delta[Y'] + 3) \bigr] \nonumber \\
\leq \, & (2L_0 + 4d) +  (4L_0 + 8d) + \log \Big( \E_{z \sim Z}[3\Delta[X'] + 3]\Big) + \log \Big( \E_{z \sim Z}[3\Delta[Y'] + 3]\Big) \nonumber \\
\leq \,& 6L_0 + 12 d+ 2 \log(6L_0 + 12d + 3) \\
\leq \,& 12L_0 + 24d + 3,
\end{align}
where we have used $\E_{z \sim Z}\bigl[\Delta[X']\bigr], \E_{z \sim Z}\bigl[\Delta[Y']\bigr] \leq 2L_0 + 4d$ from Eq.~\eqref{eq:bound_delta_raw_outputs} (proof of Claim~\ref{claim:potential-after-bucketing}) in the third line along with application of Jensen's inequality after noting the concavity of $\log$. We used Eq.~\eqref{eq:bound_delta_raw_outputs} again in the fourth line.
\end{proof}

\subsubsection{Guarantees of a good iteration}
In Algorithm~\ref{algo:ggmt-recursion}, we choose the GGMT family to be applied uniformly at random. We thus do not have the promise that the energy $\Phi$ will reduce for all choices made. Moreover, even if $\Phi$ reduced, it is possible that the algorithmic access procedures of the resulting variables $X_t,Y_t$ in iteration $t$ require prohibitively many calls to previous access procedures making them expensive to implement. We thus define a \emph{good} iteration as one where $\Phi$ reduces, the slack of resulting $X_t,Y_t$ remain bounded and thereby the cost of their algorithmic access remains bounded. We will make this more formal shortly. Let us quickly comment on the promise of any iteration in Algorithm~\ref{algo:ggmt-recursion}.

We show that an iteration in Algorithm~\ref{algo:ggmt-recursion} is good with high probability. In particular, we consider the two cases of $(i)$ when the entropic Ruzsa distance $d_{t-1} := \der{X_{t-1}}{Y_{t-1}}$ at the beginning of an iteration $t$ is higher than the specified terminal Ruzsa distance $d_{\rm term}$, and $(ii)$ when $d_{t-1} \leq d_{\rm term}$. 

In case $(i)$, we show that Algorithm~\ref{algo:ggmt-recursion} will choose a pair $X_t,Y_t$ which will reduce the potential $\Phi$ and keep the cost of algorithmic access bounded, with high probability. In case $(ii)$, since $d_{t-1} \leq d_{\rm term}$ for $t < T_{\max}$, it is possible that $X_t,Y_t$ are chosen such that the potential increases and we escape the terminal regime. We show that Algorithm~\ref{algo:ggmt-recursion} is \emph{stable} and that with high probability, $X_t, Y_t$ are chosen such that we remain in the terminal regime. The above two stated results are shown formally in the claim below.
\begin{claim}\label{claim:good-iteration}
Let $\eta \in (0,1)$ be the promise of Theorem~\ref{thm:ggmt-decrement}. Suppose we are in iteration $t$. Define absolute constants $L_0 = 256/\eta$, $\xi = 768/\eta$ and
$d_{\rm term}=\max\{L_0 + 1, 16/\eta \log(64/\eta)\}$. Suppose that we have algorithmic access to $X_{t-1},Y_{t-1}$ such that
$$
\Delta[X_{t-1}],\Delta[Y_{t-1}]\leq L_0,\qquad d_{t-1} :=\der{X_{t-1}}{Y_{t-1}}.
$$
Define $d := \max\{d_{t-1}, d_{\rm term}\}$. Then, running one iteration of steps~\ref{algo_line:choose_candidate}--\ref{algo_line:bucket} of
Algorithm~\ref{algo:ggmt-recursion}, starting from the current pair
$X_{t-1},Y_{t-1}$ outputs $X_t,Y_t$ with probability $\geq \eta/40$ such that
$$
\Phi_{X_{t-1},Y_{t-1}}(X_t, Y_t) \leq (1 - \eta/2)d, \quad \Delta[X_t], \Delta[Y_t] \leq L_0, \quad \calC_{\rm access} \leq \xi(d+1).
$$
\end{claim}
\begin{proof}
To simplify the notation, let us denote $X:= X_{t-1}$ and $Y:= Y_{t-1}$ as the variables at the start of iteration $t$. Let the outputs from step~\ref{algo_line:choose_candidate} of Algorithm~\ref{algo:ggmt-recursion} be $X':= X'_{t}$ and $Y':= Y_{t}$, which are obtained before bucketing. Let $Z$ be the random variable(s) denoting any of the natural fiber labels chosen as part of obtaining $X',Y'$. Let $\widehat{X} := X_t$ and $\widehat{Y} := Y_t$ be the variables obtained after applying bucketing. We will still denote $d_{t-1} := \der{X_{t-1}}{Y_{t-1}}$.

By \Cref{thm:ggmt-decrement}, there is a family $F^{\star}$ such that, 
$$
\E_{z \sim Z} \bigl[\Phi_{X,Y}(X',Y') \mid F^{\star}\bigr] \leq (1-\eta)d_{t-1}.
$$
Step~\ref{algo_line:choose_candidate} of Algorithm~\ref{algo:ggmt-recursion} chooses $F^{\star}$ with probability at least $1/5$. From Claim~\ref{claim:potential-after-bucketing}, we then have that
\begin{equation}\label{eq:ub_exp_Phi}
\E_{Z,J_X,J_Y}[\Phi_{X,Y}(\widehat{X},\widehat{Y}) \mid F^{\star}] \leq (1-\eta)d_{t-1} + 2\log(6L_0 + 12d_{t-1} + 3).    
\end{equation}
Let us analyze the upper bound of the above expression by considering two cases: $(i)$ $d_{t-1} > d_{\rm term}$, and $(ii)$ $d_{t-1} \leq d_{\rm term}$.

For case $(i)$, we now upper bound the right hand side of Eq.~\eqref{eq:ub_exp_Phi}. Choosing $d_{\rm term} = \max\{L_0 + 1, 16/\eta \log(64/\eta)\}$ ensures that 
$$
2\log(6 L_0 + 12 d_{t-1} + 3) \leq 2\log(18 d_{t-1})
$$ 
since $d_{\rm term} \leq d_{t-1}$. To bound $2 \log (18 d_{t-1})$, we use the following observation. Consider the function $g(x) := 2 \log(18x)/x, \forall x > 0$. We note that its derivative is $g'(x) = 2(1 - \ln(18x))/(x^2 \ln 2) < 0$ for $x \geq \exp(1)/18$ and hence $g(x)$ is a decreasing function for $x \geq \exp(1)/18$. Let us define $x_0 := 16/\eta \log(64/\eta)$ which we note is greater than $\exp(1)/18$ for $\eta \leq 1$. We then have that
$$
\frac{2 \log(18 d_{t-1})}{d_{t-1}} \leq \frac{2 \log(18x_0)}{x_0} \leq \frac{4 \log (64/\eta)}{x_0} = \frac{\eta}{4} \implies 2 \log(18 d_{t-1}) \leq \frac{\eta}{4} d_{t-1},
$$
where we used $18 x_0 \leq (64/\eta)^2$ in the second inequality, and definition of $x_0$ in the final inequality before the implication. Combining the above bound with Eq.~\eqref{eq:ub_exp_Phi} then gives us that
\begin{equation}\label{eq:ub_exp_Phi_high}
\E\bigl[\Phi_{X,Y}(\widehat{X},\widehat{Y})\mid F^{\star}\bigr] \leq\left(1-\frac{3\eta}{4}\right)d_{t-1}.    
\end{equation}
In case $(ii)$, we can then bound the expression from Eq.~\eqref{eq:ub_exp_Phi} as
$$
\E_{J_X,J_Y}[\Phi_{X,Y}(\widehat{X},\widehat{Y}) \mid F^{\star}] \leq (1-\eta)d_{\rm term} + 2\log(6L_0 + 12d_{\rm term} + 3).
$$
Let us choose $d_{\rm term} = \max\{L_0 + 1, 16/\eta \log(64/\eta)\}$ as before. This ensures that $2\log(6 L_0 + 12 d + 3) \leq 2 \log(18d_{\rm term})$. Given our choice, we can write $d_{\rm term} = 16x/\eta$ for some $x \geq \log(64/\eta) \geq 6$ since $\eta < 1$. We note that $6x \leq 2^{x}$ since $x \geq 6$ and hence $1/\eta \leq 2^x/64$. This gives us $2 \log(18 d_{\rm term}) = 2 \log(288x/\eta)$ and $288x/\eta \leq 2^{2x}$ from our earlier observations. This overall gives us that $2 \log(18 d_{\rm term}) \leq 4x = \eta d/4$. We then have
\begin{equation}\label{eq:ub_exp_Phi_stab}
\E\bigl[\Phi_{X,Y}(\widehat{X},\widehat{Y})\mid F^{\star}\bigr] \leq\left(1-\frac{3\eta}{4}\right)d_{\rm term}.    
\end{equation}

Let us define $d := \max\{d_{t-1},d_{\rm term}\}$. We then have for both $(i)$ $d_{t-1} > d_{\rm term}$ (Eq.~\eqref{eq:ub_exp_Phi_high}) and $(ii)$ $d_{t-1} \leq d_{\rm term}$ (Eq.~\eqref{eq:ub_exp_Phi_stab}) that
\begin{equation}\label{eq:ub_exp_Phi_all}
\E\bigl[\Phi_{X,Y}(\widehat{X},\widehat{Y})\mid F^{\star}\bigr] \leq\left(1-\frac{3\eta}{4}\right)d.    
\end{equation}
Since the potential is nonnegative, Markov's inequality now gives
\begin{equation}\label{eq:markov_potential}
\Pr[\Phi_{X,Y}(\widehat{X},\widehat{Y}) \leq (1-\eta/2) d \mid F^{\star}] \geq 1 - \frac{\E[\Phi_{X,Y}(\widehat{X},\widehat{Y})\mid F^{\star}\bigr]}{1- \eta/2} \geq 1 - \frac{1-3\eta/4}{1-\eta/2} = \frac{\eta/4}{1-\eta/2} \geq \eta/4,    
\end{equation}
where we have used Eq.~\eqref{eq:ub_exp_Phi_all} in the second inequality and the fact that $\eta \in (0,1)$ in the final inequality. 

Noting that after the natural label $Z=z$ has been fixed, the laws of $X',Y'$ are fixed and we can then apply Claim~\ref{claim:bucket_regularizing} to obtain (which is true regardless of the value of $d_{t-1}$ in comparison with $d_{\rm term}$)
$$
\E_{J_X} \left[ \Delta[\widehat{X}] \mid Z = z, F^\star \right] < 8, \quad \E_{J_Y} \left[ \Delta[\widehat{Y}] \mid Z = z, F^{\star} \right] < 8.
$$
Noting that the above expression is true for all $Z=z$, we then have after adding both expressions that
$$
\E_{Z,J_X,J_Y} \left[\Delta[\widehat{X}] + \Delta[\widehat{Y}] \mid F^{\star} \right] < 16.
$$
Applying Markov's inequality again, we obtain
\begin{equation}\label{eq:ub_prob_Delta_hatX_hatY}
\Pr[(\Delta[\widehat{X}] + \Delta[\widehat{Y}]) \geq L_0 \mid F^{\star} ] \leq 16/L_0.    
\end{equation}
Finally, from Claim~\ref{claim:cost-algo-access-after-bucketing} which holds true for any family $f$, we have that
$$
\E[\calC_{\rm access} \mid F^{\star}] \leq 24(L_0 + d_{t-1} + 1) \leq 48d,
$$
where we have used the fact that $d_{\rm term} \geq L_0 + 1$ and used the definition $d = \max\{d_{t-1},d_{\rm term}\}$. Let $\xi > 0$ be a parameter to be determined later. Applying Markov's inequality considering the above expression for $\calC_{\rm access}$, we then have that
\begin{equation}\label{eq:ub_prob_cost_access}
    \Pr[\calC_{\rm access} \geq \xi(d+1) \mid F^{\star}] \leq \frac{48 d}{\xi (d+1)}.
\end{equation}
Let us define $E_1$ as the event that $\Phi_{X,Y}(\widehat{X},\widehat{Y}) \leq (1-\eta/2)d$, $E_2$ as the event that $\Delta[\widehat{X}] + \Delta[\widehat{Y}] \leq L_0$, and $E_3$ as the event that $\calC_{\rm access} \leq \xi(d+1)$. Let us denote $\overline{E}$ as the complement of the event $E$. Applying an union bound, we have that
\begin{align}
& \Pr[ \cap_{i \in [3]} E_i \mid F^{\star}] \nonumber \\
\geq \, &  1 - \sum_{i \in [3]} \Pr[\overline{E}_i \mid F^{\star}] \nonumber \\
\geq \, & \Pr[\Phi_{X,Y}(\widehat{X},\widehat{Y}) \leq (1-\eta/2) d \mid F^{\star}] - \Pr[(\Delta[\widehat{X}] + \Delta[\widehat{Y}]) \geq L_0 \mid F^{\star}] - \Pr[\calC_{\rm access} \geq \xi(d+1) \mid F^{\star}] \nonumber \\
\geq \, & \frac{\eta}{4} - \frac{16}{L_0} - \frac{48d}{\xi(d+1)},
\end{align}
where we have used Eqs~\eqref{eq:markov_potential},\eqref{eq:ub_prob_Delta_hatX_hatY},and \eqref{eq:ub_prob_cost_access} in the last inequality. Choosing $L_0 = 256/\eta$ and $\xi = 768/\eta$ ensures that
$$
\Pr[\cap_{i \in [3]} E_i \mid F^{\star}] \geq \eta/8.
$$
Noting that $\cap_{i \in [3]} E_i$ is the desired good iteration and $F^{\star}$ is chosen uniformly at random (independent of other choices), we then have that
$$
\Pr[\cap_{i \in [3]} E_i] \geq \eta/40.
$$
This completes the proof.
\end{proof}

\subsubsection{Large fraction of trajectories are good and hit terminal regime}
We have so far shown that any iteration in Algorithm~\ref{algo:ggmt-recursion} is good (Claim~\ref{claim:good-iteration}). We now show that there is a large fraction of trajectories (or sequence of iterations) that terminate quickly and where each iteration in the trajectory is promised to be good. Let us define the maximum number of iterations that Algorithm~\ref{algo:ggmt-recursion} is allowed to run for as
\begin{equation}\label{eq:Tmax}
T :=\min\{t\geq0:(1- \eta/2)^t\log K \leq d_{\rm term}\}.    
\end{equation}
\begin{claim}
\label{claim:many-good-trajectories}
Consider the context of Theorem~\ref{thm:ggmt-recursion}. Running Algorithm~\ref{algo:ggmt-recursion} for $T = O(\log \log K)$ iterations ensures that with probability $\geq \Omega(1/\polylog K)$, each iteration is good in the sense of Claim~\ref{claim:good-iteration} and
$$
d_T \leq d_{\rm term}, \quad \Delta[X_T],\Delta[Y_T] \leq L_0
$$
and
$$
\sum_{t < T} (\widetilde{d}_t + 1) = O(\log K), \quad \der{U_A}{X_T} = O(\log K),
$$
where $\widetilde{d}_t := \max\{\der{X_t}{Y_t}, d_{\rm term}\}$.
\end{claim}
\begin{proof}
At $t=0$, we set the initial random variables $X_0$ and $Y_0$ to have the distribution of $U_A$. The initial entropic Ruzsa distance is then
\begin{equation}\label{eq:d0}
d_0 := \der{X_0}{Y_0} = \entropy[U_A+U_A'] - \entropy[U_A] \leq\log|A+A| - \log|A| \leq\log K,
\end{equation}
where we have used that $(U_A + U_A')$ is supported on $(A+A)$ in the second inequality and that $A$ has doubling constant $K$ in the final inequality. The initial slack is zero from a quick calculation:
$$
\Delta[X_0] = \Delta[Y_0] = \entropy[U_A] + \log M_{U_A} = \log |A| - \log |A| = 0,
$$
where we noted that $M_{U_A} = 1/|A|$ by definition. Let $G_s$ be the event that the iteration $s \in [T]_0$ is good in the sense of Claim~\ref{claim:good-iteration} (i.e., decreases the potential, has constant slack, and has bounded cost of algorithmic access). The event $\cap_{s < t} G_s$ then corresponds to all the iterations from $0$ up to $t$ being good. Note that at the end of iteration $s$, the resulting current pair $X_s,Y_s$ satisfies the input requirements to Claim~\ref{claim:good-iteration}. We then have from Claim~\ref{claim:good-iteration} that
\begin{equation}\label{eq:cond_prob_good_t}
\Pr[G_t \mid \cap_{s < t} G_s] \geq \eta/40
\end{equation}
Let us define $p:= \eta/40$. Applying the chain rule for conditional probabilities over an entire $T$-length trajectory then gives us
\begin{equation}\label{eq:prob_good_trajectory}
\Pr[\cap_{s \in [T]} G_s ] = \prod_{t=1}^T \Pr[G_t \mid \cap_{s < t} G_s] \geq p^T,
\end{equation}
where we applied Eq.~\eqref{eq:cond_prob_good_t} in the second inequality. Let us define $\rho := (1-\eta/2)$ and recall the definition of $T = \min\{t \geq 0 : \rho^t \log K \leq d_{\rm term}\}$ from Eq.~\eqref{eq:Tmax}. Suppose that $\log K > d_{\rm term}$ or else $X_0$ already satisfies the required guarantees of Theorem~\ref{thm:ggmt-recursion} and we can set $T=0$. We then have that
\begin{equation}\label{eq:ub_Tmax}
T \leq \frac{\log\Big(\log K/d_{\rm term}\Big)}{\log(1/\rho)} + 1 \leq 1 + \frac{2}{\eta} \log \log K = O\left(\eta^{-1} \log \log K \right),
\end{equation}
where we used the fact that $d_{\rm term} \geq 1$ and that $\log(1/\rho) = \log(1/(1-\eta/2)) \geq \eta/2$ since $-\ln(1-x) \geq x, \forall x < 1$.
Substituting the above bound on $T$ back into Eq.~\eqref{eq:prob_good_trajectory}, we then have that the probability of a good trajectory is 
\begin{equation}\label{eq:lb_prob_good_trajectory}
\Pr[\cap_{s \in [T]} G_s ] \geq p^T \geq p (p)^{2/\eta \log \log K}
 \geq p(\log K)^{2/\eta \log p} \geq \frac{\eta/40}{(\log K)^{2/\eta \log (40/\eta)}} = \Omega(1/\polylog K),
\end{equation}
where we used $0 < p \leq 1$ and Eq.~\eqref{eq:ub_Tmax} in the second inequality, the fact that $p \geq \eta/40$ in the fourth inequality and $\eta=O(1)$ in the final inequality.

We now derive the promises on $X_T,Y_T$ obtained at the end of a good trajectory in $\mathcal{G} := \cap_{s \in [T]} G_s$. Claim~\ref{claim:good-iteration} promises us that on every good iteration, we have
\begin{equation}\label{eq:implication_on_der}
d_{t+1} \leq\Phi_{X_t,Y_t}(X_{t+1},Y_{t+1}) \leq\rho\widetilde d_t,
\end{equation}
where $\widetilde{d}_t = \max\{\der{X_t}{Y_t},d_{\rm term}\}$ and noted the definition of $\Phi$ for the lower bound. Thus in every iteration $t$, while $d_t>d_{\rm term}$, we have $d_{t+1}\leq\rho d_t$. If $\tau < T$ is the first iteration that $d_\tau \leq d_{\rm term}$, then for a good trajectory in $\mathcal{G}$, $d_T \leq d_{\rm term}$ due to Eq.~\eqref{eq:implication_on_der}. If there is no such iteration $\tau < T$, then we have that
$$
d_T \leq \rho^T d_0 \leq \rho^T \log K \leq d_{\rm term},
$$
by the definition of $T$ (see Eq.~\eqref{eq:Tmax}). Thus, we have that for every trajectory in $\mathcal{G}$, $d_T \leq d_{\rm term}$.

To evaluate $\sum_{t<T}(\widetilde{d}_t+1)$ for a trajectory in $\mathcal{G}$, we will again define $\tau$ to be the first iteration such that $d_{\tau} \leq d_{\rm term}$. We then have that
\begin{equation}\label{eq:ub_sum_dt}
\sum_{t<T}(\widetilde{d}_t+1) \leq \sum_{t<\tau}\rho^t \log K + \sum_{\tau \leq t < T} d_{\rm term} + T \leq  \frac{\log K}{1 - \rho} +(d_{\rm term}+1)T = O(\log K)
\end{equation}
where we used Eq.~\eqref{eq:implication_on_der} and the fact that $d_0 \leq \log K$ followed by the upper bound on $T$ from Eq.~\eqref{eq:ub_Tmax} along with the fact that $\eta=O(1)$.

To bound the entropic Ruzsa distance $\der{U_A}{X_T}$, we note that $X_0 = U_A$ and use the entropic Ruzsa triangle inequality to obtain
$$
\der{U_A}{X_T} \leq \sum_{t=1}^T d[X_{t-1};X_t] \leq\frac{\rho}{\eta} \sum_{t=1}^T \widetilde d_t = O(\log K),
$$
where in the second inequality used that the definition of the potential $\Phi$ and Claim~\ref{claim:good-iteration} implies $\eta \der{X_{t-1}}{X_{t}} \leq\Phi_{X_{t-1},Y_{t-1}}(X_{t},Y_{t}) \leq\rho\widetilde d_{t-1}$ and then used Eq.~\eqref{eq:ub_sum_dt} in the final inequality. The bound on $\Delta[X_T],\Delta[Y_T]$ follows from the promise of the last good iteration (Claim~\ref{claim:good-iteration}) if $T>0$ or from the fact that $\Delta[X_0]=\Delta[Y_0] = 0$ if $T=0$. This completes the proof.
\end{proof}

\subsubsection{Cost of good trajectories}
\begin{claim}\label{claim:cost_good_trajectories}
Consider the context of Theorem~\ref{thm:ggmt-recursion}. Let $\mathcal{G}$ be the set of trajectories satisfying the promise of Claim~\ref{claim:many-good-trajectories}. Then, Algorithm~\ref{algo:ggmt-recursion} implements algorithmic access to $X_T,Y_T$ as the end of a length-$T$ trajectory $\Gamma$ satisfying the promise of Claim~\ref{claim:many-good-trajectories} with cumulative cost of $\calC_S$ sample complexity to $U_A$, $\calC_Q$ query complexity to $1_A$, and time complexity $\calC_{TC}$ such that
\begin{align*}
\E[\calC_S \mid \Gamma \in \mathcal{G}] = P_1(K), \quad 
\E[\calC_Q \mid \Gamma \in \mathcal{G}] = P_2(K), \text{ and }
\E[\calC_{TC} \mid \Gamma \in \mathcal{G}] = O(n \cdot P_3(K)),
\end{align*}
where $P_1,P_2,P_3$ are effective polynomials independent of $\Gamma$.
\end{claim}
\begin{proof}
Let us consider a trajectory $\Gamma$ and denote $F_t$ to be the selected GGMT family from Theorem~\ref{thm:ggmt-decrement}, $Z_t$ to be the tuple of natural fiber/endgame labels required by $F_t$ (or $\emptyset$ tuple for a sum) and let $J_{X,t},J_{Y,t}$ be the two bucket labels. We will collectively denote the tuple $\gamma_t := (F_t, Z_t, J_{X,t},J_{Y,t})$. We can then describe the trajectory $\Gamma$ as
$$
\Gamma := \{\gamma_t\}_{t \in [T]} = \{(F_t,Z_t,J_{X,t},J_{Y,t}\}_{t \in [T]}.
$$
Note that with the description of $\Gamma$ up to iteration $t$, we can recursively algorithmically access $X_t,Y_t$ even though this is not held explicitly. We also do not record rejected proposals as part of the construction of $\Gamma$. We will use $\Gamma_{\leq t} := \{\gamma_s\}_{s \in [t]}$ to denote the trajectory up to iteration $t$. 

We now determine the cost of a good trajectory $\Gamma \in \mathcal{G}$ satisfying Claim~\ref{claim:many-good-trajectories}. For resource $R \in \{S,Q,TC\}$, let $R(\calA)$ be the number of samples from
$U_A$, queries to $1_A$, or time complexity, respectively,
used by one call to an access procedure $\calA$ (which is either $\Sample$ or $\Coin$). Let $\calC_R^{(t)}$ be the $R$-cost of implementing algorithmic access to $X_{t},Y_{t}$ in iteration $t$ (which involves choosing $F_{t}$, generating the natural labels $Z_{t}$, generating bucket labels $J_{X,t},J_{Y,t}$, and defining the resulting composed sampler and coin descriptions). Let the per-call cost of accessing the resulting samplers be denoted as $\textsf{SampleCost}_{R,t}$ and density coins be denoted as $\textsf{CoinCost}_{R,t}$, which we define for a fixed supported trajectory $\Gamma_{\leq t} = G_{\leq t}$ as
\begin{align}
\textsf{SampleCost}_{R,t}(G_{\leq t}) &:=\max_{V\in\{X_t,Y_t\}} \max_{v\in\supp(V)} \E\!\left[ R(\Sample_V) \,\middle|\, \Gamma_{\leq t}= G_{\leq t}, \Sample_V=v \right], \label{eq:cost_sample_t} \\
\textsf{CoinCost}_{R,t}(G_{\leq t}) &:=\max_{V\in\{X_t,Y_t\}} \sup_{x\in\mathbb F_2^n} \max_{b\in \supp(\Coin_V(x))} \E\!\left[ R(\Coin_V(x)) \,\middle|\, \Gamma_{\leq t}= G_{\leq t}, \Coin_V(x)=b \right] \label{eq:cost_coin_t}.
\end{align}
We will also define 
\begin{equation}\label{eq:cost_R_t}
\textsf{Cost}_{R,t}(G_{\leq t}) :=\max\bigl\{ \textsf{SampleCost}_{R,t}(G_{\leq t}), \textsf{CoinCost}_{R,t}(G_{\leq t}) \bigr\}
\end{equation}
as the maximum cost over one call to any sampler or coin available at iteration $t$. Note that this maximizes over the two laws $X_t,Y_t$, the type of access, every coin input, and every supported returned sample or bit. The expectation still averages over the internal randomness of that call. For a depth-$t$ access call $\mathcal A$, let $O$ be its outputs. Then we have that
\begin{align}
\E[R(\mathcal A)\mid \Gamma_{\leq t}=G_{\leq t}] 
&= \sum_o \Prob(O=o \mid\Gamma_{\leq t} = G_{\leq t}) \cdot \E[R(\mathcal A)\mid\Gamma_{\leq t}=G_{\leq t},O=o] \nonumber \\
& \leq \max_o \E[R(\mathcal A) \mid\Gamma_{\leq t}= G_{\leq t},O=o] \nonumber \\
& \leq \textsf{Cost}_{R,t}(G_{\leq t}). \label{eq:cost_R_traj}
\end{align}

Fix a supported good trajectory $G\in\mathcal G$, and let
$C_t(G)$ be the value of $C_{{\rm access},t}$ from
Claim~\ref{claim:cost-algo-access-after-bucketing} in iteration $t+1$. We now describe the effect of the involved rejection samplers. Suppose one component is a rejection sampler with one-trial acceptance probability $a$, accepting trial $M$, and returned sample $\mathsf{Out}$. If
$r_v:=\Prob(\text{one trial accepts and returns }v)$, then
$$
\Prob(M=m\mid\mathsf{Out}=v) = \frac{(1-a)^{m-1}r_v}{r_v/a} =a(1-a)^{m-1}, \quad  \E[M\mid\mathsf{Out}=v] =a^{-1}.
$$
Thus, if one trial has conditional expected \(R\)-cost at most \(c\)
for every outcome of its lower-level calls, then the tower property
gives
\begin{equation*}
 \E\!\left[
   \text{\(R\)-cost of the rejection sampler}
   \,\middle|\,\mathsf{Out}=v
 \right]
 \leq \frac{c}{a}.
\end{equation*}
This is the only reason that output conditioning is mentioned: it
shows that the usual reciprocal-acceptance multiplier \(a^{-1}\)
remains valid for $\textsf{Cost}_{R,t+1}$.

We now bound $\textsf{Cost}_{R,t}$. Recall that algorithmic access procedures of $\Sample$ or $\Coin$ for $X_t,Y_t$ at iteration $t$ are built by calling procedures from iteration $t-1$. At the root, $X_0=Y_0=U_A$. A sampler call makes one call to $U_A$, a density-coin call makes one membership query to $1_A$, and either call requires $O(n)$ time complexity. Hence
\begin{equation}\label{eq:cost_at_t0}
\textsf{Cost}_{S,0}\leq1,\qquad
\textsf{Cost}_{Q,0}\leq1,\qquad
\textsf{Cost}_{T,0}=O(n).
\end{equation}
Let us now fix a supported good trajectory $\Gamma = G \in \mathcal{G}$ and let $C_t(G)$ be the value of $C_{{\rm access},t}$ from Claim~\ref{claim:cost-algo-access-after-bucketing} for iteration $t$. Sum operations use only constantly many parent calls. A fiber sampler with one-trial acceptance probability $a$ uses a geometric number of trials of mean $a^{-1}$. Noting that $\log a^{-1} \leq C_{\rm fib}$ (defined in Eq.~\eqref{def:cost-access-after-bucketing}, and hence $a^{-1} \leq 2^{C_{\rm fib}} \leq 2^{C_t(G)}$.

For a bucket label $j$, with probability $\pi_j$, Claim~\ref{claim:bucket_algo_access} gives a sample overhead of $O(2^{j+1}/\pi_j)$ and overall uses at most $2^{j+1}$ parent-coin calls. Note that the logarithms of all these multipliers for the two output laws $X_t,Y_t$ are included in $C_t(G)$. Therefore, the number of calls to algorithmic access procedures at iteration-$t$ by access procedures at iteration-$(t+1)$ is $D \cdot 2^{C_t(G)}$ for an absolute constant $D \geq 1$. We then have
\begin{align}
\textsf{Cost}_{S,t+1}(G_{\leq t+1}) &\leq D \cdot 2^{C_{t}(G)} \bigl(\textsf{Cost}_{S,t}(G_{\leq t})+1\bigr),\nonumber\\
\textsf{Cost}_{Q,t+1}(G_{\leq t+1}) &\leq D \cdot 2^{C_{t}(G)} \bigl(\textsf{Cost}_{Q,t}(G_{\leq t})+1\bigr), \nonumber \\
\textsf{Cost}_{TC,t+1}(G_{\leq t+1}) &\leq D \cdot 2^{C_{t}(G)} \bigl(\textsf{Cost}_{TC,t}(G_{\leq t})+O(n)\bigr).
\label{eq:all_costs_R_t}
\end{align}
Since $\Gamma = G \in \mathcal{G}$, we have from Claim~\ref{claim:good-iteration} and Claim~\ref{claim:many-good-trajectories} that
\begin{equation}\label{eq:sum_Caccess}
\sum_{t<T}C_{t}(G) \leq\xi\sum_{t<T} (\widetilde{d}_t + 1) =O(\log K)
\end{equation}
Iterating Eq.~\eqref{eq:all_costs_R_t} with initial condition of Eq.~\eqref{eq:cost_at_t0} gives
\begin{align}\label{eq:recursive_cost_R_t}
\textsf{Cost}_{S,t}(G_{\leq t}), \textsf{Cost}_{Q,t}(G_{\leq t}) & \leq (2D)^t 2^{\sum_{s<t} C_{s}(G)}, \quad 
\textsf{Cost}_{TC,t}(G_{\leq t}) \leq O(n)(2D)^t2^{\sum_{s<t}C_{s}(G)}.
\end{align}
Used the fact that $2^{\sum_{s<T}C_s(G)}=K^{O(1)}$ from Eq.~\eqref{eq:sum_Caccess} and $(2D)^T=(\log K)^{O(1)}$ from Claim~\ref{claim:many-good-trajectories} gives us 
$$
\textsf{Cost}_{S,T}(G)=\poly(K), \quad \textsf{Cost}_{Q,T}(G)=\poly(K), \quad \textsf{Cost}_{TC,T}(G)=O(n \cdot \poly(K)) \,\, \forall G \in \mathcal{G}.
$$
From Eq.~\eqref{eq:cost_R_traj}, we also have this applies to the cost not conditioned on the returned outcomes.

We have so far bounded the cost of one call to the implemented algorithmic access procedures in any iteration $t \in [T]$. We now bound the cumulative cost $\calC_R$ over the $T$ rounds to generate all the natural/bucket labels and thereby implement the required access procedures as part of the trajectory. Let $\calC_R^{(t)}$ be the portion of $\calC_R$ spent in iteration $t$ i.e., $\calC_R=\sum_{t<T}C_R^{(t)}$. Generating the natural labels in iteration $t+1$ require only constantly many calls to the access procedures from iteration $t$. Conditional on a fixed bucket label $j$, its label experiment uses fewer than $2^{j+1}$ current coin calls. Those multipliers are again bounded by $2^{C_t(G)}$. Hence, for an absolute constant $D_1 \geq 1$,
\begin{align}
\E[\calC_S^{(t+1)}\mid\Gamma=G] &\leq D_1 2^{C_t(G)}
\bigl(\textsf{Cost}_{S,t}(G_{\leq t}) +1\bigr),\nonumber\\
\E[\calC_Q^{(t+1)}\mid\Gamma=G] &\leq D_1 2^{C_t(G)}
\bigl(\textsf{Cost}_{Q,t}(G_{\leq t}) +1\bigr), \nonumber \\
\E[\calC_{TC}^{(t+1)}\mid\Gamma=G] &\leq D_1 2^{C_t(G)}
\bigl(\textsf{Cost}_{TC,t}(G_{\leq t})+O(n)\bigr). \nonumber
\end{align}
Substituting Eq.~\eqref{eq:recursive_cost_R_t} into the above equation and then summing over $t<T$ gives us
\begin{align}
\E[\calC_S\mid\Gamma=G],\
\E[\calC_Q\mid\Gamma=G]
&\leq D_2T(2D)^T2^{\sum_{s<T}C_s(G)} = \poly(K),\nonumber\\
\E[\calC_{TC}\mid\Gamma=G]
&\leq O(n)D_2T(2D)^T2^{\sum_{s<T}C_s(G)}
=O(n \cdot \poly(K)),
\end{align}
for an absolute constant $D_2 \geq 1$. Finally, averaging over the
conditional distribution of the good trajectory as
\begin{align*}
\E[\calC_R\mid\Gamma\in\mathcal G] = \E_{G\sim(\Gamma\mid\Gamma\in\mathcal G)} \left[\E[\calC_R\mid\Gamma=G]\right]
\end{align*}
gives us
$$
\E[\calC_S\mid\Gamma\in\mathcal G] = P_1(K), \quad \E[\calC_Q\mid\Gamma\in\mathcal G] = P_2(K), \quad \E[\calC_{TC}\mid\Gamma\in\mathcal G] = O(n \cdot P_3(K)),
$$
for some polynomials $P_1,P_2,P_3$. This completes the proof.
\end{proof}

In the previous claim, we bounded the expected complexities for trajectories executed by Algorithm~\ref{algo:ggmt-recursion} conditioned on them being good as per Claim~\ref{claim:many-good-trajectories}. We now show that by imposing a complexity budget on Algorithm~\ref{algo:ggmt-recursion}, we can still obtain an algorithm that succeeds with high probability.
\begin{claim}\label{claim:ggmt_tree_with_limited_budget}
Consider the context of Theorem~\ref{thm:ggmt-recursion}. There exists a complexity budget $B$ that can be imposed on Algorithm~\ref{algo:ggmt-recursion} to ensure that it always uses $\poly(K)$ samples from $U_A$, $\poly(K)$ membership queries to $1_A$, and $O(n \poly(K))$ time complexity such that with probability $\Omega(1/\polylog K)$,
the algorithm does not abort and returns exact algorithmic access to $X_T$
satisfying
\begin{align*}
\der{X_T}{X_T} \leq 2d_{\rm term}=O(1), \quad \der{U_A}{X_T} =O(\log K), \quad \Delta[X_T] \leq L_0=O(1).
\end{align*}
On this event, each later call to $\Sample_{X_T}$ or $\Coin_{X_T}$ has expected sample complexity $\poly(K)$, membership-query query complexity $\poly(K)$, and time complexity $O(n \cdot \poly(K))$.
\end{claim}
\begin{proof}
We will re-use the notation from Claim~\ref{claim:cost_good_trajectories} and its proof. Suppose we run Algorithm~\ref{algo:ggmt-recursion} for $T$ many iterations where $T$ is fixed according to Claim~\ref{claim:many-good-trajectories}. We then have that the cumulative complexities $\calC_S, \calC_Q, \calC_{TC}$ for generating a $T$-length trajectory (Claim~\ref{claim:cost_good_trajectories}) satisfies
\begin{equation}\label{eq:int1_up_cond_costs}
\E[\calC_S \mid \Gamma\in\mathcal G]\leq P_1(K),\qquad \E[\calC_Q\mid\Gamma\in\mathcal G]\leq P_2(K),\qquad \E[\calC_{TC}\mid\Gamma\in\mathcal G]\leq P_3(n,K).   
\end{equation}
Let us define the following budgets $B_S, B_Q, B_{TC}$ corresponding to the sample complexity, query complexity and time complexity, respectively
\begin{equation}\label{def:budgets_R}
B_S:=\lceil6P_1(K)\rceil,\qquad
 B_Q:=\lceil6P_2(K)\rceil,\qquad
 B_{T_C}:=\lceil6P_3(n,K)\rceil.   
\end{equation}
In Algorithm~\ref{algo:ggmt-recursion}, we thus define three counters from the start of iteration $1$ till the end of iteration $T$ corresponding to each resource and update them over the course of an iteration and trajectory. If an operation is to be taken that would make any counter exceed the corresponding budget as defined above, Algorithm~\ref{algo:ggmt-recursion} aborts and returns $\perp$.

We now comment on those trajectories whose complexities would not exceed the defined budgets of Eq.~\eqref{def:budgets_R} and satisfies Claim~\ref{claim:many-good-trajectories}. Using the conditional
Markov inequality, we have that for every $R\in\{S,Q,TC\}$,
\begin{equation}
\Prob[\calC_R > B_R \mid \Gamma\in\mathcal G] \leq \frac{\E[\calC_R\mid\Gamma\in\mathcal G]}{B_R} \leq \frac{1}{6}.
\end{equation}
By an union bound, we then have that
$$
\Prob\Big[\calC_S > B_S\ \text{or}\ \calC_Q>B_Q \ \text{or}\ \calC_{T_C} > B_{T_C} \mid\Gamma\in\mathcal G\Big] \leq \frac{1}{2}.
$$
We then obtain that
$$
\Prob\Big[\Gamma\in\mathcal G\text{ and no abort}\Big] = \Prob[\Gamma\in\mathcal G] \cdot
   \Prob[\text{no abort}\mid\Gamma\in\mathcal G] \geq \frac{1}{2} \Omega(1/\polylog K) = \Omega(1/\polylog K),
$$
where we used Claim~\ref{claim:many-good-trajectories}. Note that on this event, the complexities $\calC_R$ of Algorithm~\ref{algo:ggmt-recursion} will be within the budgets defined in Claim~\ref{claim:cost_good_trajectories}. Since $\Gamma \in \mathcal{G}$, we have from Claim~\ref{claim:many-good-trajectories} that 
$$
\der{X_T}{Y_T} \leq d_{\rm term}, \quad \Delta[X_T] \leq L_0, \quad \text{and} \quad \der{U_A}{X_T} = O(\log K).
$$
Finally, applying the entropic Ruzsa triangle inequality gives us
$$
\der{X_T}{X_T} \leq 2 \der{X_T}{Y_T} \leq 2 d_{\rm term}.
$$
This proves the desired results of $X_T$. The exact returned access and its expected per-call resource bounds follow from Claim~\ref{claim:cost_good_trajectories}. This proves the overall claim.
\end{proof}

\subsection{Putting everything together}
We are now ready to prove Theorem~\ref{thm:ggmt-recursion}.
\begin{proof}[Proof of Theorem~\ref{thm:ggmt-recursion}]
We use Algorithm~\ref{algo:ggmt-recursion} with initialized variables
of $X_0=Y_0=U_A$. The sample access for each is simply $U_A$ and the density coin for each is queries to $1_A$. We note that the slack $\Delta[X_0]=\Delta[Y_0]=0$ and $d_0 := \der{X_0}{Y_0} \leq \log K$. We run Algorithm~\ref{algo:ggmt-recursion} for $T=O(\log \log K)$ many iterations as in Claim~\ref{claim:many-good-trajectories} and under the budgeted complexity of Claim~\ref{claim:ggmt_tree_with_limited_budget}.

From Claim~\ref{claim:ggmt_tree_with_limited_budget}, we that have the promise that with probability $\Omega(1/\operatorname{polylog}K)$, the algorithm outputs $X_T$ along with exact algorithmic access to it, satisfying
$$
\der{X_T}{X_T} = O(1), \qquad \der{U_A}{X_T} = O(\log K), \qquad \Delta[X_T]=O(1).
$$
Applying the entropic PFR theorem (Theorem~\ref{thm:entropic-pfr}) to $X_T$ then gives us the promise that there exists a subspace $H\leq\mathbb F_2^n$ such that
$$
d[X_T;U_H]=O(d[X_T;X_T])=O(1)
$$
The algorithm uses the following complexity (Claim~\ref{claim:ggmt_tree_with_limited_budget}): $\poly(K)$ samples from
$U_A$, $\poly(K)$ membership queries to $1_A$, and $O(n \cdot \poly(K))$ time. From Claim~\ref{claim:many-good-trajectories}, we have that the algorithm implements exact algorithmic access to $X_T$ and the complexity of any call to $\Sample_{X_T},\Coin_{X_T}$ obey the same bounds as of the overall algorithm (Claim~\ref{claim:cost_good_trajectories}). This completes the proof.
\end{proof}

\section{Extracting a PFR subspace}
\label{sec:extract-pfr-subspace}
In Theorem~\ref{thm:ggmt-recursion} of Section~\ref{sec:ggmt-recursion}, we showed how using sample and query access to $A$ which is a set with doubling constant $K$, we could output algorithmic access (Definition~\ref{def:algo_access}) to a random variable $X$ for which there exists some subspace $H \leq \FF_2^n$ such that
$$
\der{X}{U_H} = O(1), \quad \der{U_A}{X}=O(\log K), \quad \text{and} \quad \Delta[X]=O(1).
$$
From the entropic Ruzsa triangle inequality, we also thus have $\der{X}{X} = O(1)$. Towards proving Theorem~\ref{thm:main}, our goal is to now learn the the PFR subspace $H$ for $A$. We do this in two steps:
\begin{enumerate}
\item \Cref{sec:terminalsubspace} takes $X$ such that $\der{X}{X} = O(1)$ and outputs a basis for a subspace $V \le \FF_2^n$ such that $\der{X}{U_V} = O(1)$ and therefore $\der{U_A}{U_V} = O(\log K)$.
\item Since the learned $V$ may be larger than $A$, in \Cref{sec:finalsubspaceoutput}, we show how to pass to a suitable $V'$ with $V' \leq |V$ and $|V'|\leq |A|$. Using samples and membership queries to $A$, we find a translate $a+V'$ and verify if $|A \cap (a+V')|/|A| \geq K^{-O(1)}$. 
\end{enumerate}

\subsection{Extracting a candidate subspace}
\label{sec:terminalsubspace}
Taking the span of samples from $X$ is too crude: even when $\der{X}{X}=O(1)$, that span can be much larger than $A$. We instead pass to the self-sum $X+X'$, retain only sums of sufficiently high probability, and span those heavy self-sums. The self-sum also removes irrelevant affine translations in the support of $X$.
The entropic PFR theorem gives a subspace $W_0$ with $\der{X}{U_{W_0}}=O(1)$, so all but a sufficiently small constant fraction of the mass of $X$ lies in $O(1)$ cosets of $W_0$. It turns out that enlarging $W_0$ by their coset directions gives an auxiliary subspace $W$ that contains every sufficiently heavy self-sum. These heavy self-sums also carry a constant fraction of the mass of $X+X'$. The span $V$ of $O(n)$ retained heavy self-sums therefore has constant index in $W$. By index, we mean that $[W:V] := |W|/|V| = O(1)$. It satisfies $\der{X}{U_V}=O(1)$, so the triangle inequality gives $\der{U_A}{U_V}=O(\log K)$. Additionally, the bound on the slack of $X$ helps in controlling how efficiently heavy self-sums can be detected. We now formally state the terminal subspace extraction result below for which we will use Algorithm~\ref{alg:terminal-subspace}.
\begin{lemma}
\label{lem:extract-subspaces}
Let $C \geq 0$ be a known absolute constant. Suppose that we are given algorithmic access to a random variable $X$ satisfying $\der{X}{X} + \Delta[X] \leq C$. Then, there is an algorithm that, with probability at least $0.99$, outputs a basis for a linear subspace $V\leq\FF_2^n$ satisfying $\der{X}{U_V}=O(1)$. The algorithm makes $O(n \log n)$ calls to $\Sample_X,\Coin_X$ and runs in $O(n^3)$ time.
\end{lemma}

\begin{myalgorithm}
\begin{algorithm}[H]
    \caption{Terminal subspace extraction}
    \label{alg:terminal-subspace}
    \setlength{\baselineskip}{1.8em} %
    \DontPrintSemicolon %
    \KwInput{Algorithmic access to $X$ (Definition~\ref{def:algo_access}) satisfying $\der{X}{X} + \Delta[X] \leq C$ where $C \geq 0$ is a known absolute constant.}
    \KwOutput{A basis for a linear subspace $V\leq\FF_2^n$ satisfying $\der{X}{U_V}=O(1)$ with probability at least $0.99$.}
    \vspace{2mm}
    Set parameters $\tau := 2^{-2(C+1)}$, $m_1=\ceil{16(n+8)}$ and $m_2=\ceil{48 \tau^{-2} \ln(400 m_1)}$. \\
    Use \Cref{lem:closure-sum} to construct algorithmic access to $Z=X+X'$, where $X'$ is an independent copy of $X$.  \\ %
    Draw $z_1,\ldots,z_{m_1}$ independently from $\Sample_Z$. For each $i$, let $\widehat{r}_i$ be the empirical mean of $m_2$ independent calls to $\Coin_Z(z_i)$, and retain $z_i$ if $\widehat r_i\geq\tau/2$. \\
    \textbf{Return} A basis for the span of the retained vectors.
\end{algorithm}
\end{myalgorithm}

Before proving Lemma~\ref{lem:extract-subspaces}, we first comment on the structural properties of the heavy self-sums.
\begin{claim}
\label{clm:heavy-self-sums}
Fix an absolute constant $C\geq0$. Suppose that $X$ has envelope $M_X$ and $\der{X}{X}+\Delta[X]\leq C$. For $Z=X+X'$, where $X'$ is an independent copy of $X$, set
\[
    \tau:=2^{-2(C+1)},
    \qquad
    r(z):=\frac{p_Z(z)}{M_X},
    \qquad
    V_0:=\Span\{z:r(z)\geq\tau/4\}.
\]
Then, $\Pr[r(Z)\geq\tau]\geq1/2$. Moreover, every subspace $V\leq V_0$ with $|V|\geq1/(8M_X)$ satisfies $\der{X}{U_V}=O(1)$.
\end{claim}
\begin{proof}
For every $z$, the envelope property gives $p_Z(z)=\sum_xp_X(x)p_X(x+z)\leq M_X\sum_xp_X(x)=M_X$. Hence $\log(1/r(Z))$ is nonnegative and has expectation $\entropy[Z]+\log M_X=\der{X}{X}+\Delta[X]\leq C$. Markov's inequality gives $\Pr[r(Z)<\tau] = \Pr[ \log(1/r(Z)) > 2(C+1)] \leq C/(2(C+1))<1/2$.

By \Cref{thm:entropic-pfr}, there exists a a subspace $W_0$ such that $\der{X}{U_{W_0}}=O(1)$. Let $\pi$ be the quotient map modulo $W_0$ and put $Q:=\pi(X)$. Contraction (\Cref{lem:basic-entropic-facts}) gives $\entropy[Q]/2=\der{Q}{0}\leq\der{X}{U_{W_0}}=O(1)$. Applying the entropy-tail bound in \Cref{lem:basic-entropic-facts} with $\varepsilon:=\tau/16$, choose a set $S$ such that $\Pr[Q\in S]\geq1-\varepsilon$ and $|S|\leq2^{\entropy[Q]/\varepsilon}=O(1)$. Set $W:=\pi^{-1}(\Span(S+S))$. Since $W/W_0=\Span(S+S)$ and $|S+S|\leq|S|^2$,
\[
    [W:W_0]=|\Span(S+S)|
    \leq2^{|S+S|}\leq2^{|S|^2}=O(1).
\]

The triangle inequality and the nested-subspace identity in \Cref{lem:basic-entropic-facts} now give
\[
    \der{X}{U_W}
    \leq\der{X}{U_{W_0}}+\der{U_{W_0}}{U_W}
    =\der{X}{U_{W_0}}+\frac12\log[W:W_0]
    =O(1).
\]
The entropy-difference bound in the same lemma states that $|\entropy[X]-\entropy[U_W]|\leq2\der{X}{U_W}$. Since $\entropy[U_W]=\log|W|$, it follows that
\[
    \log(M_X|W|)
    =\Delta[X]+\entropy[U_W]-\entropy[X]
    \leq C+2\der{X}{U_W}
    =O(1).
\]
We next show that every self-sum at the cutoff defining $V_0$ lies in $W$. This will make any sufficiently large $V\leq V_0$ have constant index in $W$. Let $E:=\{x:\pi(x)\notin S\}$. If $z\notin W$, then at least one of $x$ and $x+z$ lies in $E$ for every $x$: otherwise $\pi(z)=\pi(x)+\pi(x+z)\in S+S\subseteq\Span(S+S)$. Therefore
\[
\begin{aligned}
    p_Z(z)
    &\leq \sum_{x\in E}p_X(x)p_X(x+z)
       + \sum_{x:x+z\in E}p_X(x)p_X(x+z)\\
    &\leq 2M_X\Pr[X\in E]
    \leq \frac{\tau}{8}M_X.
\end{aligned}
\]
Thus every $z$ with $r(z)\geq\tau/4$ lies in $W$, and hence $V_0\leq W$.

If $V\leq V_0$ and $|V|\geq1/(8M_X)$, then
\[
    [W:V]\leq8M_X|W|=O(1),
    \qquad
    \der{X}{U_V}\leq\der{X}{U_W}+\frac12\log[W:V]=O(1),
\]
as required. This completes the proof.
\end{proof}

\begin{proof}[Proof of \Cref{lem:extract-subspaces}]
Let $C \geq 0$ be the absolute constant such that $\der{X}{X}+\Delta[X]\leq C$ and take $\tau := 2^{-2(C+1)}$ as in \Cref{clm:heavy-self-sums}. We also set the constants $m_1 := \ceil{16(n+8)}$ and $m_2 = \ceil{48 \tau^{-2} \ln(400 m_1)}$. Let $r,V_0$ be as in Claim~\ref{clm:heavy-self-sums} and let us define $H:=\{z:r(z)\geq\tau\}$. Conditional on $z_i$, the estimate $\widehat r_i$ is the mean of $m_2$ Bernoulli variables of mean $r(z_i)$. Hence the additive Chernoff bound and a union bound give
\[
    \Pr\!\left[\max_{i\in[m_1]}|\widehat r_i-r(z_i)|>\frac{\tau}{4}\right]
    \leq 2m_1\exp\!\left(-\frac{m_2\tau^2}{48}\right)\leq0.005,
\]
where the last inequality follows from our choices of $m_1,m_2$. On the complementary event, every proposal in $\calH$ is retained and every retained proposal lies in $V_0$. For $i\in[m_1]$, let $V_i$ be the span of the proposals in $H$ among $z_1,\ldots,z_i$. If $V$ is the output, then $V_{m_1}\leq V\leq V_0$.

For $i\in[m_1-1]$, while $|V_i|<1/(8M_X)$, independence of the next proposal and the envelope bound give
\[
    \Pr[z_{i+1}\in H\setminus V_i\mid z_1,\ldots,z_i]
    \geq \Pr[Z\in H]-\Pr[Z\in V_i]
    \geq \frac12-M_X|V_i|>\frac38.
\]
Every such proposal increases $\dim V_i$. Since $M_X\geq\lVert p_X\rVert_\infty\geq2^{-n}$, at most $n$ increases are needed to reach $|V_i|\geq1/(8M_X)$. Binomial domination and a Chernoff bound therefore give $\Pr[|V_{m_1}|\geq1/(8M_X)]\geq0.995$.

On the intersection of these events, $V\leq V_0$ and $|V|\geq|V_{m_1}|\geq1/(8M_X)$. The second conclusion of \Cref{clm:heavy-self-sums} therefore gives $\der{X}{U_V}=O(1)$, and the union bound gives success probability at least $0.99$.

The complexity of Algorithm~\ref{alg:terminal-subspace} is as follows. By \Cref{lem:closure-sum}, we have that each call to $\Sample_Z$ uses $2$ calls to $\Sample_X$, and each call to $\Coin_Z$ uses one call to $\Sample_X$ along with one call to $\Sample_X$. Since we make $m_1$ calls to $\Sample_Z$ and $m_1 m_2$ calls to $\Coin_Z$ in Algorithm~\ref{alg:terminal-subspace}, overall we make $O(m_1 m_2) = O(n \log n)$ many calls to $\Sample_X,\Coin_X$. The generation of proposals and estimation of densities $\widehat{r}_i$ consume $O(n)$ time complexity but the main contribution is in computing the output basis over $O(n)$ retained vectors of length $n$, which consumes $O(n^3)$ time using Gaussian elimination. This completes the proof.
\end{proof}

\subsection{Verifying a good PFR subspace}
\label{sec:finalsubspaceoutput}

We now have a linear subspace $V$ that is close to a good terminal distribution $X$ in entropic Ruzsa distance. A good GGMT trajectory also guarantees that $X$ is within $O(\log K)$ of $U_A$, so by triangle inequality we have 
$$
\der{U_A}{U_V} \le O(\log K).
$$
Since $\der{U_A}{U_V} \le O(\log K)$, the distribution of the $V$-coset of a uniform point of $A$ has small entropy, so a noticeable fraction ($1/\poly(K)$) of $A$ lies in a coset of $V$.
However, $\der{U_A}{U_V} \le O(\log K)$ does not guarantee the $|V| \le |A|$ promise of \Cref{thm:main}. Further, the algorithm chooses a good trajectory with probability $1/ \polylog K$, so to amplify the success probability via repeated trials, we need a way to detect whether the candidate $V$ obtained is indeed a good PFR subspace for $A$. To finish the algorithm, we do the~following:
\begin{enumerate}
    \item We pick an arbitrary subspace $V' \le V$ of codimension $O(\log K)$ such that either $V'=\{0\}$ or $|V'| \le |A|/2$. Now we have a small enough subspace $V'$. Further, the heavy $V$-coset splits into at most $\poly(K)$ cosets of $V'$, so there is a heavy $V'$-coset. We then observe that a random point $a \sim U_A$ is likely to lie in such a heavy $V'$-coset with constant probability.
    \item It remains only to recognize when this has happened. To this end, for a candidate $a+V'$, we estimate
\[
    \alpha:=\frac{|A\cap (a+V')|}{|A|}
    \qquad\text{and}\qquad
    \beta:=\frac{|A\cap (a+V')|}{|V'|}.
\]
The first quantity checks that $a+V'$ captures a noticeable fraction of $A$, while the relation $ \frac{\beta}{\alpha}=\frac{|A|}{|V'|}$ certifies that $V'$ is not larger than $A$. Both quantities can be estimated with $K^{O(1)}$ samples and membership queries, yielding an efficiently verifiable PFR candidate. We make this rigorous below.
\end{enumerate}

\begin{lemma}
\label{lem:detect-pfr-subspace}
Let $A\subseteq\FF_2^n$ be nonempty, let $R\geq1$, and let $0<\delta\leq1/2$. Given uniform-sample and membership-oracle access to $A$ and a basis for a subspace $V\leq\FF_2^n$, there is an algorithm that either rejects or accepts and outputs an affine coset $a+V'$ with $V'\leq V$ and $\dim V'=\dim V-O(R)$. It has the following guarantees.
\begin{itemize}
    \item \textbf{Completeness.} If $\der{U_A}{U_V}\leq R$, then the algorithm accepts with probability $\Omega(1)$.
    \item \textbf{Soundness.} Except with probability $\delta$, every accepted output satisfies
    \[
        |V'|\leq|A|,\qquad
        \frac{|A\cap(a+V')|}{|A|}\geq2^{-O(R)},\qquad
        \frac{|A\cap(a+V')|}{|V'|}\geq2^{-O(R)}.
    \]
\end{itemize}
The algorithm uses $2^{O(R)}\poly(n,\log(1/\delta))$ time, samples, and membership queries.
\end{lemma}

\begin{claim}
\label{clm:coset-localization}
Let $A\subseteq\FF_2^n$ be nonempty, let $R\geq 1$, and suppose that we are given a sampler for $U_A$ and a basis for a linear subspace $V\leq\FF_2^n$ satisfying $\der{U_A}{U_V}\leq R$. Then, there is an algorithm that, with probability $\Omega(1)$, outputs an affine coset $a+V'$, where $a\sim U_A$ and $V'\leq V$ satisfies $\dim V' = \dim V-O(R)$ and either $V'=\{0\}$ or $|V'|\leq |A|/2$, such that
\begin{align*}
    \frac{|A \cap (a+V')|}{|A|} \ge 2^{-O(R)}.
\end{align*}
Further the algorithm uses $\poly(n)$ time and one sample from $U_A$.
\end{claim}

\begin{proof}
Let $Q$ denote the random variable corresponding to the coset of $V$ containing a uniform sample from $A$, that is, $Q = U_A + V$. If $a \sim U_A$ and $v \sim U_V$, then $a+v$ conditioned on $Q = a + V$ is uniform on this coset of $V$. Therefore, by the entropy chain rule, we have that
\begin{align*}
    \der{U_A}{U_V}
    &= \entropy[U_A + U_V]
       - \frac{1}{2} (\entropy[U_A] + \entropy[U_V]) \\
    &= \entropy[Q] + \entropy[U_V]
       - \frac{1}{2}\log(|A||V|)
     = \entropy[Q] +\frac12\log\frac{|V|}{|A|}.
\end{align*}
Therefore, $|V|\leq2^{2R}|A|$. Further, $\entropy[U_A] = \entropy[U_A + U_V | U_V] \le \entropy[U_A + U_V]  = \entropy[Q]+\log|V|$, so the displayed identity also gives $\entropy[Q]\leq2R$. Let $V'\leq V$ be any subspace of codimension $r:=\min\{\dim V,\lceil2R\rceil+1\}$. Then either $V'=\{0\}$ or $|V'|\leq|A|/2$. If $Q'$ denotes the $V'$-coset containing a uniform sample from $A$, then $Q$ is determined by $Q'$ and each $V$-coset contains $[V:V']=2^r$ cosets of $V'$. Hence $\entropy[Q']\leq\entropy[Q]+r=O(R)$. Draw one sample $a\sim U_A$ and output $a+V'$. Its $U_A$-mass is $p_{Q'}(Q')=|A\cap(a+V')|/|A|$. Since $\E[\log(1/p_{Q'}(Q'))]=\entropy[Q']=O(R)$, Markov's inequality shows that this mass is at least $2^{-O(R)}$ with probability $\Omega(1)$.
\end{proof}

\begin{claim}[Verifying a candidate subspace]
\label{clm:verify-affine-candidate}
Let $A \subseteq \FF_2^n$, and let $0 < \eta \le 1$ and $0 < \delta < 1$. There is an algorithm that, given membership-oracle and uniform-sample access to $A$ and an affine representation of a candidate affine subspace $a+V'$, satisfies the following guarantees:
\begin{itemize}
    \item \textbf{Completeness.} The algorithm accepts with probability at least $1 -\delta$ if
    \begin{align*}
        \left(V'=\{0\}\ \text{or}\ |V'|\leq|A|/2\right)
        \quad \text{and} \quad
        \frac{|A\cap(a+V')|}{|A|}\geq\eta.
    \end{align*}
    \item \textbf{Soundness.} The algorithm rejects with probability at least $1 - \delta$ if $a+V'$ does not satisfy
    \begin{align*}
     |V'|\leq|A|,\qquad
     \frac{|A\cap(a+V')|}{|A|}=\Omega(\eta),\qquad
     \frac{|A\cap(a+V')|}{|V'|}=\Omega(\eta).
    \end{align*}
\end{itemize}
The algorithm uses $\poly(n,1/\eta,\log(1/\delta))$ time, samples, and membership queries.
\end{claim}

\begin{proof}
Given its affine representation, we can implement membership-oracle access as well as sample access for $a+V'$ in time $\poly(n)$. Using $m = \poly(1/\eta, \log (1/\delta))$ independent samples from $U_A$ and $U_{a+V'}$, we estimate the quantities
\[
 \alpha:=\frac{|A\cap(a+V')|}{|A|}
 =\Pr_{x\sim U_A}[x\in a+V'],
 \quad \text{ and } \quad
 \beta:=\frac{|A\cap(a+V')|}{|V'|}
 =\Pr_{y\sim U_{a+V'}}[y\in A]
\]
to get empirical estimates $\hat{\alpha}$ and $\hat{\beta}$. If $V'=\{0\}$, we accept exactly when $\hat{\alpha}\geq\eta/2$. Otherwise, we accept exactly when $\hat{\alpha} \ge \eta/2$ and $\hat{\beta} \ge 3\hat{\alpha}/2$. Completeness and soundness of this test can be argued as follows. By Chernoff's inequality and a union bound, with probability at least $1-\delta$, both estimates satisfy $|\hat{\alpha} - \alpha| \le \eta/16$ and $|\hat{\beta} - \beta| \le \eta/16$. If the completeness hypothesis holds, then $\hat{\alpha}\geq\eta/2$. If $V'=\{0\}$, the algorithm therefore accepts. Otherwise, $\beta=\alpha |A|/|V'|\geq2\alpha$, and
\begin{align*}
    \hat{\beta}\geq2\alpha-\frac{\eta}{16}
    >\frac32\left(\alpha+\frac{\eta}{16}\right)
    \geq\frac32\hat{\alpha},
\end{align*}
so the algorithm again accepts with probability at least $1-\delta$. Conversely, suppose that the algorithm accepts on the same concentration event. This implies that $\alpha \ge 7\eta/16$. If $V'=\{0\}$, then $a\in A$, so $\beta=1$ and $|V'|=1\leq|A|$. Otherwise, $\beta \ge 11 \eta/16$ and
\begin{align*}
    \beta - \alpha \ge \frac{1}{2} \hat{\alpha} - \frac{2 \eta}{16} > 0,
\end{align*}
so $|V'| / |A| = \alpha / \beta < 1$. Therefore, with probability at least $1-\delta$, every accepted candidate satisfies $|V'|\leq|A|$ and the required density bounds.
\end{proof}

\begin{proof}[Proof of \Cref{lem:detect-pfr-subspace}]
Apply \Cref{clm:coset-localization} to obtain an affine coset $a+V'$, and then apply \Cref{clm:verify-affine-candidate} to $a+V'$ with $\eta=2^{-O(R)}$ and failure probability $\delta$, accepting exactly when that test accepts.
 If $\der{U_A}{U_V}\leq R$, then \Cref{clm:coset-localization} shows with probability $\Omega(1)$ that $a+V'$ satisfies the completeness hypotheses of \Cref{clm:verify-affine-candidate}. Conditional on this event, the test accepts with probability at least $1-\delta$, proving completeness. The soundness conclusion follows directly from the soundness part of that claim.
 The localization step uses $\poly(n)$ time and one sample from $U_A$, while the test uses $2^{O(R)}\poly(n,\log(1/\delta))$ resources.
\end{proof}

\section{Proof of main theorem}
We prove \Cref{thm:main} by combining the GGMT traversal of \Cref{sec:ggmt-recursion} with extraction and verification from \Cref{sec:extract-pfr-subspace}. Each trial constructs algorithmic access to a random variable $X_T$ at the leaf of the GGMT tree, extracts a candidate subspace $V$, and verifies a coset $a+V'$ with $V'\leq V$. Verification lets us amplify the success probability by repeating these steps independently and accepting only cosets with the required size and overlap bounds. Since several subroutines have a rejection sampling component that might not terminate, we impose a sufficiently large polynomial budget on the time, sample, and membership-query complexity of each trial, aborting and starting a new trial if the budget is exceeded.

Fix sufficiently large effective universal constants $C_1,C_2,C_3$ and set
\[
    R:=C_1\log(2K),\qquad
    N:=\left\lceil C_3(\log(2K))^{C_2}\log\frac{2}{\delta}\right\rceil,
    \qquad
    \delta_{\rm ver}:=\frac{\delta}{2N}.
\]

\begin{myalgorithm}
\begin{algorithm}[H]
    \caption{Algorithmic PFR}
    \label{algo:algo-pfr}
    \setlength{\baselineskip}{1.5em}
    \DontPrintSemicolon
    \KwInput{Uniform sampling and membership-oracle access to a nonempty set $A\subseteq\FF_2^n$ satisfying $|A+A|\leq K|A|$, and $0<\delta<1$.}
    \KwOutput{An affine coset $a+V'$, represented by $a$ and a basis for $V'$, or failure.}
    \vspace{2mm}
    \RepTimes{$N$}{
        Run Algorithm~\ref{algo:ggmt-recursion} to obtain algorithmic access to $X_T$. If it aborts, start the next trial. \\
        Run Algorithm~\ref{alg:terminal-subspace} on $X_T$ with $C=2d_{\rm term}+L_0$, under the extraction budgets, to obtain a basis for $V$. If a budget is exceeded, start the next trial. \\
        Apply the verification procedure of \Cref{lem:detect-pfr-subspace} to $V$ with parameters $R$ and $\delta_{\rm ver}$. \\
        \If{verification accepts and outputs $a+V'$}{
            \textbf{return} $a$ and a basis for $V'$.
        }
    }
    \textbf{return} failure.
\end{algorithm}
\end{myalgorithm}

\begin{proof}[Proof of \Cref{thm:main}]
We first bound the success probability of one trial from below. By \Cref{claim:ggmt_tree_with_limited_budget}, with probability $\Omega(1/\polylog(2K))$, the traversal returns algorithmic access to $X_T$ satisfying
\[
    \der{X_T}{X_T}\leq2d_{\rm term},\qquad
    \der{U_A}{X_T}=O(\log K),\qquad
    \Delta[X_T]\leq L_0,
\]
where $d_{\rm term}$ and $L_0$ are the absolute constants of \Cref{claim:good-iteration}. Conditioning on this event, the hypothesis of \Cref{lem:extract-subspaces} holds with $C=2d_{\rm term}+L_0$, so extraction returns a basis for $V\leq\FF_2^n$ satisfying $\der{X_T}{U_V}=O(1)$ with probability at least $0.99$. Since it makes $\poly(n)$ calls to the algorithmic access of $X_T$, and each call has $\poly(n,K)$ complexity, for a sufficiently large polynomial budget, Markov's inequality shows that extraction succeeds within the budget with probability at least constant probability. When extraction succeeds, the triangle inequality gives
\[
    \der{U_A}{U_V}
    \leq\der{U_A}{X_T}+\der{X_T}{U_V}
    =O(\log(2K))\leq R.
\]
By the completeness guarantee of \Cref{lem:detect-pfr-subspace}, verification then accepts with conditional probability $\Omega(1)$, and therefore one trial accepts with probability $\Omega(1/\polylog(2K))$. By the choice of $N$, the probability that no trial accepts is at most
\begin{align*}
    \left(1-\frac{1}{C_3(\log(2K))^{C_2}}\right)^N
    &\leq\exp \left(-\frac{N}{C_3(\log(2K))^{C_2}}\right) \le \frac{\delta}{2}.
\end{align*}
The soundness guarantee of \Cref{lem:detect-pfr-subspace} applies to every candidate $V$, including those from unsuccessful traversals or extractions. Except with probability $\delta_{\rm ver}$ per verification call, an accepted coset $a+V'$ satisfies
\[
    |V'|\leq|A|,\qquad
    \frac{|A\cap(a+V')|}{|A|}\geq2^{-O(R)}=K^{-O(1)},
    \qquad
    \frac{|A\cap(a+V')|}{|V'|}\geq2^{-O(R)}=K^{-O(1)}.
\]
There are at most $N$ verification calls, so the probability that any accepted output violates these bounds is at most $N\delta_{\rm ver}=\delta/2$. Together with the preceding acceptance bound, this shows that the algorithm returns a coset satisfying all three bounds with probability at least $1-\delta$.

The imposed budgets for traversal and extraction bound the complexity of each trial by $\poly(n,K)$. Verification uses
\[
    2^{O(R)}\poly\!\left(n,\log\frac{1}{\delta_{\rm ver}}\right)
    =\poly\!\left(n,K,\log\frac{1}{\delta}\right)
\]
time, samples, and membership queries. This proves the desired complexity bounds in \Cref{thm:main}.

Finally, we prove the covering conclusion via a standard argument. Setting $B:=A\cap(a+V')$ and choosing a maximal collection of pairwise disjoint translates $t+B$ with $t\in A$. Since $B\subseteq A$, these translates lie in $A+B\subseteq A+A$, so their number is at most
\[
    \frac{|A+B|}{|B|}\leq\frac{K|A|}{|B|}=K^{O(1)}.
\]
For every $x\in A$, maximality ensures that $x+B$ intersects a chosen translate $t+B$. Since $B \subseteq a + V'$, and since we are in characteristic $2$, we have that $x\in t+B+B\subseteq t+V'$. The corresponding cosets $t+V'$ therefore cover $A$, proving the required conclusion for $V'$.
\end{proof}

\section{Quadratic Goldreich-Levin}
\begin{theorem}\label{thm:qgl}
    Let $\epsilon, \delta > 0$.
    There exists a universal constant $c_1, c_2 > 0$ and a randomized algorithm which, given oracle access to a function $f : \FF_2^n \to [-1,1]$, either outputs a quadratic form $q$ or $\bot$ satisfying:
    \begin{itemize}
        \item \textbf{Completeness:} If $\| f \|_{U^3} \ge \epsilon$, then with probability at least $1 - \delta$, the algorithm outputs a quadratic form $q$ such that
        \begin{align*}
            \E_x \left[f(x) \cdot (-1)^{q(x)}\right] \ge c_1 \epsilon^{c_2}.
        \end{align*}
        \item \textbf{Soundness:} For every input $f$, including those not satisfying the promise, the probability of outputting a quadratic form $q$ such that
        \begin{align*}
            \E_x \left[f(x) \cdot (-1)^{q(x)}\right] < c_1 \epsilon^{c_2}/2
        \end{align*}
        is at most $\delta$.
    \end{itemize}
    Further, the algorithm uses $\poly\left(n, {1}/{\epsilon}, \log {1}/{\delta} \right)$ queries and time.
\end{theorem}

The remainder of this section is devoted to proving \Cref{thm:qgl}. First, \Cref{lem:qgl-pfr-input} uses \cite[Section 4]{tulsiani2014quadratic} to extract from $f$ a function $\phi : \FF_2^n \to \FF_2^n$ and a set $S$ of pairs $(x, \phi(x))$ such that $|\widehat{f_x}(\phi(x))|$ is large for every $(x, \phi(x)) \in S$, and $S$ has small doubling constant. This guarantee is saying that $\phi(x)$ is a large Fourier coefficient of the derivative function $f_x$. When $f$ is close to quadratic, $\phi(x)$ is close to linear and therefore the set of pairs $(x, \phi(x))$ is close to a linear subspace that has high overlap with the set $S$.

Since $S$ has small doubling constant, we can use the PFR algorithm in \Cref{thm:main} to recover a subspace $W$ that covers $S$ with a small number of cosets. From $W$, we can use \Cref{lem:qgl-pfr-output-post-processing} to recover a quadratic $q$ that is correlated with $f$ in three steps:
\begin{enumerate}
  \item First, \Cref{clm:qgl-linearization} shows that the existence of a large set $S$ with large Fourier coefficients implies the existence of a linear map $T$ such that the average squared derivative Fourier coefficient $\E_x|\widehat{f_x}(Tx)|^2$ is large.
  \item Then, \Cref{clm:qgl-symmetrization} symmetrizes the linear map $T$ to get a symmetric zero-diagonal matrix $B$ such that $\E_x|\widehat{f_x}(Bx)|^2$ is large. 
  \item Finally, \Cref{clm:qgl-quadratic-recovery} converts the symmetric zero-diagonal matrix $B$ into a quadratic polynomial $q$ such that $\E_y \left[ f(y)(-1)^{q(y)}\right]$ is large.
\end{enumerate}
The algorithm described so far works for Boolean-valued functions $f : \FF_2^n \to \{-1, 1\}$; we use the rounding procedure of \cite{tulsiani2014quadratic} (see \Cref{lem:boolean-rounding}) to extend the algorithm to functions $f : \FF_2^n \to [-1, 1]$.

\begin{lemma}\label{lem:qgl-pfr-input}
Let $\epsilon > 0$ be an arbitrary real number. 
Let $f : \FF_2^n \to \{-1,1\}$ be a function such that $\| f \|_{U^3} \ge \epsilon$. 
Then there exists a randomized algorithm that uses $\poly(n, 1/\epsilon)$ time and queries to $f$ to provide exact membership-oracle access and independent uniform-sampling access to a set $S\subseteq \FF_2^n \times \FF_2^n$, for which the following properties hold with probability at least $\epsilon^{O(1)}$:
\begin{itemize}
  \item The set $S$ is large $|S| \ge  \epsilon^{O(1)} \cdot 2^n$ and has small doubling constant, i.e., $|S + S| \le K |S|$ for some $K = \epsilon^{-O(1)}$.
  \item There exists a function $\phi : \FF_2^n \to \FF_2^n$ such that $S \subseteq \{(x, \phi(x)) : x \in \FF_2^n\}$, and for every $(x, \phi(x)) \in S$, we have $|\widehat{f_x}(\phi(x))| \ge \epsilon^{O(1)}$.
\end{itemize}
Further, each call to the constructed membership and uniform-sampling oracles for $S$ uses $\poly(n, 1/\epsilon)$ expected time and queries to $f$.
\end{lemma}

\begin{proof}
By Tulsiani and Wolf~\cite[Section~4]{tulsiani2014quadratic}, given oracle access to $f$, there exists an algorithm that simulates oracle access to a function $\phi: \FF_2^n \to \FF_2^n$, and a randomized algorithm $\calS$ such that with probability at least $\epsilon^{O(1)}$, the following guarantees hold:
\begin{itemize}
  \item There exist sets $A^- \subseteq A^+ \subseteq \{(x, \phi(x)) : x \in \FF_2^n\}$ such that
  \begin{align*}
    |A^-|\geq\epsilon^{O(1)}2^n, \qquad |A^++A^+|\leq\epsilon^{-O(1)}2^n,
  \end{align*}
  and $|\widehat{f_x}(\phi(x))|\geq\epsilon^{O(1)}$ for every $(x, \phi(x)) \in A^+$.
  \item For all $(x, \phi(x)) \in A^-$, the algorithm $\calS$ accepts with probability at least $1 - 2^{-n-2}$, and for all $(x, \phi(x)) \notin A^+$, the algorithm $\calS$ accepts with probability at most $2^{-n-2}$.
\end{itemize}
Further, the algorithm that evaluates $\phi$, as well as $\calS$ use $\poly(n, 1/\epsilon)$ time and queries to $f$. For the rest of this proof, we will condition on the good event that the above guarantees hold.

We fix a random tape for $\calS$ so that it is deterministic, and define the set
\begin{align*}
  S = \{(x, \phi(x)) : x \in \FF_2^n \text{ and } \calS (x, \phi(x)) = 1\}.
\end{align*}
A point in $A^-$ is omitted from $S$ with probability at most $2^{-n-2}$, while a point outside $A^+$ is included in $S$ with probability at most $2^{-n-2}$. Therefore a union bound gives that conditioned on the good event, with probability at least $1 - 2^n \cdot 2^{-n-2} = \frac{3}{4}$, we have $A^- \subseteq S \subseteq A^+$. On this event, which occurs with probability $\epsilon^{O(1)}$, we have
\begin{align*}
  |S| &\ge |A^-| \ge \epsilon^{O(1)} 2^n, \\
  |S + S| &\le |A^+ + A^+| \le \epsilon^{-O(1)} 2^n \le \epsilon^{-O(1)} |S|,
\end{align*}
and for every $(x, \phi(x)) \in S$, since $S \subseteq A^+$, we have $|\widehat{f_x}(\phi(x))| \ge \epsilon^{O(1)}$.

Finally, we can implement a membership oracle for $S$ as follows: on input $(x, y)$, compute $\phi(x)$, and check if $y = \phi(x)$. If not, reject. Otherwise, return the value of $\calS(x, \phi(x))$. A uniform sampler can be implemented by repeatedly drawing a fresh uniform $x \in \FF_2^n$ and accepting exactly when $(x, \phi(x)) \in S$. On the good event above, the expected number of trials until acceptance is at most $\epsilon^{-O(1)}$. Each membership and sample query therefore uses $\poly(n, 1/\epsilon)$ time and queries to $f$.
\end{proof}

\begin{lemma}
\label{lem:qgl-pfr-output-post-processing}
Let $0<\alpha,\gamma, \delta \leq1$ be arbitrary real numbers, and let $L\geq1$.  Let $f:\FF_2^n\to\{-1,1\}$. Suppose that there is a set $S\subseteq \{(x, \phi(x)) : x \in \FF_2^n\} \subseteq \FF_2^n\times\FF_2^n$ for some function $\phi:\FF_2^n\to\FF_2^n$ such that $|S|\geq\alpha2^n$, and $|\widehat{f_x}(\phi(x))|\geq\gamma$ for every $(x,\phi(x))\in S$. Suppose further that there exists a linear subspace $W\leq\FF_2^n\times\FF_2^n$, $|W|\leq|S|$ such that $S$ is covered by at most $L$ cosets of $W$.

Then, there is a randomized algorithm which, given oracle access to $f$ and the basis for $W$, outputs with probability at least $1-\delta$, a polynomial $q:\FF_2^n\to\FF_2$ of degree at most two satisfying
\[
 \E_y \left[ f(y)(-1)^{q(y)} \right] \geq \frac{\alpha\gamma^2}{2L^2}.
\]
The algorithm uses $\poly(n,L,1/\alpha,1/\gamma,\log(1/\delta))$ time and queries to $f$.
\end{lemma}

\begin{claim}\label{clm:qgl-linearization}
Under the hypotheses of \Cref{lem:qgl-pfr-output-post-processing}, there is an algorithm that computes, from the basis of $W$, in $\poly(n)$ time, a linear map $T$ such that
\begin{align*} %
 \E_x\big|\widehat{f_x}(Tx)\big|^2
 \geq\frac{\alpha\gamma^2}{L^2}.
\end{align*}
\end{claim}

\begin{proof}
  Define $\pi(x, y) = x$ to be the first-coordinate projection, and define $\ker \pi_W := \{ w \in \FF_2^n : (0, w) \in W\}$. The algorithm first computes a basis $W$ of the form 
  \begin{align*}
    (0, w_1), \ldots, (0, w_k), (u_1, v_1), \ldots, (u_\ell, v_\ell),
  \end{align*}
  where $\{w_1, \ldots, w_k\}$ is a basis for $\ker \pi_W$, and $\{u_1, \ldots, u_\ell\}$ is a basis for $\pi(W)$. The algorithm then defines the linear map $T$ by setting $T(u_i) = v_i$ for $i \in [\ell]$, and extending $T$ to be zero on the complement of $\pi(W)$ in $\FF_2^n$. Therefore, for every $x \in \pi(W)$, we have that $(x, T(x)) \in W$. Such a linear map exists, because the vectors $u_i$'s are linearly independent. This can be done in $\poly(n)$ time using Gaussian elimination.

  Since $S$ is covered by at most $L$ cosets of $W$, there exists a coset $a + W$ such that $|S \cap (a + W)| \ge |S|/L$. Since $(S \cap (a+W)) \subseteq \{(x, \phi(x)) : x \in \FF_2^n\}$, $\pi$ is injective on $S \cap (a + W)$, and therefore $|S \cap (a+W)| \le |\pi(a+W)| = |\pi(W)| = |W| / |\ker \pi_W|$. Rearranging gives
  \begin{align}\label{eqn:small-kernel}
    |\ker \pi_W| \le |W| / |S \cap (a+W)| \le |S| / |S \cap (a+W)| \le L.
  \end{align}
  
  We now claim that linear map $T$ agrees with $\phi$ on (the first-coordinate projections of) many points in $S \cap (a + W)$, up to an affine translation. To see this, write $a = (a_1, a_2)$. Then, since we are in characteristic $2$, for every $(x, \phi(x)) \in S \cap (a + W)$, we have that $(x + a_1, \phi(x) + a_2) \in W$. This means that $x + a_1 \in \pi(W)$, and so $(x+a_1 , T(x + a_1)) \in W$. Since both $(x + a_1, \phi(x) + a_2)$ and $(x + a_1, T(x + a_1))$ are in $W$, their difference is in $\ker \pi_W$. Therefore, there exists some $w \in \ker \pi_W$ such that
  \begin{align*}
    \phi(x) + a_2 = T(x + a_1) + w = Tx + (Ta_1 + w + a_2).
  \end{align*}
  By \Cref{eqn:small-kernel}, there exists a fixed $w^* \in \ker \pi_W$ for which the above equation holds for at least $|S \cap (a + W)| / |\ker \pi_W| \ge |S| / L^2$ points $(x, \phi(x)) \in S \cap (a + W)$.
  Letting $b := Ta_1 + w^* + a_2$, we have concluded that for at least $|S| / L^2$ points $(x, \phi(x)) \in S$, we have that $\phi(x) = Tx + b$.

  Finally, we have that
  \begin{align*}
    \E_x\left|\widehat{f_x}(Tx)\right|^2 &\ge \E_x\left|\widehat{f_x}(Tx+b)\right|^2 = \frac{1}{2^n}\sum_{x \in \FF_2^n}
      \left|\widehat{f_x}(Tx+b)\right|^2 \ge \frac{1}{2^n}\cdot\frac{|S|}{L^2}\cdot\gamma^2 \ge \frac{\alpha\gamma^2}{L^2},
  \end{align*}
  where first inequality follows by \ref{clm:removing-affine-offset},   finishing the proof.
\end{proof}

\begin{claim}[Lemma 4.16 of \cite{tulsiani2014quadratic}, taken from the proof of Theorem 2.3 of \cite{samorodnitsky2007low}]
  \label{clm:qgl-symmetrization}
For a Boolean $f$ and $0<\eta\leq1$, any linear map $T$ satisfying
$\E_x|\widehat{f_x}(Tx)|^2\geq\eta$ can be converted in $\poly(n)$ time
into a symmetric zero-diagonal matrix $B$ satisfying
\begin{align*}
 \E_x\big|\widehat{f_x}(Bx)\big|^2\geq\eta^2.
\end{align*}
\end{claim}

\begin{claim}\label{clm:qgl-quadratic-recovery}
For a Boolean $f$ and $0<\eta,\delta\leq1$, if a symmetric zero-diagonal
matrix $B$ satisfies $\E_x|\widehat{f_x}(Bx)|^2\geq\eta^2$, then there
is a randomized algorithm which, given $B$ and oracle access to $f$,
outputs with probability at least $1-\delta$ a quadratic $q$ satisfying
\[
 \E_y \left[ f(y)(-1)^{q(y)} \right]\geq\frac\eta2.
\]
The algorithm uses $\poly(n,1/\eta,\log(1/\delta))$ time and queries.
\end{claim}

We use the same integration argument originally described in \cite{green2008inverse} and used in \cite[Lemma~4.17]{tulsiani2014quadratic}, followed by a Goldreich--Levin step (see \Cref{thm:goldreich-levin}) to recover the linear part of the quadratic. We give the algorithm and its analysis here for completeness.

\begin{proof}
The algorithm is as follows: given symmetric $B$ with zero diagonal, define the strictly upper triangular matrix $M$  such that $M + M^{\mathsf T} = B$. Then, define the function $$g(y) = f(y)(-1)^{\langle y, My \rangle}$$ and simulate oracle access to $g$ by making one query to $f$ and computing the quadratic phase. Next, run the Goldreich--Levin algorithm on $g$ with threshold $\eta/2$ and failure probability $\delta/2$. This returns a list of $O(\eta^{-2})$ frequencies which, with probability at least $1 - \delta/2$, contains every $\xi$ such that $|\widehat g(\xi)| \geq \eta/2$. Finally, using fresh random queries, the algorithm estimates the Fourier coefficient at every frequency in the list to additive error $\eta/8$, with total failure probability at most $\delta/2$, and selects a frequency $\xi$ with largest estimated absolute value. It then chooses $c \in \FF_2$ so that multiplying the estimated coefficient by $(-1)^c$ makes it nonnegative, and outputs
\[
 q(x):=\langle x,Mx\rangle+\langle\xi,x\rangle+c.
\]
If the list is empty, the algorithm outputs $\bot$. This completes the description of the algorithm.

We now prove correctness. We first show that $g$ must have a large Fourier coefficient.
\begin{align*}
  \max_{\xi\in\FF_2^n}|\widehat g(\xi)|^2 &= \left(\max_{\xi\in\FF_2^n}|\widehat g(\xi)|^2\right)
        \sum_{\xi\in\FF_2^n}\widehat g(\xi)^2\ge \sum_{\xi\in\FF_2^n}\widehat g(\xi)^4 = \E_x\left|\E_y \left[ g(y)g(y+x) \right]\right|^2 \ge \eta^2,
\end{align*}
where first equality is by Parseval's identity, second equality is by Claim~\ref{clm:qgl-autocorrelation-identity} and
 the last inequality follows from the hypothesis, and because
\begin{align*}
    \E_y \left[ g(y)g(y+x) \right] &= \E_y f_x(y) (-1)^{\langle y, My \rangle + \langle y+x, M(y+x) \rangle} \\
    &=\E_y f_x(y) (-1)^{\langle x, Mx \rangle + \langle y, (M + M^{\mathsf T})x \rangle} \\
    &=\E_y f_x(y) (-1)^{\langle x, Mx \rangle + \langle y, Bx \rangle} \\
    &= (-1)^{\langle x,Mx\rangle}\widehat{f_x}(Bx).
\end{align*}
Hence some Fourier coefficient of $g$ has magnitude at least $\eta$, and when the Goldreich--Levin call and subsequent estimates succeed, the selected frequency $\xi$ has estimated magnitude at least $7\eta/8$, and true magnitude at least $3\eta/4$. The estimated sign is correct, and so $(-1)^c\widehat g(\xi)=|\widehat g(\xi)| \ge \eta/2$, in particular.

Therefore, the output is a polynomial of degree at most two, and on these success events its correlation with $f$ is
\begin{align*}
 \E_y f(y)(-1)^{q(y)}
 &=(-1)^c\E_y g(y)(-1)^{\langle\xi,y\rangle} =(-1)^c\widehat g(\xi)
  =|\widehat g(\xi)|
  \geq\frac\eta2.
\end{align*}

The Goldreich--Levin call and the simultaneous estimates each fail with probability at most $\delta/2$. A union bound therefore gives overall success probability at least $1-\delta$. Constructing $M$ takes $\poly(n)$ time. Goldreich--Levin uses $\poly(n,1/\eta,\log(1/\delta))$ time and queries, and Chernoff's bound followed by a union bound over the $O(\eta^{-2})$ outputted frequencies shows that the subsequent Fourier estimation can be done within the same polynomial resource bound.
\end{proof}

\begin{proof}[Proof of \Cref{lem:qgl-pfr-output-post-processing}]
Setting $\eta:=\frac{\alpha\gamma^2}{L^2}$, \Cref{clm:qgl-linearization} computes a linear map $T$ with $\E_x\big|\widehat{f_x}(Tx)\big|^2 \ge \eta$. Applying \Cref{clm:qgl-symmetrization} to $T$ gives a symmetric zero-diagonal matrix $B$ with $\E_x\big|\widehat{f_x}(Bx)\big|^2 \ge \eta^2$.  Finally, \Cref{clm:qgl-quadratic-recovery} returns, with probability at least $1-\delta$, a quadratic $q$ whose correlation with $f$ is at least $\eta/2$. All these algorithms use $\poly(n,1/\eta,\log(1/\delta))$ time and queries.
\end{proof}

\begin{lemma}[Boolean rounding]\label{lem:boolean-rounding}
Let $0<\epsilon,\eta,\delta\leq1$.  Suppose there is a randomized algorithm which, given oracle access to any function $\widetilde f:\FF_2^n\to\{-1,1\}$ satisfying $\|\widetilde f\|_{U^3}\geq\epsilon/2$, outputs, with probability at least $1-\delta$, a quadratic polynomial $q$ such that
\[
\E_x[\widetilde f(x)(-1)^{q(x)}]\geq\eta.
\]
Then there is a randomized algorithm which, given oracle access to any function $f:\FF_2^n\to[-1,1]$ satisfying $\|f\|_{U^3}\geq\epsilon$, outputs, with probability at least $1-2\delta$, a quadratic polynomial $q$ such that
\[
\E_x [f(x)(-1)^{q(x)}]\geq {\eta}/{2}.
\]
The new algorithm has the same running time and query complexity as the Boolean algorithm, up to polynomial overhead in $n$, $1/\epsilon$, $1/\eta$, and $\log(1/\delta)$.
\end{lemma}

\begin{proof}
Independently for each $x\in\FF_2^n$, define a random Boolean function $\widetilde f$ by
\[
\Pr[\widetilde f(x)=1]=\frac{1+f(x)}{2}.
\]
The algorithm implements oracle access to $\widetilde f$ lazily: the first time $x$ is queried, it samples $\widetilde f(x)$ using one query to $f(x)$ and caches the result.  Thus repeated queries receive the same answer and the oracle has exactly the law of one fixed random Boolean function. We use the concentration argument from \cite[Lemma~3.3]{tulsiani2014quadratic}.  Since $\E[\widetilde f(x)]=f(x)$ and there are only $2^{O(n^2)}=\exp(o(2^n))$ quadratic phases, with probability at least $1-\delta$ the random function $\widetilde f$ simultaneously satisfies
\[
\|\widetilde f\|_{U^3}\geq\|f\|_{U^3}-{\epsilon}/{2}
\]
and, for every quadratic polynomial $q$,
\[
\left|\E_x\bigl(\widetilde f(x)-f(x)\bigr)(-1)^{q(x)}\right|
\leq{\eta}/{2}.
\]
The first inequality follows by applying concentration to the degree-eight polynomial defining the eighth power of the $U^3$ norm.  The second follows from Hoeffding's inequality for each fixed quadratic phase, followed by a union bound over all quadratic phases.

Assume that this concentration event occurs.  Since $\|f\|_{U^3}\geq\epsilon$, we have $\|\widetilde f\|_{U^3}\geq\epsilon/2$, so the Boolean algorithm returns, with probability at least $1-\delta$, a quadratic polynomial $q$ satisfying
\[
\E_x[\widetilde f(x)(-1)^{q(x)}]\geq\eta.
\]
For the same $q$, the uniform correlation bound above gives
\begin{align*}
\E_x[ f(x)(-1)^{q(x)}]
&\geq \E_x[\widetilde f(x)(-1)^{q(x)}]
   -\left|\E_x\bigl(\widetilde f(x)-f(x)\bigr)(-1)^{q(x)}\right| \\
&\geq\eta-\frac{\eta}{2}
=\frac{\eta}{2}.
\end{align*}
A union bound over the concentration event and the Boolean algorithm gives success probability at least $1-2\delta$. 
Finally, each new query to $\widetilde f$ uses one query to $f$, and cached values can be retrieved in polynomial time.  The concentration argument is used only in the analysis and requires no additional oracle queries, which proves the claimed complexity bound.
\end{proof}

\begin{proof}[Proof of Theorem~\ref{thm:qgl}]
First, we handle the case where $f:\FF_2^n\to\{-1,1\}$ is Boolean-valued. By \Cref{lem:qgl-pfr-input}, with probability $\epsilon^{O(1)}$ we construct oracle access to a set $S\subseteq\FF_2^{2n}$ such that
\[
 |S|\geq\alpha2^n,\qquad |S+S|\leq K|S|,
 \qquad \alpha,\gamma\geq\epsilon^{O(1)},
 \qquad K\leq\epsilon^{-O(1)},
\]
and every $(x,\phi(x))\in S$ satisfies $|\widehat{f_x}(\phi(x))|\geq\gamma$. Conditioning on the event above, we apply \Cref{thm:main} to $S$, with any fixed constant failure probability. It returns a subspace $W$ with $|W|\leq|S|$ such that $S$ is covered by at most $L=K^{O(1)}=\epsilon^{-O(1)}$ cosets of $W$. Finally, applying \Cref{lem:qgl-pfr-output-post-processing} with a fixed constant failure probability, we obtain a quadratic polynomial $q$ satisfying
\[
 \E_x f(x)(-1)^{q(x)}
 \geq {\alpha\gamma^2}/{2L^2}
 \geq \epsilon^{O(1)}.
\]
Thus one trial finds a quadratically correlated phase with probability $\epsilon^{O(1)}$ in expected cost $\poly(n,1/\epsilon)$. (The expected cost comes from sampling from $S$.) Truncating each trial to a sufficiently large polynomial number of steps, and repeating $\epsilon^{-O(1)}\log(1/\delta)$ independent trials, we obtain a quadratic with correlation at least $\epsilon^{O(1)}$ with probability at least $1-\delta/2$. Using fresh uniform queries, we can estimate the correlation of each candidate quadratic with $f$, obtaining the stated completeness and soundness guarantees for the Boolean-valued case. Finally, we bootstrap the Boolean-valued algorithm to functions $f:\FF_2^n\to[-1,1]$ by \Cref{lem:boolean-rounding}. Estimating the correlation of the returned quadratic with $f$ to additive error at most half of the Boolean correlation guarantee gives the stated completeness and soundness guarantees for the general case.  The total time and query complexity are $\poly(n,1/\epsilon,\log(1/\delta))$.
\end{proof}

\section{Improper agnostic tomography of stabilizer states}
In this section, we will present some applications of algorithmic PFR theorem to quantum learning. It was established in~\cite{SCforstabilizers,arunachalam2026tomography} that \emph{approximate} algorithmic PFR in polynomial time has implications for learning stabilizer-like states. We first state the main combinatorial approximate algorithmic PFR theorem and prove it, before  stating our applications. To state our theorem, we need the notion of approximately uniform distributions.
\begin{definition}[Approximately uniform distribution]
\label{def:flat-distribution}
Let $A$ be finite and nonempty. A distribution $\mu$ supported on $A$ is
$R$-uniform if for every $x\in A$, we have
\[
\frac{1}{R|A|}
    \leq \mu(x)
    \leq \frac{R}{|A|}
\]
\end{definition}

In approximate algorithmic PFR theorem, we do not have access to uniform samples from $A$, but rather have access to samples from $R$-uniform distributions supported on $A$ (observe that $R=1$ is precisely uniform sampling access to $A$). We show that as long as $R$ is polynomial in the other parameters, one can recover the same complexity as the usual algorithmic PFR theorem that we proved above. 

\begin{theorem}[Approximate algorithmic PFR]
\label{thm:approxmain}
Let $0<\delta<1$, and let $A\subseteq \F_2^n$ be nonempty and satisfy $|A+A|\leq K|A|$. Given  membership-oracle access to $A$ and independent samples from an $R$-uniform
distribution on $A$, there is a randomized algorithm that, with probability at least $1-\delta$, outputs a basis for a subspace $V\leq \F_2^n$ such that $|V| \le |A|$ and $A$ can be covered by at most $K^{O(1)}$ translates of $V$.
The algorithm uses $\poly(n,K,R,\log(1/\delta))$ time, samples, and membership queries.
\end{theorem}

\subsection{Algorithmic PFR  with approximate sampling access}

We first record that, when $A$ has small doubling, samples from an
$R$-uniform distribution can be efficiently converted into nearly uniform
samples from $A$.

\begin{lemma}
\label{lem:flat-uniformization}
Let $A\subseteq\F_2^m$ be nonempty with
$    |A+A|\leq K|A|,$
and let $\mu$ be an $R$-uniform distribution on $A$.
Given membership access to $A$ and independent samples from $\mu$, for
every $\zeta\in(0,1/2)$ one can sample from a distribution
$\widetilde U_A$ satisfying
\[
    \|\widetilde U_A-U_A\|_{\rm TV}\leq\zeta
\]
using
$  \widetilde{O}\!\left(
        {R^6K^2/}{\zeta^2}
    \right)$
samples from $\mu$ and membership queries.
\end{lemma}

\begin{proof}
Consider the following lazy random walk on $A$. Starting from $x\in A$, 
with probability $1/2$ remain at $x$; otherwise sample
$U,V\sim\mu$, set $y=x+U+V$, and move to $y$ if $y\in A$. Write
\[
    \nu(g):=\Pr_{U,V\sim\mu}[U+V=g].
\]
For distinct $x,y\in A$, the transition probability is
\[
    P(x,y)=\frac12\nu(x+y)=P(y,x),
\]
so the uniform distribution $U_A$ is stationary and the chain is
reversible. Moreover, since
\[
    \mu(x)\leq \frac{R}{|A|}=R\,U_A(x),
\]
the initial distribution $\mu$ is an $R$-warm start with respect to
$U_A$ by definition.   For
$B,C\subseteq A$, define the stationary flow from $B$ to $C$ by
\[
    Q(B,C)
    :=
    \sum_{x\in B}\sum_{y\in C}U_A(x)P(x,y).
\]
Thus $Q(B,A\setminus B)/U_A(B)$ is the probability that one step leaves
$B$ when the current state is uniform on $B$.  A small value means that
$B$ can trap the walk for a long time.  To allow
us to ignore very small sets, for $s\in(0,1/2)$ define the \emph{$s$-conductance} as
\[
    \Phi_s
    :=
    \min_{\substack{B\subseteq A\\ s<U_A(B)\leq1/2}}
    \frac{Q(B,A\setminus B)}{U_A(B)-s}.
\]
Below we use a  standard warm-start mixing bound shows that a lower bound on
$\Phi_s$, together with the $R$-warmness above, gives an explicit bound
on the number of steps needed for the walk to approach $U_A$ in total
variation distance.  It therefore remains to lower-bound $\Phi_s$.
 For $B\subseteq A$, write
$C=A\setminus B$, $N=|A|$, and $b=|B|/N$. Let
\[
    r_{B,C}(g)
    :=
    |\{(x,y)\in B\times C:x+y=g\}|.
\]
Since $\mu$ is $R$-uniform, every $u\in A$ satisfies
$   \mu(u)\geq \frac{1}{RN}.$ 
Hence, for any $g\in\F_2^m$,
\[
\begin{aligned}
    \nu(g)
    =\Pr_{U,V\sim\mu}[U+V=g] =\sum_{\substack{u,v\in A\\u+v=g}}\mu(u)\mu(v) \geq
    \sum_{\substack{u,v\in A\\u+v=g}}
    \frac{1}{R^2N^2}
    =
    \frac{r_A(g)}{R^2N^2}.
\end{aligned}
\]
Moreover, every representation $g=x+y$ with
$(x,y)\in B\times C$ gives the distinct reversed representation
$g=y+x$ with $(y,x)\in C\times B$.  Since $B$ and $C$ are disjoint,
these two ordered representations are distinct.  Therefore
\[
    r_A(g)
    \geq
    r_{B,C}(g)+r_{C,B}(g)
    =
    2r_{B,C}(g),
\]
and hence
\[
    \nu(g)
    \geq
    \frac{2r_{B,C}(g)}{R^2N^2}.
\]
 Hence the stationary flow across the cut satisfies
\begin{align}
\label{eq:lowerboundonQBC}
    Q(B,C)
    =
    \frac1{2N}\sum_g r_{B,C}(g)\nu(g)
    \geq
    \frac1{R^2N^3}\sum_g r_{B,C}(g)^2.
\end{align}
The support of $r_{B,C}$ is contained in
$B+C\subseteq A+A$, and therefore Cauchy--Schwarz gives
\[
    \sum_g r_{B,C}(g)^2
    \geq
    \frac{|B|^2|C|^2}{|B+C|}
    \geq
    \frac{|B|^2|C|^2}{KN}.
\]
Plugging into Eq.~\eqref{eq:lowerboundonQBC} gives
\[
    Q(B,C)
    \geq
    \frac{b^2(1-b)^2}{R^2K}.
\]
Recall that we proved
\[
    Q(B,A\setminus B)
    \geq
    \frac{b^2(1-b)^2}{R^2K},
    \qquad
    b:=\frac{|B|}{N}.
\]
Now fix a set $B$ appearing in the definition of $\Phi_s$, so that
$    s<b\leq\frac12.$ 
Since $b\leq1/2$, we have $1-b\geq1/2$, and therefore
$    (1-b)^2\geq\frac14.$ 
Hence
\[
\begin{aligned}
    \frac{Q(B,A\setminus B)}{b-s}
    &\geq
    \frac{b^2(1-b)^2}{R^2K(b-s)} \geq
    \frac{b^2}{4R^2K(b-s)}\geq \frac{s}{R^2 K}
\end{aligned}
\]
where the final inequality used  $    b^2-4s(b-s)
    =
    (b-2s)^2
    \geq0$. 
Therefore every set $B$ in the minimization defining $\Phi_s$ satisfies
\[
    \frac{Q(B,A\setminus B)}{b-s}
    \geq
    \frac{s}{R^2K}.
\]
Taking the minimum over all such $B$ gives
\[
    \Phi_s
    \geq
    \frac{s}{R^2K}.
\]
We will use the following standard warm-start form of the
Lov\'asz--Simonovits mixing bound~\cite{lovaszsimonovits1993}. 
\begin{lemma}[Warm-start mixing bound]
\label{lem:warm-start-mixing}
Let $P$ be a finite lazy reversible Markov chain with stationary
distribution $\pi$. Suppose the initial distribution $\sigma$ is
$M$-warm with respect to $\pi$, meaning \text{for every state }x,
$    \sigma(x)\leq M\pi(x).$ 
For $s\in(0,1/2)$, let
\[
    \Phi_s
    :=
    \min_{\substack{B: s<\pi(B)\leq1/2}}
    \frac{Q(B,B^c)}{\pi(B)-s},
\]
where
$    Q(B,C)
    :=
    \sum_{x\in B}\sum_{y\in C}\pi(x)P(x,y).$ 
Then, for every $t\geq0$,
\[
    \|\sigma P^t-\pi\|_{\rm TV}
    \leq
    Ms
    +
    M\exp\!\left(-\frac{\Phi_s^2t}{2}\right).
\]
\end{lemma}
As we saw above, $\mu$ is $R$-warm with respect to $U_A$, we may apply
\Cref{lem:warm-start-mixing}.  Taking $s=\zeta/(2R)$ and using
\[
    \Phi_s\geq\frac{s}{R^2K}
    =
    \frac{\zeta}{2R^3K},
\]
we obtain
\[
    \|\mu P^t-U_A\|_{\rm TV}
    \leq
    \frac{\zeta}{2}
    +
    R\exp\!\left(
        -\frac{\zeta^2t}{8R^6K^2}
    \right).
\]
Hence the total-variation distance is at most $\zeta$ once
\[
    t
    \geq
    \frac{8R^6K^2}{\zeta^2}
    \log\frac{2R}{\zeta}.
\]
Thus $t=poly(R,K,1/\zeta,\log(R/\zeta))$ steps suffice. Each step uses two fresh samples from $\mu$ and one query.
\end{proof}

We can now transfer any polynomial-time PFR algorithm requiring uniform
samples to the approximately uniform-sample setting.

\begin{proof}[Proof of Theorem~\ref{thm:approxmain}]
Run our main algorithmic PFR results algorithm, replacing each requested uniform
sample from $A$ by an independent output of
\Cref{lem:flat-uniformization}.
Let
\[
    Q:=(mK+2)^{C_{\rm PFR}}
\]
be an upper bound on the total number of samples requested by the PFR
algorithm, and generate each replacement sample to total-variation
accuracy
$ \zeta:=1/{100Q}$.  
By the coupling characterization of total variation distance, each
replacement sample can be coupled to an exactly uniform sample so that
they disagree with probability at most $\zeta$. A union bound over the
at most $Q$ sample requests therefore couples the entire execution to
the exact-uniform execution with failure probability at most
$   Q\zeta\leq1/{100}.$
Consequently the success probability is at least
$2/3-1/100\geq 3/5$.  By \Cref{lem:flat-uniformization}, each replacement sample uses
\[
    O\!\left(
        R^6K^2Q^2\log(RQ)
    \right)
\]
operations and oracle calls. Since $Q$ itself is polynomial in $n$ and
$K$, the total complexity is bounded by
$\poly(n,K,R)$
 The structural
guarantee on $H$ is exactly that supplied by our algorithmic PFR theorem.
\end{proof}

\subsection{Quantum implications}
We present some implication of the algorithmic PFR theorem in this section.
\subsubsection{Self-correction.} In a recent work,~\cite{SCforstabilizers} considered the task of \emph{self-correction}, a weakening of agnostic learning. Recall that in agnostic learning, for an unknown state $\ket{\psi}$ which is promised to be $\opt$-close to a state in a class $\mathcal{C}$, the goal is, given copies of $\ket{\psi}$ output an element of $\mathcal{C}$ which is $\opt-\varepsilon$ close to $\ket{\psi}$. In the self-correction task, they consider if the task becomes easier if the goal was to output an element in $\mathcal{C}$ which is $\poly(\opt)$ close to $\ket{\psi}$. To this end, they proved the following theorem.\footnote{We remark that in~\cite{SCforstabilizers}, they invoke the algorithmic PFR under the sampling access arising from its BSG reduction reduction step, while our theorem is stated for approximately-uniform sampling distributions. This causes no essential difference since one can use the quantitative bounds in the BSG construction in their paper to imply that, after conditioning on the relevant set, the induced sampling distribution is \(\gamma^{-O(1)}\)-flat. Thus our approximate PFR algorithm applies with only \(\poly(1/\gamma)\) overhead, and the remainder of the self-correction argument in~\cite{SCforstabilizers} goes through unchanged.}
\begin{theorem}[\cite{SCforstabilizers}]
\label{thm:SC}
    Let $\uptau>0$. Let~$\ket{\psi}$ be an unknown $n$-qubit quantum state such that 
    $$
    \max_{\ket{\phi}\in \textsf{stabilizer}}|\langle \psi|\phi\rangle|^2 \geq \uptau.$$ Assuming the approximate algorithmic PFR conjecture, there is a protocol that with probability $1-\delta$, outputs a $\ket{\phi}\in \textsf{stabilizer}$ such that $|\langle\phi |\psi\rangle|^2 \geq  \tau^C$ (for a universal constant $C>1$) using $\poly(n,1/\uptau,\log(1/\delta))$ time and copies of~$\ket{\psi}$.
\end{theorem}
Using our approximate algorithm PFR result in Theorem~\ref{thm:approxmain}, we remove the assumption in their work 
\begin{corollary}\label{cor:poly-time-SC}
 The guarantees of Theorem~\ref{thm:SC} hold without any assumptions.
\end{corollary}

\subsubsection{Improper agnostic tomography.} In a later work,~\cite{arunachalam2026tomography} provided a boosting framework that took the self-correction procedure to give an algorithm for learning stabilizer decompositions of any quantum state $\ket{\psi}$ given copies. Particularly, the following is now true as a consequence of Corollary~\ref{cor:poly-time-SC} applied to \cite[Theorem~1.2]{arunachalam2026tomography}.
\begin{corollary}
Let $\varepsilon,\delta \in (0,1)$. Suppose $\ket{\psi}$ is an unknown $n$-qubit state. Then, there is an algorithm that uses copies of $\ket{\psi}$ and with probability $\geq 1-\delta$, outputs a list of $\kappa=\poly(1/\varepsilon)$ many states $\{\ket{\phi_i}\}_{i \in [\kappa]}$ and coefficients $\beta \in \mathbb{C}^k$ such that $\ket{\psi}$ can be expressed (up to a global phase) as
$$
\ket{\psi} = \sum_{i=1}^\kappa \beta_i \ket{\phi_i} + \alpha \ket{\phi_R}, \text{ where } |\alpha|^2 \cdot \max_{\ket{\varphi} \in \textsf{stabilizer}} |\langle \varphi | \psi \rangle|^2 \leq \varepsilon.
$$
The sample and time complexity of this algorithm is $\poly(n,1/\varepsilon)$.
\end{corollary}

Moreover, this implies a polynomial-time agnostic learner for the class of stabilizer states, albeit the output of the agnostic learner is \emph{improper}, i.e., need not be a stabilizer state itself. The following was shown in \cite{arunachalam2026tomography}.
\begin{theorem}[{\cite[Theorem~1.3]{arunachalam2026tomography}}]
\label{thm:improperagn}
Let $\delta, \varepsilon \leq \opt \in (0,1)$. Suppose $\ket{\psi}$ is an unknown $n$-qubit state with (unknown) optimal stabilizer fidelity 
$$
\opt=\max_{\ket{\phi}\in \textsf{stabilizer}}|\langle \psi|\phi\rangle|^2 
$$
Assuming that Self-correction is solvable in $\poly(n,1/\varepsilon)$ time, there is an algorithm, that with probability~$\geq 1-\delta$,  uses copies of $\ket{\psi}$ to output a state $\ket{\phi}$, which is a sum of $\kappa = \poly(1/\varepsilon)$ many $\ket{\phi_i}\in \textsf{stabilizer}$ such~that
$$
|\langle \phi | \psi \rangle|^2 \geq \opt - \varepsilon,
$$
The sample and time complexity of the algorithm is $\poly(n,1/\varepsilon,\log 1/\delta)$.
\footnote{We remark that their main result is stated in terms of \emph{model class}, which has three requirements: $(i)$ the class is efficiently representable, $(ii)$ it admits a self-correction algorithm, $(iii)$ there is an efficient circuit that can prepare the states in the class. For the class of stabilizer state, $(i,iii)$ are well-known in literature and only $(ii)$ was a bottleneck, which is stated explicitly in this theorem statement.} 
\end{theorem}
Since Theorem~\ref{thm:approxmain} gave a polynomial-time algorithm for self-correction, we also obtain an improper agnostic learner for stabilizer states. 
\begin{corollary}
\label{cor:agnosticimproper}
 The guarantees of Theorem~\ref{thm:improperagn} hold without any assumptions.
\end{corollary}

\subsubsection{Learning states with bounded extent} Finally, in~\cite{arunachalam2026tomography} they also used the boosting framework on top of self-correction to obtain a tomography algorithm for the class of states with stabilizer extent $\xi$.  For an arbitrary state $\ket{\psi}$, we define the stabilizer extent of $\ket{\psi}$as\footnote{We remark that in literature~\cite{bravyi2016trading,bravyi2019simulation} extent was defined the \emph{squared} $\ell_1$ norm. For the purposes of this paper, the squared $\ell_1$ norm or the $\ell_1$ norm as defined below will not change the main results.}
\begin{equation}\label{def:extent_C}
\xi(\ket{\psi})=\min\Big\{ \sum_i |c_i| : \ket{\psi}=\sum_i c_i \ket{\phi_i}, \ket{\phi_i}\in \textsf{stabilizer}\Big\}.  
\end{equation}

\begin{theorem}[\cite{arunachalam2026tomography}]
\label{result:learn_stab_extent}
Let $\xi > 0,\varepsilon \in (0,1)$. Suppose $\ket{\psi}$ is an unknown $n$-qubit state with stabilizer extent $\xi(\ket{\psi})\leq \xi$. Then, there exists an algorithm that learns $\ket{\psi}$  up to trace distance $\varepsilon$ in time and sample complexity $\poly(n,\xi,1/\varepsilon)$  \emph{assuming} an algorithmic~$\textsf{PFR}$~conjecture. 
Furthermore, every $0.1$-error tomography protocol   needs $\Omega(n+\xi^2)$ copies.
\end{theorem}

Using our approximate algorithm PFR result in Theorem~\ref{thm:approxmain}, we remove the assumption in their work 
\begin{corollary}
The guarantees of Theorem~\ref{result:learn_stab_extent} hold without any assumptions.
\end{corollary}

\subsubsection{High stabilizer-dimension states}
We have so far described tomography results for stabilizer states or quantum states that admit stabilizer structure. We now turn our attention to states with high stabilizer dimension. We say that an $n$-qubit pure quantum state $\ket{\psi}$ has stabilizer dimension of $k$ if $\ket{\psi}$ is stabilized by an Abelian group of $2^k$ Pauli operators and denote this class of states as $\mathcal{S}(k)$. We then have the following self-correction protocol for states with stabilizer dimension $n-t$ i.e., $\calS(n-t)$.
\begin{corollary}\label{cor:poly-time-T-doped}
Let $\varepsilon,\delta \in (0,1)$ and $0 \leq t \leq n$. Suppose that $\ket{\psi}$ is an unknown $n$-qubit state with the promise $\max_{\ket{\varphi} \in \calS(n-t)} |\langle \varphi | \psi \rangle|^2 \geq \tau$. Then, there is an algorithm that uses copies of $\ket{\psi}$ to output a stabilizer state $\ket{\phi}$ such that with probability $\geq 1-\delta$,
$$
|\langle \phi | \psi \rangle|^2 \geq \poly(2^{-t} \tau).
$$
The algorithm consumes sample and time complexity $\poly(n,2^t,1/\tau)$.
\end{corollary}
We quickly remark that the promise of $\poly(2^{-t} \tau)$ is a consequence of the local inverse theorem of Gowers-$3$ norm of quantum states and refer the reader to \cite[Section~5]{SCforstabilizers}.

Using the boosting framework of \cite{arunachalam2026tomography}, we can then use the above self-correction result to obtain an improper agnostic learner of states with high stabilizer dimension.
\begin{corollary}
Let $\delta, \varepsilon \leq \opt \in (0,1)$ and $0 \leq t \leq n$. Suppose $\ket{\psi}$ is an unknown $n$-qubit state with (unknown) optimal fidelity 
$$
\opt=\max_{\ket{\phi}\in \calS(n-t)} |\langle \psi|\phi\rangle|^2. 
$$
Then, there is an algorithm that with probability~$\geq 1-\delta$,  uses copies of $\ket{\psi}$ to output a state $\ket{\phi}$, which is a sum of $\kappa = \poly(2^t/\varepsilon)$ many $\ket{\phi_i}\in \textsf{stabilizer}$ such~that
$$
|\langle \phi | \psi \rangle|^2 \geq \opt - \varepsilon,
$$
The sample and time complexity of the algorithm is $\poly(n,2^t,1/\varepsilon,\log 1/\delta)$.
\end{corollary}
Note that the previously best known result for agnostic learning of states in $\calS(n-t)$ had time complexity $\poly(n (2^t/\varepsilon)^{O(\log(1/\varepsilon))})$~\cite{chen2025stabilizer} and we are able to remove this quasipolynomial dependence on $\varepsilon$. Moreover, Corollary~\ref{cor:poly-time-T-doped} can be utilized to learn states which have bounded extent $\xi$ with respect to $\calS(n-t)$ in $\poly(n\cdot \xi/\varepsilon)$ time.

\subsubsection{Learning Clifford structure in unitaries and Hamiltonians}
We now describe some implications of the polynomial-time self-correction protocol of Corollary~\ref{cor:poly-time-SC} for agnostic tomography of Clifford unitaries and tomography of Clifford-structured objects. Note that Clifford unitaries have previously been shown to be efficiently testable~\cite{gross2021schur,hinsche2025clifford}. We briefly introduce some relevant notation. We will utilize the Choi–Jamiolkowski isomorphism and denote the Choi state of an $n$-qubit unitary $U$ as $\choiket{U} := (U \otimes I)\ket{\mathrm{EPR}_n}$ where $\ket{\mathrm{EPR}_n} = 2^{-n/2} \sum_{x \in \FF_2^n} \ket{x}\ket{x}$ is the state over $n$ many $\mathrm{EPR}$ pairs. We will say that $U$ has Clifford fidelity $\opt$ if $\max_{V \in \textsf{Clifford}} |\langle\!\langle V \choiket{U}|^2 = \opt$.

As a consequence of Corollary~\ref{cor:poly-time-SC} and Theorem~A.1 of \cite{dutt2026learning}, the following is now true regarding polynomial-time self-correction of $n$-qubit Clifford unitaries.
\begin{corollary}\label{cor:poly-time-SC-for-Cliffs}
Let $\tau,\delta \in (0,1)$. Given an unknown $n$-qubit unitary $U$ that has Clifford fidelity $\geq \tau$, there is a quantum algorithm that with probability $\geq 1-\delta$ outputs a Clifford unitary $V$ such that $|\langle\!\langle V \choiket{U}|^2 \geq \poly(\tau)$. The algorithm uses $\poly(n,1/\tau)$ query and time complexity.
\end{corollary}

We can then use the above result to give polynomial-time algorithms for learning Clifford decompositions of unitaries and learning unitaries with bounded Clifford extent. We will not give the details here and refer the reader to \cite[Section~4]{dutt2026learning}. We however quickly state the following implication regarding tomography of $n$-qubit Hamiltonians which are Clifford-structured, obtained from applying Corollary~\ref{cor:poly-time-SC-for-Cliffs} to \cite[Theorem~4.4]{dutt2026learning}.
\begin{corollary}
Let $\varepsilon,\delta \in (0,1)$ and $\xi \geq 1$. Suppose $H$ is an unknown $n$-qubit Hamiltonian wwith the promise $\xi = \min\{\norm{c}_1 : H = \sum_i c_i V_i \text{ where } V_i \in \textsf{Cliffords}\}$ and satisfies $\Tr[H]=0$. Then, there is an algorithm that given query access to $\exp(-iHt)$ outputs $\widehat{H}$ expressable as a linear combination of $\poly(\xi/\varepsilon)$ many Clifford unitaries with probability $\geq 1-\delta$ such that
$$
\|\widehat{H} - H \|_{\overline 2} \leq \varepsilon,
$$
where $\norm{A}_{\overline{2}} = \sqrt{2^{-n} \Tr[A^\dagger A]}$ is normalized Frobenius norm. The algorithm uses $t = \poly(\varepsilon/\xi)$ and total time evolution of $\poly(n \cdot \xi/\varepsilon \log(1/\delta))$. The algorithm uses query and time complexity $\poly(n \cdot \xi/\varepsilon \log(1/\delta))$.
\end{corollary}

\bibliographystyle{alpha}
\bibliography{refs}

\newcommand{\etalchar}[1]{$^{#1}$}
\begin{thebibliography}{CFMdW10}

\bibitem[ACDG26]{arunachalam2026classical}
Srinivasan Arunachalam, Davi {Castro-Silva}, Arkopal Dutt, and Tom Gur.
\newblock Classical and {Q}uantum {P}olynomial {F}reiman-{R}uzsa {A}lgorithms.
\newblock In {\em 17th Innovations in Theoretical Computer Science Conference (ITCS 2026)}, pages 11--1. Schloss Dagstuhl--Leibniz-Zentrum f{\"u}r Informatik, 2026.

\bibitem[AD25]{ad2024tolerant}
Srinivasan Arunachalam and Arkopal Dutt.
\newblock Polynomial-time tolerant testing stabilizer states.
\newblock In {\em Proceedings of the 57th Annual ACM Symposium on Theory of Computing}, STOC '25, page 1234–1241. Association for Computing Machinery, 2025.

\bibitem[AD26a]{SCforstabilizers}
Srinivasan Arunachalam and Arkopal Dutt.
\newblock Learning stabilizer structure of quantum states.
\newblock In {\em Proceedings of the 58th Annual ACM Symposium on Theory of Computing}, STOC '26, page 1949–1959, New York, NY, USA, 2026. Association for Computing Machinery.

\bibitem[AD26b]{arunachalam2026tomography}
Srinivasan Arunachalam and Arkopal Dutt.
\newblock Tomography of quantum states with bounded extent.
\newblock {\em arXiv preprint arXiv:2606.07425}, 2026.

\bibitem[ADL14]{aggarwal2014non}
Divesh Aggarwal, Yevgeniy Dodis, and Shachar Lovett.
\newblock Non-malleable codes from additive combinatorics.
\newblock In {\em Proceedings of the Forty-Sixth Annual ACM Symposium on Theory of Computing}, STOC '14, page 774–783. Association for Computing Machinery, 2014.

\bibitem[AGG{\etalchar{+}}24]{asadi2024quantum}
Vahid~R Asadi, Alexander Golovnev, Tom Gur, Igor Shinkar, and Sathyawageeswar Subramanian.
\newblock Quantum worst-case to average-case reductions for all linear problems.
\newblock In {\em Proceedings of the 2024 Annual ACM-SIAM Symposium on Discrete Algorithms (SODA)}, pages 2535--2567. SIAM, 2024.

\bibitem[AGGS22]{asadi2022worst}
Vahid~R. Asadi, Alexander Golovnev, Tom Gur, and Igor Shinkar.
\newblock Worst-case to average-case reductions via additive combinatorics.
\newblock In {\em Proceedings of the 54th Annual ACM SIGACT Symposium on Theory of Computing}, STOC 2022, page 1566–1574. Association for Computing Machinery, 2022.

\bibitem[BBC{\etalchar{+}}19]{bravyi2019simulation}
Sergey Bravyi, Dan Browne, Padraic Calpin, Earl Campbell, David Gosset, and Mark Howard.
\newblock Simulation of quantum circuits by low-rank stabilizer decompositions.
\newblock {\em Quantum}, 3:181, 2019.

\bibitem[BC26]{briet2026near}
Jop Bri{\"e}t and Davi {Castro-Silva}.
\newblock A near-optimal quadratic {G}oldreich-{L}evin algorithm.
\newblock In {\em Proceedings of the 2026 Annual ACM-SIAM Symposium on Discrete Algorithms (SODA)}, pages 6233--6239. SIAM, 2026.

\bibitem[BDL13]{bhowmick2013new}
Abhishek Bhowmick, Zeev Dvir, and Shachar Lovett.
\newblock New bounds for matching vector families.
\newblock In {\em Proceedings of the Forty-Fifth Annual ACM Symposium on Theory of Computing}, STOC '13, page 823–832. Association for Computing Machinery, 2013.

\bibitem[BFF{\etalchar{+}}01]{batu2001testing}
Tugkan Batu, Eldar Fischer, Lance Fortnow, Ravi Kumar, Ronitt Rubinfeld, and Patrick White.
\newblock Testing random variables for independence and identity.
\newblock In {\em Proceedings 42nd IEEE Symposium on Foundations of Computer Science}, pages 442--451. IEEE, 2001.

\bibitem[BLR14]{ben2014additive}
Eli {Ben-Sasson}, Shachar Lovett, and Noga {Ron-Zewi}.
\newblock An additive combinatorics approach relating rank to communication complexity.
\newblock {\em Journal of the ACM (JACM)}, 61(4):1--18, 2014.

\bibitem[BNOZ25]{bedert2025strong}
Benjamin Bedert, Tamio-Vesa Nakajima, Karolina Okrasa, and Stanislav Zivn\'{y}.
\newblock Strong sparsification for 1-in-3-sat via {P}olynomial {F}reiman-{R}uzsa.
\newblock In {\em 2025 IEEE 66th Annual Symposium on Foundations of Computer Science (FOCS)}, pages 2470--2479, 2025.

\bibitem[BSS16]{bravyi2016trading}
Sergey Bravyi, Graeme Smith, and John~A Smolin.
\newblock Trading classical and quantum computational resources.
\newblock {\em Physical Review X}, 6(2):021043, 2016.

\bibitem[BvDH25]{bao2025tolerant}
Zongbo Bao, Philippe van Dordrecht, and Jonas Helsen.
\newblock Tolerant testing of stabilizer states with a polynomial gap via a generalized uncertainty relation.
\newblock In {\em Proceedings of the 57th Annual ACM Symposium on Theory of Computing}, STOC '25, page 1254–1262. Association for Computing Machinery, 2025.

\bibitem[CBA{\etalchar{+}}26]{castro2026algorithmic}
Davi {Castro-Silva}, Jop Bri{\"e}t, Srinivasan Arunachalam, Arkopal Dutt, and Tom Gur.
\newblock An algorithmic {P}olynomial {F}reiman-{R}uzsa theorem.
\newblock {\em arXiv preprint arXiv:2604.04547}, 2026.

\bibitem[CFMdW10]{chakraborty2010newresultsquantumproperty}
Sourav Chakraborty, Eldar Fischer, Arie Matsliah, and Ronald de~Wolf.
\newblock New results on quantum property testing, 2010.

\bibitem[CGYZ25]{chen2025stabilizer}
Sitan Chen, Weiyuan Gong, Qi~Ye, and Zhihan Zhang.
\newblock Stabilizer bootstrapping: {A} recipe for efficient agnostic tomography and magic estimation.
\newblock In {\em Proceedings of the 57th Annual ACM Symposium on Theory of Computing}, STOC '25, page 429–438, New York, NY, USA, 2025. Association for Computing Machinery.

\bibitem[DJJ{\etalchar{+}}26]{dutt2026learning}
Arkopal Dutt, Dale Jacobs, John Jeang, Saeed Mehraban, and Vladimir Podolskii.
\newblock Learning clifford-structured quantum unitaries and hamiltonians.
\newblock {\em arXiv preprint arXiv:2608.09912}, 2026.

\bibitem[Fel09]{feldman2009distribution}
Vitaly Feldman.
\newblock Distribution-specific agnostic boosting.
\newblock {\em arXiv preprint arXiv:0909.2927}, 2009.

\bibitem[GGMT24]{ggmt2024torsion}
W~Timothy Gowers, Ben Green, Freddie Manners, and Terence Tao.
\newblock Marton's {C}onjecture in abelian groups with bounded torsion.
\newblock {\em arXiv preprint arXiv:2404.02244}, 2024.

\bibitem[GGMT25]{ggmt2025conjecture}
William~Timothy Gowers, Ben Green, Freddie Manners, and Terence Tao.
\newblock On a conjecture of marton.
\newblock {\em Annals of Mathematics}, 201(2):515--549, 2025.

\bibitem[GIKL24]{grewal2023improved}
Sabee Grewal, Vishnu Iyer, William Kretschmer, and Daniel Liang.
\newblock Improved stabilizer estimation via {B}ell difference sampling.
\newblock In {\em Proceedings of the 56th Annual ACM Symposium on Theory of Computing}, STOC 2024, page 1352–1363, New York, NY, USA, 2024. Association for Computing Machinery.

\bibitem[GIKL25]{grewal2023efficient}
Sabee Grewal, Vishnu Iyer, William Kretschmer, and Daniel Liang.
\newblock Efficient {L}earning of {Q}uantum {S}tates {P}repared {W}ith {F}ew {N}on-{C}lifford {G}ates.
\newblock {\em {Quantum}}, 9:1907, November 2025.

\bibitem[GIKL26]{Grewal2026agnostictomography}
Sabee Grewal, Vishnu Iyer, William Kretschmer, and Daniel Liang.
\newblock Agnostic {T}omography of {S}tabilizer {P}roduct {S}tates.
\newblock {\em {Quantum}}, 10:2027, March 2026.

\bibitem[GL89]{goldreich1989hard}
Oded Goldreich and Leonid~A Levin.
\newblock A hard-core predicate for all one-way functions.
\newblock In {\em Proceedings of the twenty-first annual ACM symposium on Theory of computing}, pages 25--32, 1989.

\bibitem[GNW21]{gross2021schur}
David Gross, Sepehr Nezami, and Michael Walter.
\newblock Schur--{W}eyl duality for the {C}lifford group with applications: {P}roperty testing, a robust {H}udson theorem, and de {F}inetti representations.
\newblock {\em Communications in Mathematical Physics}, 385(3):1325--1393, 2021.

\bibitem[GT08]{green2008inverse}
Ben Green and Terence Tao.
\newblock An inverse theorem for the gowers norm.
\newblock {\em Proceedings of the Edinburgh Mathematical Society}, 51(1):73--153, 2008.

\bibitem[GT09]{green2009freiman}
Ben Green and Terence Tao.
\newblock {Freiman}'s theorem in finite fields via extremal set theory.
\newblock {\em Combinatorics, Probability and Computing}, 18(3):335--355, 2009.

\bibitem[HBvD{\etalchar{+}}26]{hinsche2025clifford}
Marcel Hinsche, Zongbo Bao, Philippe van Dordrecht, Jens Eisert, Jop Bri\"{e}t, and Jonas Helsen.
\newblock Clifford testing: Algorithms and lower bounds.
\newblock In {\em Proceedings of the 58th Annual ACM Symposium on Theory of Computing}, STOC '26, page 869–873, New York, NY, USA, 2026. Association for Computing Machinery.

\bibitem[HILL99]{HILL99}
Johan H{\aa}stad, Russell Impagliazzo, Leonid~A. Levin, and Michael Luby.
\newblock A pseudorandom generator from any one-way function.
\newblock {\em SIAM Journal on Computing}, 28(4):1364--1396, 1999.

\bibitem[KM93]{KM93}
Eyal Kushilevitz and Yishay Mansour.
\newblock Learning decision trees using the fourier spectrum.
\newblock {\em SIAM Journal on Computing}, 22(6):1331--1348, 1993.

\bibitem[KSS92]{kearns1992toward}
Michael~J Kearns, Robert~E Schapire, and Linda~M Sellie.
\newblock Toward efficient agnostic learning.
\newblock In {\em Proceedings of the fifth annual workshop on Computational learning theory}, pages 341--352, 1992.

\bibitem[LOH24]{Leone2024learningtdoped}
Lorenzo Leone, Salvatore F.~E. Oliviero, and Alioscia Hamma.
\newblock Learning t-doped stabilizer states.
\newblock {\em {Quantum}}, 8:1361, May 2024.

\bibitem[LS93]{lovaszsimonovits1993}
L{\'a}szl{\'o} Lov{\'a}sz and Mikl{\'o}s Simonovits.
\newblock Random walks in a convex body and an improved volume algorithm.
\newblock {\em Random Structures \& Algorithms}, 4(4):359--412, 1993.

\bibitem[Mon17]{montanaro-bell-sampling}
Ashley Montanaro.
\newblock {Learning stabilizer states by Bell sampling}, 2017.

\bibitem[MT25]{mehraban2024improved}
Saeed Mehraban and Mehrdad Tahmasbi.
\newblock Improved bounds for testing low stabilizer complexity states.
\newblock In {\em Proceedings of the 57th Annual ACM Symposium on Theory of Computing}, STOC '25, page 1222–1233. Association for Computing Machinery, 2025.

\bibitem[Ruz99]{ruzsa1999analog}
Imre Ruzsa.
\newblock An analog of freiman's theorem in groups.
\newblock {\em Ast{\'e}risque}, 258(199):323--326, 1999.

\bibitem[Sam07]{samorodnitsky2007low}
Alex Samorodnitsky.
\newblock Low-degree tests at large distances.
\newblock In {\em Proceedings of the Thirty-Ninth Annual ACM Symposium on Theory of Computing}, STOC '07, page 506–515. Association for Computing Machinery, 2007.

\bibitem[Tao10]{tao2010sumset}
Terence Tao.
\newblock Sumset and inverse sumset theory for shannon entropy.
\newblock {\em Combinatorics, Probability and Computing}, 19(4):603--639, 2010.

\bibitem[Tre04]{TrevisanCoding}
Luca Trevisan.
\newblock Some applications of coding theory in computational complexity.
\newblock {\em Quaderni di Matematica}, 13:347--424, 2004.

\bibitem[TW14]{tulsiani2014quadratic}
Madhur Tulsiani and Julia Wolf.
\newblock Quadratic goldreich--levin theorems.
\newblock {\em SIAM Journal on Computing}, 43(2):730--766, 2014.

\bibitem[ZB11]{zewi2011affine}
Noga Zewi and Eli {Ben-Sasson}.
\newblock From affine to two-source extractors via approximate duality.
\newblock In {\em Proceedings of the Forty-Third Annual ACM Symposium on Theory of Computing}, STOC '11, page 177–186. Association for Computing Machinery, 2011.

\end{thebibliography}

\end{document}